\documentclass[superscriptaddress,amsmath,amssymb,pre]{revtex4-2}
\usepackage[margin=0.8in]{geometry}
\usepackage{algorithm}
\usepackage{algpseudocode}
\usepackage{enumitem}
\usepackage{amsmath, amssymb, amsthm, amsfonts} 
\usepackage{mathtools}
\usepackage{bbm}
\usepackage{natbib}
\newcommand{\M}{\mu}
\newcommand{\C}{\chi}
\newcommand{\N}{{\mathbb N}}
\def\Z{{\cal Z}}
\newcommand{\E}{{\mathbb E}} 
\newcommand{\R}{{\mathbb R}} 
\newcommand{\Ncl}{\mathcal{N}}

\newcommand{\hQ}{\widehat{Q}}

\newcommand{\w}{\omega}

\def\Z{{\cal Z}}
\def\eps{{\epsilon}}

\def\w{{\omega}}

\def\di{{\partial i}}
\def\da{{\partial a}}

\newtheorem{theorem}{Theorem}
\newtheorem{lemma}{Lemma}

\newtheorem{assumption}{Assumption}

\DeclarePairedDelimiterX{\norm}[1]\lVert\rVert{#1}

\usepackage{wrapfig}

\usepackage{hyperref}
\usepackage{verbatim}
\usepackage{mwe}
\usepackage{bm}
\hypersetup{
    colorlinks=true,
    citecolor=blue,
    linkcolor=blue,
    filecolor=magenta,      
    urlcolor=blue,
    pdftitle={Overleaf Example},
    pdfpagemode=FullScreen,
    }
\newcommand{\eqdef}{=\vcentcolon}   
\newcommand{\dd}{\mathop{}\!\mathrm{d}}
\newcommand{\agbrs}[1]{\left\langle #1 \right\rangle}	
\newcommand{\rdbrs}[1]{\left( #1 \right)}				
\newcommand{\sqbrs}[1]{\left\lbrack #1 \right\rbrack}	
\newcommand{\rdbrsv}[2]{\left( #1 \,\middle|\, #2 \right)}				
\begin{document}
\title{Sampling with flows, diffusion and autoregressive neural networks: \texorpdfstring{\\}{}A spin-glass perspective}
\author{Davide Ghio}
\affiliation{IdePHICS laboratory, \'Ecole Polytechnique Fédérale de Lausanne (EPFL), CH-1015 Switzerland}
\author{Yatin Dandi}
\affiliation{IdePHICS laboratory, \'Ecole Polytechnique Fédérale de Lausanne (EPFL), CH-1015 Switzerland}
\affiliation{SPOC laboratory, \'Ecole Polytechnique Fédérale de Lausanne (EPFL), CH-1015 Switzerland}
\author{Florent Krzakala}
\affiliation{IdePHICS laboratory, \'Ecole Polytechnique Fédérale de Lausanne (EPFL), CH-1015 Switzerland}
\author{Lenka Zdeborová}
\affiliation{SPOC laboratory, \'Ecole Polytechnique Fédérale de Lausanne (EPFL), CH-1015 Switzerland}
\begin{abstract}
Recent years witnessed the development of powerful generative models based on flows, diffusion or autoregressive neural networks, achieving remarkable success in generating data from examples with applications in a broad range of areas. 
A theoretical analysis of the performance and understanding of the limitations of these methods remain, however, challenging. In this paper, we undertake a step in this direction by analysing the efficiency of sampling by these methods on a class of problems with a known probability distribution and comparing it with the sampling performance of more  traditional methods such as the Monte Carlo Markov chain and Langevin dynamics.   
We focus on a class of probability distribution widely studied in the statistical physics of disordered systems that relate to spin glasses, statistical inference and constraint satisfaction problems. 

We leverage the fact that sampling via flow-based, diffusion-based or autoregressive networks methods can be equivalently mapped to the analysis of a Bayes optimal denoising of a modified probability measure. Our findings demonstrate that these methods encounter difficulties in sampling stemming from the presence of a first-order phase transition along the algorithm's denoising path. 
Our conclusions go both ways: we identify regions of parameters where these methods are unable to sample efficiently, while that is possible using standard Monte Carlo or Langevin approaches. We also identify regions where the opposite happens: standard approaches are inefficient while the discussed generative methods work well. 
\end{abstract}
\maketitle
\section{Introduction}
The field of machine learning recently witnessed the development of powerful generative models able to produce new data-samples based on learning on datasets of existing samples. Among the most prominent ones achieving recent successes are flow-based models \cite{rezende2015variational,dinh2016density,chen2018neural,albergo2023building,lipman2022flow}, diffusion-based models \cite{sohl2015deep, ho2020denoising,song2019generative,song2020score} and generative autoregressive neural networks \cite{larochelle2011neural,van2016pixel,vaswani2017attention,graves2013generating}. These approaches are achieving remarkable success in diverse areas such as image generation \cite{ramesh2021zero}, language modelling \cite{brown2020language}, generation of molecules \cite{xu2021geodiff,trinquier2021efficient} or theoretical physics \cite{kanwar2020equivariant}. A theoretical understanding of the capabilities of these models, of their limitations and performance, remains, however, a challenge. A major aspect of these techniques is in the learning of the probability measure or its representation from examples. In this paper, we focus instead on a restricted setting where the probability distribution we aim to sample from is known beforehand. Our main goal is to contribute to setting the theoretical understanding of the capabilities and limitations of these powerful generative models. 

When applied to parametric probability distributions, the generative models are indeed designed with the goal to provide samples uniform at random from the distribution. Here, we study whether they are able to {\it efficiently} sample from types of probability distributions that are encountered in the study of mean-field spin glasses and related statistical inference problems \cite{gross1984simplest,mezard1987spin,mezard2009information,krzakala2007gibbs,zdeborova2016statistical,bandeira2018notes}. 

There are several benefits to studying this class of probability distributions. On the one hand, due to numerous studies in the statistical physics of disordered systems, we possess a comparatively good grasp of parameter regions where traditional sampling methods—like Monte Carlo sampling or Langevin dynamics—are effective and where they are not \cite{cugliandolo1993analytical,bouchaud1998out,ben2006cugliandolo,berthier2011theoretical,montanari2006rigorous,antenucci2019glassy}. On the other hand, the tools available for outlining the phase diagrams of these problems \cite{mezard1987spin,mezard2009information,donoho2010message,barbier2019optimal} turn out to be highly effective in analytically describing the performance of generative techniques such as flow-based, diffusion-based, or autoregressive networks as samplers for the respective probability measures.

The above-mentioned tools, along with their mathematically rigorous counterparts, have recently been applied to the analysis and design of sampling algorithms in mean-field spin glass models in the context of stochastic localization \cite{eldan2013thin,eldan2020taming,chen2022localization,el2022sampling,montanari2023posterior,montanari2023sampling}. This was later found to have a close relationship with diffusion-based models \cite{montanari2023sampling,montanari2023posterior}. In particular, \cite{el2022sampling,montanari2023posterior} showed how one can turn the message-passing denoising algorithms into samplers, a technique we shall use as well in the present study.
While we could not find similar studies for flow-based models, closely related work on autoregressive networks exists on the decimation of message-passing algorithms used for finding solutions of the random K-satisfiability problem in \cite{montanari2007solving,ricci2009cavity}. Here, we build on these works and bring the following contributions:
\begin{itemize}
\item Using the formalism of stochastic interpolants \cite{albergo2023building,albergo2023stochastic}, we analyse sampling with flow-based methods, which leads to Bayesian denoising with an additive white  Gaussian noise (AWGN). This turns out to be equivalent to what arises in the analysis of diffusion-based sampling and stochastic localization derived in \cite{el2022information,el2022sampling,montanari2023sampling}. 
\item In the case of autoregressive generative networks the analysis leads instead to a Bayesian denoising problem now correcting erased variables (the Binary Erasure Channel \cite{richardson2008modern}) as in \cite{ricci2009cavity}. 
\item Focusing on prototypical exactly solvable models where one can perform asymptotic analysis, we then study the phase diagrams of the corresponding Bayesian denoising problems as a function of the strength of the signal-to-noise ratio. 
\item We investigate the feasibility of sampling using these techniques. We posit that these methods are capable of efficient sampling, provided there is no range in noise amplitude exhibiting the metastability characteristic of first-order phase transitions. If such metastability exists, the denoising becomes computationally hard. 
\item We locate these metastable regions in prototypical models, specifically the Ising and spherical $p$-spin models, the bicoloring of random hypergraphs problem, and a matrix estimation problem with a sparse spike.
\item We also provide a GitHub repository with our numerical experiments in \cite{GITHUB}, \\ see \href{https://github.com/IdePHICS/DiffSamp}{https://github.com/IdePHICS/DiffSamp}.
\end{itemize}

\begin{figure}[t]
\centering \includegraphics[width=0.6\textwidth]{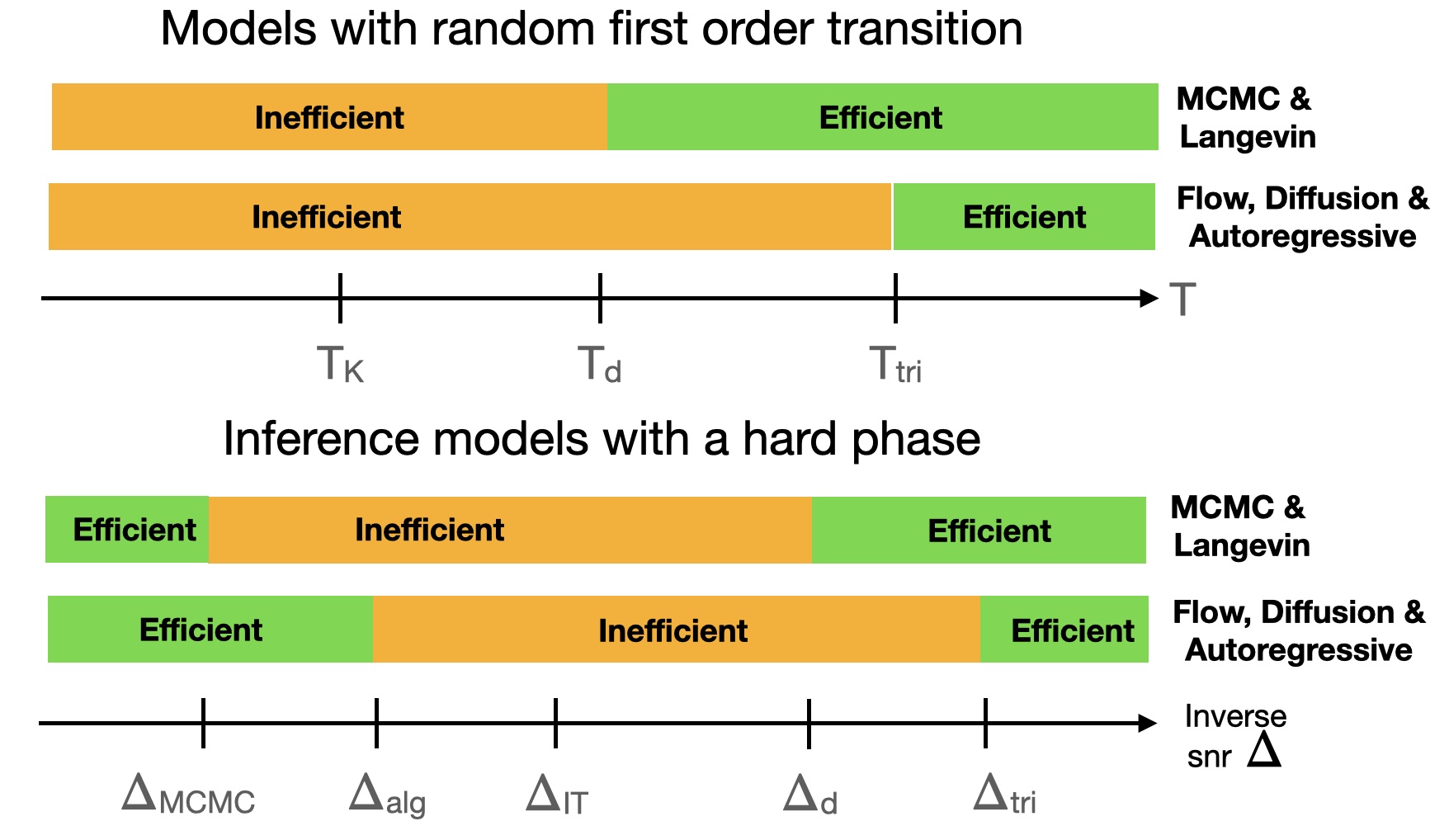}
    \caption{Schematic summary of the comparison of the efficiency of sampling with flow-based, diffusion-based and autoregressive methods versus Langevin or Monte-Carlo approaches in spin glass models with a random first-order transition (top), and in statistical inference models with a computationally hard phase (bottom). The computationally hard phase in inference problems appears for For $\Delta_{\rm alg} < \Delta < \Delta_{\rm IT}$ where efficient algorithms achieving close-to-optimal estimation error are not known and conjectured not to exist \cite{gamarnik2022disordered}. }
    \label{fig:mainfig}
\end{figure}
In terms of comparison between the flow-based, diffusion-based or autoregressive generative models with more traditional Monte Carlo Markov chains (MCMC) and Langevin sampling procedures, the following picture emerges (summarized in Fig.~\ref{fig:mainfig}) for two classes of models.

\begin{itemize}
\item Disordered models that exhibit a phase diagram of the random-first-order-theory (RFOT) type, also often called discontinuous one-step replica symmetry breaking \cite{gross1984simplest,biroli2012random}, are typical in the mean-field theory of glass transition \cite{berthier2011theoretical}, but they also appear in a variety of random constraint satisfaction problems \cite{krzakala2007gibbs}. For such models, there exists a so-called dynamical temperature $T_d$ such that Monte Carlo Markov chains or Langevin algorithms are predicted to sample efficiently for $T>T_d$ \cite{cugliandolo1993analytical,bouchaud1998out,ben2006cugliandolo,berthier2011theoretical,montanari2006rigorous} while the sampling problem is computationally hard for $T<T_d$. 
Recent work \cite{ciarella2023machine} showed empirically that the region $T<T_d$ is also hard when sampling with autoregressive neural networks.
Our analysis of the currently used flow-based, diffusion-based and autoregressive networks-based sampling procedures goes further and reveals analytically that they perform {\it worse} than the traditional techniques and that there is another temperature $T_{\rm tri}$, that depends on the detail of the method, such that for $T_d < T < T_{\rm tri}$ these generative methods do not sample efficiently while traditional procedures do. 
\item A second class of problems are the noisy random statistical inference problems with a hard phase and a statistical-to-computational gap \cite{zdeborova2016statistical,bandeira2018notes,bandeira2022franz,gamarnik2022disordered}. In such problems, a good statistical inference can only be performed below a critical value of the noise (or inverse signal-to-noise ratio) $\Delta_{\rm IT}$, at least with infinite computational power. It turns out, however, that the best algorithms we know for such problems (in particular message-passing ones) are only able to work below a second threshold $\Delta_{\rm alg}$ and fail to learn in the "hard" region $\Delta_{\rm alg} < \Delta < \Delta_{\rm IT}$. Existing literature predicts an even {\it more significant gap} to exist in those cases for MCMC or Langevin-based samplers \cite{antenucci2019glassy,mannelli2020marvels,angelini2023limits} down to a noise amplitude $\Delta_{\rm MCMC} < \Delta_{\rm alg}$. In this setting, however, it appears that flow-based, diffusion-based, and autoregressive network-based methods outperform the standard approaches and sample efficiently as soon as $\Delta<\Delta_{\rm alg}$.

\end{itemize}

\subsection{Further related works}
The present study of sampling algorithms owes a lot to recent lines of work. In particular, we shall follow the helpful approach of continuous-time flows based on stochastic interpolants \cite{albergo2023building,albergo2023stochastic} as the starting point. Additionally, our way of setting up the statistical physics equivalent problem for a continuous-time flow method closely follows the one pioneered recently in \cite{el2022information,el2022sampling} for stochastic localization and generalized in \cite{montanari2023posterior} for diffusion models. 

The difficulty in denoising we uncover turns out to be connected to an easy-hard transition arising in statistical inference problems \cite{zdeborova2016statistical} often called the statistical-to-computational gap \cite{bandeira2018notes,bandeira2022franz}. For these models, mean-field algorithms from spin glass theory, such as the Belief-Propagation (BP) equations \cite{mezard2009information} and the Approximate-Message-Passing (AMP) \cite{donoho2010message}, or the equivalent Thouless-Anderson-Palmer (TAP) \cite{thouless1977solution} approach, are believed to be among the most efficient one \cite{celentano2020estimation}. Using AMP for sampling via diffusion method (or via stochastic localization) was discussed recently in \cite{el2022sampling,montanari2023posterior}. The resulting algorithm turns out to have very close connections to the so-called reinforcement BP \cite{dall2008entropy}, and a direct connection exists between the analysis in the present paper and the reinforced and decimated version of message passing algorithms \cite{mezard2002analytic,chavas2005survey,biazzo2012performance,montanari2007solving}. 

Finally, we note that since we assume that the model is known, we are not discussing here the limitations of training denoisers models from a limited number of data points. This difficulty was discussed recently in \cite{biroli2023generative} in the context of statistical physics models.

\section{\label{sec:DampDiff}Sampling with flow and diffusion-based models}
We start by discussing continuous-time flow-based models \cite{chen2018neural,lipman2022flow,albergo2023building}. 
The goal is to sample a vector ${\bf x}_0 \in {\mathbb R}^N$ from a so-called target distribution $P_0({\bf x}_0)$. Continuous-time flow-based models \cite{chen2018neural, lipman2022flow,liu2022flow,albergo2023building} achieve this by constructing an ODE (flow) that continuously transforms Gaussian noise ${\bf z}\sim \mathcal{N}({\bf 0},\mathbb{I}_N)$ to a sample ${\bf x}_0\sim P_0$. One way to obtain such a flow is by inverting the time-evolving law of stochastic or deterministic processes bridging from a sample ${\bf x}_0\sim P_0$ to a Gaussian noise ${\bf z}\sim \mathcal{N}({\bf 0},\mathbb{I}_N)$. Such a continous-time reverse process underlies a wide range of flow-models \cite{lipman2022flow,albergo2023building} as well as the score-based diffusion models \cite{song2020score}. In practice, we would only observe samples from the distribution $P_0({\bf x}_0)$, whereas in this paper we aim to focus on the limitations of these procedures, we will thus assume that the distribution $P_0({\bf x}_0)$ is known as this can only render the task easier than when only samples from $P_0({\bf x}_0)$ are available. 

To present this method clearly, we find it convenient to use the formalism of stochastic linear interpolants introduced in \cite{albergo2023building, albergo2023stochastic}. Starting from a vector ${\bf x}_0$ sampled from $P_0$ and a Gaussian vector ${\bf z}$, we consider the process ---or \textit{one-sided stochastic interpolant}--- defined by
\begin{equation}\label{eq:x(t)}
    {\bf y}(t) = \alpha(t) {\bf x}_0 + \beta(t) {\bf z}\\, \quad t\in[0,1]\,,
\end{equation}
where the functions $\alpha$ and $\beta$ are generic, but constrained by the following relations:
\begin{eqnarray}
    &\alpha(0)=\beta(1)=1;\quad \alpha(1)=\beta(0)=0; \\
    &\forall t\in[0,1] : \alpha(t)\geq 0,\; \Dot{\alpha}(t)\leq 0, \;\beta(t)\geq 0,\; \Dot{\beta}(t)\geq 0
\end{eqnarray}
such that indeed we have ${\bf y}(0)={\bf x}_0$ and ${\bf y}(1)={\bf z}$. 
Suppose that $P_0({\bf x_0})$ admits a density $\rho_0$ w.r.t the Lebesgue measure.
It can be shown \citep{albergo2023stochastic}, see also Appendix~\ref{sec:AnalysisAlgo}, that the probability density $\rho({\bf y}(t))$ associated to the measure $P_t$ of the random variable ${\bf y}(t)$ satisfies the following transport equation:
\begin{equation}\label{eq:transport}
    \partial_t\rho({\bf y},t) + \nabla \cdot (b({\bf y},t)\rho({\bf y},t)) = 0\,,
\end{equation}
where we defined the \textit{velocity field}
\begin{equation}\label{eq:v_field}
    b({\bf y},t) = \mathbb{E}[\partial_t {\bf y}(t)|{\bf y}(t)={\bf y}] = \mathbb{E}[\dot{\alpha}(t) {\bf x_0 }+ \dot{\beta}(t)z|{\bf y}(t)= {\bf y}]\,.
\end{equation}
Indeed, $b({\bf y},t)$ is simply the expected velocity of ${\bf y}(t)$ conditioned on being at ${\bf y}$ at time $t$. In Appendix~\ref{sec:AnalysisAlgo}, we also provide a formal definition of Eq.~(\ref{eq:v_field}) based on the analysis in \cite{albergo2023building}, along with a discussion of the case when the measure $P_0({\bf x_0})$ is discrete.

Equation~(\ref{eq:x(t)}) defines a forward process interpolating from ${\bf x_0 }$ to~${\bf z}$. Equation~(\ref{eq:transport}) further reveals that, in law, ${\bf y}(t)$ can be obtained by applying the vector field (Eq.~(\ref{eq:v_field})) starting from ${\bf x_0}$ at time $t=0$. The algorithm proposed by \cite{albergo2023building} relies on applying the velocity field in the reverse direction of time. Concretely, it relies on approximating the unique solution to the following ordinary differential equation starting from a random Gaussian initial condition from $t=1$, back to $t=0$:
\begin{equation}\label{eq:ODE}
    \frac{d {\bf Y}(t)}{dt} = 
    b({\bf Y}(t),t)\,.
\end{equation}
Eq.~(\ref{eq:transport}) then implies that the random variable defined by Eq.~(\ref{eq:ODE}) solved backwards in time from the final value ${\bf Y}_{t=1}={\bf z}\sim\mathcal{N}({\bf 0},\mathbb{I}_N)$ is also distributed according to $\rho({\bf y}(t))$ (which is nothing but the continuity equation for the flow defined by Eq.~(\ref{eq:ODE})) and, in particular, is distributed as the desired target $P_0({\bf x}_0)$ at time $t=0$.  If we can numerically solve this ODE, then we have a sampling algorithm for~$P_0(x)$. 

In order to do that, we discretize Eq.~(\ref{eq:ODE}) using the forward Euler method (in reverse time) to write
\begin{equation}\label{eq:Euler}
    {\bf Y}_{t-\delta} = {\bf Y}_t - \delta\,b({\bf Y}_{t},t)\,.
\end{equation}
Noticing that we can rewrite the vector field using Eq.~(\ref{eq:x(t)}) as
\begin{eqnarray}\label{eq:vecField}
    b({\bf Y}_t,t) \!\!\!&=& \!\!\!\mathbb{E}\left[\Dot{\alpha}(t) {\bf x}_0 + \frac{\Dot{\beta}(t)}{\beta(t)}\beta(t) {\bf z}\bigg|{\bf y}(t)={\bf Y}_t\right] \nonumber \\
   \!\!\! &=& \!\!\!\mathbb{E}\left[\Dot{\alpha}(t){\bf x}_0 + \frac{\Dot{\beta}(t)}{\beta(t)}\left( {\bf y}(t) - \alpha(t){\bf x}_0\right)\bigg|{\bf y}(t)={\bf Y}_t\right] \nonumber\\
   \!\!\! &=& \!\!\! \frac{\Dot{\beta}(t)}{\beta(t)}{\bf Y}_t + \left( \Dot{\alpha}(t) - \frac{\Dot{\beta}(t)\alpha(t)}{\beta(t)}\right) \mathbb{E}\left[{\bf x}_0 |{\bf y}(t)={\bf Y}_t\right]  \, ,
\end{eqnarray}
we put back Eq.~(\ref{eq:vecField}) in Eq.~(\ref{eq:Euler}) to reach
\begin{equation}\label{discrODE}
    {\bf Y}_{t-\delta} \!=\! \left(\! 1-\delta \frac{\Dot{\beta}(t)}{\beta(t)}\!\right)\!{\bf Y}_t - \delta\left(\! \Dot{\alpha}(t) - \frac{\Dot{\beta}(t)\alpha(t)}{\beta(t)}\!\right) \mathbb{E}\left[{\bf x}_0 |{\bf y}(t)={\bf Y}_t\right]
\end{equation}
which, given the initial condition ${\bf Y}_{t=1}={\bf z}$, will evolve back to a sample from the target ${\bf Y}_{t=0}\sim P_0$, provided that we are able to estimate $\mathbb{E}\left[{\bf x}_0 |{\bf y}(t)={\bf Y}_t\right]$ \textit{well} at each time $t$. In Algorithm~\ref{Algo} we report a schematic implementation of the resulting flow-based sampling technique. 
In Appendix~\ref{sec:AnalysisAlgo}, we further discuss the effect of discretization and the associated sampling guarantees under access to a perfect denoiser.

\subsection{Diffusion-based models and SDEs}
The flow-based algorithm \citep{albergo2023building} (Alg.~\ref{Algo}) relies on replicating the law of the interpolant ${\bf y}(t)$ defined by Eq.~(\ref{eq:x(t)}) through the deterministic ordinary differential equation (ODE) defined by Eq.~(\ref{eq:ODE}). Interestingly, however, the {\it same} law for ${\bf y}(t)$ can be obtained via a stochastic differential equation (SDE). 
Similarly, when ${\bf y}(t)$ is generated through a forward diffusion process applied to the data, the law can be generated through an associated time-reversed SDE (and an ODE) \cite{anderson1982reverse,song2020score}. It turns out that both flow-based (ODE) methods and diffusion-based (SDE) ones associated to the above processes (as well as stochastic localization \cite{eldan2013thin,eldan2020taming,chen2022localization}) require estimating the posterior-mean $\mathbb{E}\left[{\bf x}_0 |{\bf y}(t)={\bf Y}_t\right]$, i.e. averages w.r.t a ``tilted" probability measure. 
Utilization of such an evolving tilted measure for sampling was done first 
in the stochastic-localization-based sampling algorithm in \cite{el2022sampling, montanari2023posterior} and related to generic diffusion models in \cite{montanari2023sampling}.
The reliance of both the ODE and SDE-based approaches on the same tilted measure is a consequence of the posterior mean (or equivalently the score) yielding both the velocity field for the ODE defined by Eq.~(\ref{eq:ODE}) and the drift term for the equivalent SDE as noted in \cite{albergo2023building,albergo2023stochastic,song2020score}. Since the law of ${\bf y}(t)$ remains the same with the ODE or SDE-based approaches, the deterministic evolution of the ``tilted" measure $P({\bf x}|{\bf y}(t)={\bf Y}_t)$ matches in law with that of the associated stochastic localization process \citep{montanari2023sampling}. This is convenient for our analysis since this means that {\it all} the phase transitions discussed in our work apply to flow-based methods, diffusion-based ones and stochastic localization approaches.

\begin{algorithm}[H]
\caption{Flow-based sampling algorithm }\label{Algo}

\begin{algorithmic}
\State \textbf{Input:} \text{Denoiser}, \text{parameters:} $\delta,\alpha(t),\beta(t)$
\State ${\bf Y}_{t=1} = {\bf Z} \sim \mathcal{N}({\bf 0},\mathbb{I}_N)$, $N_{\rm steps} \equiv \lfloor1/\delta\rfloor$
\For{$l=0, 1,\dots,N_{\rm steps}-1$ }

    $t = 1 - l\delta$\,;\;$\alpha=\alpha(t)$\,;\;$\beta=\beta(t)$\,;\;$\gamma=\alpha/\beta$
    
    Compute $\hat {\bf x}_0 = \mathbb{E}\left[{\bf x}_0 |{\bf y}(t)={\bf Y}_t\right]$ using the denoiser\, \\ \hspace*{\algorithmicindent} 
    (assuming a channel 
        ${\bf Y}_t = \alpha {\bf X}_0 + \beta {\bf Z}$)

    Update the field $\boldsymbol{Y}_{t-\delta}$ via Eq.~(\ref{discrODE})
\EndFor
\State \textbf{Return} $\hat{\boldsymbol{x}}_0$ 
\end{algorithmic}
\end{algorithm}

\subsection{The posterior average and Bayes-optimal denoising}

The key difficulty is thus to estimate the average $\mathbb{E}\left[{\bf x}_0 |{\bf y}(t)= {\bf Y}_t\right]$. This is nothing but the Bayes-optimal denoising where ${\bf Y}_t$ is a noisy observation of~${\bf x}_0$ corrupted by additive white Gaussian noise in the form of Eq.~(\ref{eq:x(t)}). 

The Bayes-optimal denoising estimator $\hat {\bf x}$ (optimal in the sense of minimizing the mean-squared-error between the estimator and ${\bf x}_0$) is to use the posterior mean $\hat {\bf x} = \mathbb{E}\left[{\bf x}_0 |{\bf y}(t)={\bf Y}_t\right]$~\cite{cover1991entropy}. Determining this average can be computationally hard, as it involves an integral in high-dimension. In practical machine learning setups, one learns the denoiser using a neural network trained on data coming from the distribution. Thus, the main narrative of the diffusion and flow-based methods is that if one can access to a good denoiser, then one can turn it into a sampler.

In this paper, we aim to focus on the limitations of this procedure and we shall hence assume that the task of obtaining a denoiser has been completed to near perfection by focusing on target distributions $P_0({\bf x}_0)$ stemming from problems studied in statistical physics for which the Bayes-optimal denoising problem including its algorithmic feasibility has been extensively investigated. Before turning to these models, we write the denoising problem in terms that are more familiar in the statistical physics but also in the information-theoretic context. 

Let us consider again the process defined by Eq.~(\ref{eq:x(t)}). Using the Bayes theorem, we can see that the posterior distribution of the sample conditioned on the observation ${\bf y}(t)={\bf Y}_t$ is given by
\begin{eqnarray}
&&P({\bf x}|{\bf y}(t)={\bf Y}_t) =\frac{P_0({\bf x})P({\bf y}(t)={\bf Y}_t|{\bf x})}{P({\bf y}(t)={\bf Y}_t)} \label{tilted-eq1}
\\ 
&=& \frac{1}{Z({\bf Y}_t)}\exp \left( \frac{\alpha(t)}{\beta(t)^2}\langle {\bf Y}_t,{\bf x} \rangle - \frac{\alpha(t)^2}{2\beta(t)^2}||{\bf x}||^2\right)P_0({\bf x})\,. \nonumber
\end{eqnarray}
where in the last line we put all the terms not depending on ${\bf x}$ inside the normalization aka partition function $Z$. 

We now recall that the law of ${\bf Y}_t$ can be obtained through an observation of a sample from the target distribution ${\bf x}_0 \sim P_0$ through the AWGN channel rescaled by the factors $\alpha(t)$ and $\beta(t)$ as in Eq.~(\ref{eq:x(t)}). 

Here we denote ${\bf x_0}$ as the ``truth'' signal and make the difference between the dummy variable ${\bf x}$ in the integral. We can hence rewrite the measure in Eq.~(\ref{tilted-eq1}) further as

\begin{eqnarray}\label{eq:P_gamma}
    P_\gamma({\bf x}|{\bf x}_0,{\bf z}) \!\!\!\! &\propto& \! \!\!\!  e^{\gamma(t)^2\langle {\bf x}, {\bf x}_0 \rangle + \gamma(t)\langle z,{\bf x} \rangle - \frac{\gamma(t)^2}{2}\|{\bf x}\||^2} P_0({\bf x})\,, 
    \label{tilted-measure}
\end{eqnarray}

where $z \sim \mathcal{N}({\bf 0},\mathbb{I}_N)$ and where we have defined the effective signal-to-noise ratio 
\begin{equation}   \gamma(t)=\frac{\alpha(t)}{\beta(t)}.
\end{equation}
We see that, by construction, $\gamma$ will be such that $\gamma(1)=0$ and $\gamma(0)=+\infty$. 
We will refer to $P_\gamma$ as the {\it tilted} measure, and it will be a central object in the remaining discussion. 
We have added the index $\gamma$ in the notation of the distribution $P_\gamma$ to distinguish it from other considered probability distributions from now on.  

\section{Autoregressive networks and ancestral sampling}
Let us now discuss the second type of sampling algorithm considered in this paper -- the autoregressive networks. A classical way of sampling a vector ${\bf x} \in {\mathbb R}^N$ from a target distribution $P_0({\bf x})$, in computer science and statistics is sequential sampling, or ancestral sampling: One starts by computing the marginal probability of the first coordinate $P_0(x_1)$, and samples $x_1$ accordingly. Subsequently, the algorithm generates all the coordinates for a sample, by sampling each coordinate conditioned on its ``parent" coordinates \cite{bishop2006pattern}. In the simplest case, the distribution of each coordinate is assumed to depend on all previous coordinates. Thus, after sampling $x_1$, for each subsequent node, one looks at the distribution $P_0({\bf x}|x_1)$, one considers the marginal distribution of $x_2$ etc. Of course marginalization of a high-dimensional probability distribution is in general hard, and the strategy used in an autoregressive networks \cite{gregor2014deep} is to directly {\it learn from data} a probability distribution written in the (autoregressive) form $P_0({\bf x}) = P_0(x_1)P_0(x_2|x_1) \dots P_0(x_N|x_1,\dots,x_{N-1})$ with each term being represented via a neural network.
While this decomposition works for any ordering of the components, in practical applications the order may be relevant (this is an important point in ancestral sampling). In the present paper, we will consider the order to be random. 

We showed in the previous section that sampling through diffusion in our formalism boils down to the performance of a sequence of denoising problems, where in information theory terms the signal is observed through an additive white Gaussian noise channel. Analogously to that, we can interpret the autoregressive networks, or its ideal version sequential sampling, as the estimation of the marginal when a fraction $\theta$ of entries of one configuration sampled uniformly at random from the target $P_0$ is revealed exactly. To come back to the sampling scheme, reversing this process is equivalent to fixing one variable at a time from the marginal conditioned to the variables previously fixed until all variables are fixed. In statistical physics, this procedure has been studied and analysed under the name decimation \cite{ricci2009cavity}. In Algorithm~\ref{Algo-2} we report a schematic implementation of the autoregressive-based sampling technique.

In information-theoretic terms, this corresponds to a denoising problem under the so-called binary erasure channel (BEC) in which a transmitter sends a bit and the receiver either receives the bit correctly or with some probability $1-\theta$ receives a message that the bit was not received but ``erased'' instead. 

Repeating the way of thinking we used for sampling with diffusion, we now want to analyse the denoising for such models and write the modified measure 
\begin{equation}
    P_{\theta}({\bf x}| {\bf x}_0, S_\theta) \propto  P_0({\bf x}) \prod_{i \in S_\theta}     \delta([{x}_i-[{\bf x}_0]_{i}) 
    \label{pinning-measure}
\end{equation}
where $S_\theta$ is the set of revealed variables of size $\theta N$. We can also think of these variables as pinned to their ground truth value ${\bf x}_0$ and consequently call the measure $P_{\theta}({\bf x}| {\bf x}_0, S_\theta)$ the pinning measure. From a statistical physics point of view, the new measure can again be thought of as the original one, but with an "infinite" magnetic field pointing to the direction of the components of ${\bf x}_0$ for an expected fraction $\theta$ of components. We are again interested in the marginals of $P_\theta$ that we will denote $\hat {\rm x}(\theta)$. The evolution of the pinning measure can also be interpreted as a coordinate-by-coordinate stochastic localization process \cite{chen2022localization}.
\begin{algorithm}[H]
\caption{Autoregressive-network-based sampling algorithm}\label{Algo-2}
\begin{algorithmic}
\State \textbf{Input:} \text{BEC denoiser}
\For{$l=1,\dots,N$ }

\State Compute the posterior marginal $P_l(x_l|x_1,\dots,x_{l-1})$ with BEC denoiser Eq.~(\ref{BEC:model})

\State Assign $x_l \sim P_l(x_l|x_1,\dots,x_{l-1})$
\EndFor
\State \textbf{Return} $\boldsymbol{x}$ 
\end{algorithmic}
\end{algorithm}
\section{Properties of Bayesian-optimal denoising}
In both the cases -- diffusion and autoregressive -- we see that one has to perform an optimal denoising based on an observation ${\bf y}$. In the case of diffusion, we have access to an observable where a sample from the target ${\bf x}_0 \sim P_0$ is polluted by an AWGN channel:
\begin{equation}
    {\bf y} = \alpha {\bf x}_0 + \beta {\bf z}\,,
\end{equation}
with $z \sim \mathcal{N}({\bf 0},\mathbb{I}_N)$, while for autoregressive, it is polluted by the BEC channel where for every $i = 1,\dots, N$ independently
\begin{equation}
    {y}_i =  
    \begin{cases}
        [{\bf x}_0]_{i} &\text{with probability}\, \,  \theta \\
        * &\text{otherwise}
    \end{cases}
    \label{BEC:model}
\end{equation}
Our goal is now to study the properties of such channels and the corresponding Bayes-optimal denoisers. A crucial point is that these Bayes-optimal estimation problems lead to the {\it Nishimori identities} and the {\it single state} (or replica symmetric in the spin glass jargon) properties of the measures $P_\gamma$, Eq.~(\ref{tilted-measure}), and $P_\theta$, Eq.~(\ref{pinning-measure}), see e.g. the review \cite{zdeborova2016statistical}. While these are classical properties, we remind their derivations in the appendix, and state informally their form and main consequences here.

Concretely, we shall be interested in the evolution in time (or equivalently in $\gamma \in [0,\infty[$ (AWGN) and $\theta \in [0,1]$ (BEC)) of the following {\it order parameters}, or {\it overlaps}:
\begin{eqnarray}
\label{eq:Mu}
\M(\gamma) &\equiv& \frac{1}{N}\mathbb{E}[ \hat {\bf x}(\gamma) \cdot {\bf x}_0]\,,\\ 
\label{eq:Chi}
\C(\gamma) &\equiv& \frac{1}{N}\mathbb{E}[\norm{ \hat {\bf x}(\gamma)}^2]\,.
\end{eqnarray}

The same definitions hold if we replace $\gamma$ by $\theta$, i.e. for the autoregressive process instead of the diffusive. The expectations are over the disorder ${\bf x}_0$ and ${\bf z}$. 
The single state property relates to the self-averaging of these order parameters: Almost anywhere in $\gamma$ (AWGN) or $\theta$ (BEC), the equilibrium measure $P_\gamma$ and $P_\theta$ is a single "phase" probability measure, where overlap and order parameter are self-averaging, i.e. concentrate to a single scalar quantity as $N \to \infty$. Note that the ``almost anywhere" is not a void statement: for instance, if one is sampling a complicated glassy problem such as, say, the Sherrington-Kirkpatrick model \cite{mezard1987spin}, strictly at $\gamma=0$, the Boltzmann measure is complex and overlaps are not self-averaging. For almost all values of $\gamma$ (that is, except at some critical threshold values) the measure then corresponds to a much simpler single-phase problem, and in particular, there is no phenomenon akin to a static glassy transition or replica-symmetry-breaking, thus the replica symmetric assumption is sufficient to describe the thermodynamic behaviour of these measures. This will be instrumental in the exact asymptotic analysis.

While such results have a long history \cite{nishimori2001statistical,iba1999nishimori}, their proof in the present context follows from the analysis of optimal Bayesian denoising, and in particular the I-MMSE theorem and the fluctuation-dissipation one \cite{guo2005mutual,korada2009exact,barbier2019optimal,barbier2021overlap}, for the AWGN channel, and by the study of the so-called pinning lemma \cite{abbe2013conditional,coja2017information} in the BEC channel, that also has roots in the study of the glass transition \cite{montanari2006rigorous,biroli2008thermodynamic}.

A second property, often called the Nishimori symmetry in physics \cite{iba1999nishimori,zdeborova2016statistical}, follows from Bayes theorem and states that for optimal Bayesian denoising one has $\M(\gamma) =\C(\gamma) $ in the diffusion setting and $\M(\theta)=\C(\theta)$ in the autoregressive one.

Note that while we cannot experimentally compute $\M(\gamma)$ in simulations, we can instead compute numerically $\C(\gamma)$ which is just the norm of the estimator. The study of phase transitions in such problems is thus reduced to the behaviour of a single scalar quantity. 
\section{Prototypical exactly analysable models}
We will now analyse the properties of the tilted measure $P_\gamma$ and the pinning measure $P_\theta$ for several concrete cases of the target measure $P_0$. This will be possible exactly in the thermodynamic limit $N\to \infty$.
 We shall focus on several classical problems from spin glass theory, statistical inference, and constraint optimization problems, but analogous analysis of the tilted and pinning measures can be done for many other problems for which the phase diagram was obtained via the replica or the cavity method \cite{mezard1987spin,mezard2009information}.

First, we shall start by studying models that present a so-called random first-order transition \cite{kirkpatrick2015colloquium,berthier2011theoretical}, or in replica parlance, a discontinuous one-step replica symmetry breaking phenomenology \cite{mezard1987spin}. There are two crucial temperatures in RFOT systems. The so-called Kauzmann temperature $T_K$ below which the system behaves as an ideal glass and the so-called dynamical temperature $T_d$ that is defined as the temperature at which the point-to-set correlation length diverges and consequently Monte Carlo or Langevin dynamics equilibration time diverges as well \cite{berthier2011theoretical,montanari2006rigorous}. 
It is a widely accepted conjecture that no efficient method 
can sample such models in their low temperature $T< T_d$ glassy phase (and in fact, one can prove such hardness for classes of algorithms \cite{el2022sampling,alaoui2023shattering}). We shall focus on the ``paramagnetic" phases of these models $T>T_d$ for which the Monte Carlo or Langevin equilibration time is known to be finite and hence efficient sampling with MCMC or Langevin is possible. Specifically, we shall consider the following models: 
\begin{itemize}
      \item {\bf The Ising and spherical $p$-spin glass}~\cite{gross1984simplest}, whose Hamiltonian reads (for $p=3$):
    \begin{align}
    \label{def:pspin}
    {\cal H}({\bf x}) = -\frac{\sqrt{3}}{N} \sum_{i<j<k} J_{ijk}x_ix_jx_k \,,
    \end{align}
    
    with $J_{ijk} \sim\Ncl(0,1)$. The Boltzmann distribution is then $P_0({\bf x}) \propto \exp(-\beta{\cal H}({\bf x}))$ with $\beta=1/T$. In the Ising case, we take $x_i=\pm 1$, for $i=1,\dots,N$, in the spherical case ${\bf x} \in {\cal S}^{N-1}$ (even though, since we are discussing the high-temperature phase, we shall use the equivalent Gaussian model). This is one of the most studied models in spin glass theory. 
\item {\bf NAE-SAT, or bicoloring}: Another class of popular models arises in the context of constraint satisfaction problems, e.g. random satisfiability problem or random graph coloring \cite{krzakala2007gibbs}. Here we shall focus on a prototypical case from this class, the problem of coloring random $k$-hypergraphs with two colours. This model was studied using statistical physics techniques in \cite{Castellani2003,Gabrié_2017,ricci2019typology}. Numerous rigorous results for this model were also established e.g. in \cite{Ding2016,Coja-Oghlan2012}. The probability distribution is the following:
    \begin{equation}
    P_0({\bf x}) \propto\prod_{a=1}^{M}\omega({\bf x}_{\partial a})\,, \quad {\bf x} \in \{-1,+1\}^N
    \end{equation}
    \begin{equation}
        \alpha = \frac{M}{N}\,,\quad \omega(x_1,\dots,x_k) = 
    \begin{cases}
    0 & \text{if} \;\sum_{i=1}^k x_i = \pm k\\
    1 & \text{otherwise}
    \end{cases}
    \nonumber
    \end{equation}
    where by ${\bf x}_{\partial a}$ we refer to the group of $k$ variables entering into the clause $a$. Again, this model presents a RFOT phenomenology, with a dynamical/clustering transition $\alpha_d$. The literature supports the property that MCMC is able to sample efficiently for $\alpha<\alpha_d$ and is not for $\alpha > \alpha_d$ \cite{krzakala2007gibbs,montanari2006rigorous,krzakala2009hiding}.
\end{itemize}

In these models, by studying the tilted/pinning measures, we shall see the limitations of flow-based, diffusion and autoregressive sampling methods that will {\it fail} at temperature $T_d < T< T_{\rm tri}$ and constraint density $\alpha_{\rm tri} < \alpha< \alpha_d$. 
These methods thus turn out to be not as performing as MCMC/Langevin approaches above the dynamical transition. We expect this to be the case for any model with RFOT phenomenology. Every cloud, however, has its silver lining, and we shall also note that there is a class of models where flow, diffusion or autoregressive approaches outperform MCMC/Langevin. This will be the case in statistical inference problems presenting a hard phase, i.e. presenting a sharply defined region of the noise $\Delta_{\rm IT} > \Delta >\Delta_{\rm alg}$ that is computationally hard for message-passing algorithms and conjectured hard for any other efficient algorithm \cite{gamarnik2022disordered}. A recent line of work \cite{antenucci2019glassy,mannelli2020marvels,sarao2020complex,angelini2023limits,chen2023almost}, argues that when it comes to sampling algorithms such as Langevin or other algorithms walking in the space of configurations, such as gradient descent, then the hardness actually extends to even beyond the threshold $\Delta_{\rm alg}$ up to some $\Delta_{\rm MCMC}$ that depends on the specific algorithm. Yet, diffusion models were shown to be able to sample down to $\Delta_{\rm alg}$ \cite{montanari2023posterior}. We will also illustrate this here by studying the tilted and pinning measures of the following prototypical model: 
\begin{itemize}
  \item {\bf Sparse rank-one matrix factorization}:
    The sparse \textit{spiked Wigner} model and its phase diagram were discussed e.g. in \cite{deshpande2014information,lesieur2015phase,lesieur2017constrained,dia2016mutual}. 
    This model is a variation of the ``planted" Sherrington-Kirkpatrick model in spin glass physics. In such models, one is given a ``spiked" version of a random symmetric matrix with a rank-one perturbation with a ``planted" vector ${\bf x^*}$
    and aims at finding back ${\bf x^*}$ from the matrix. This is done by sampling from the posterior probability, and therefore one considers the probability distribution:
    \begin{align}
    P_0({\bf x)} &\propto \prod_i P_X(x_i)\prod_{i<j}\exp\left(\frac{1}{\Delta\sqrt{N}}J_{ij}x_ix_j - \frac{1}{2\Delta N}x_i^2x_j^2 \right)\,,\\
   & J_{ij} = \frac{x^*_ix^*_j}{\sqrt{N}} + z_{ij}\,, \quad z_{ij}= z_{ji}\sim\Ncl(0,\Delta)\,; \nonumber\\ & P_X(x) = (1-\rho)\delta_{x,0} + \frac{\rho}{2}\left( \delta_{x,+1} + \delta_{x,-1} \right)\,, \quad x^*_i \sim P_X\;\forall\,i.
    \nonumber
    \end{align}
\end{itemize}

\begin{figure*}[t!]
    \centering
      \vspace*{-0.45cm}
    \textbf{Flow-based and diffusion-based}\par\medskip
    \vspace*{-0.15cm}
\includegraphics[width=0.243\textwidth]{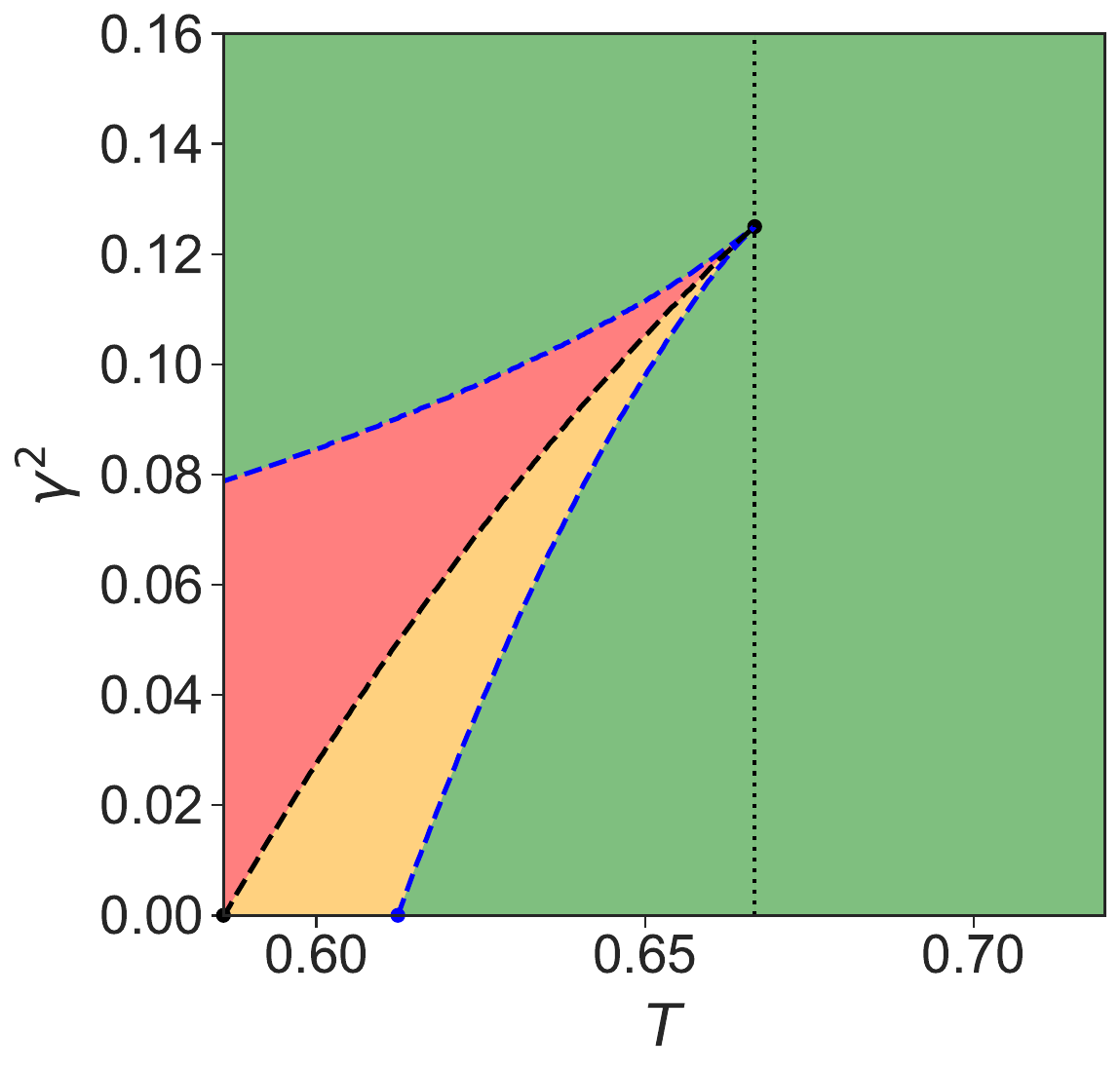}    
    \includegraphics[width=0.243\textwidth]{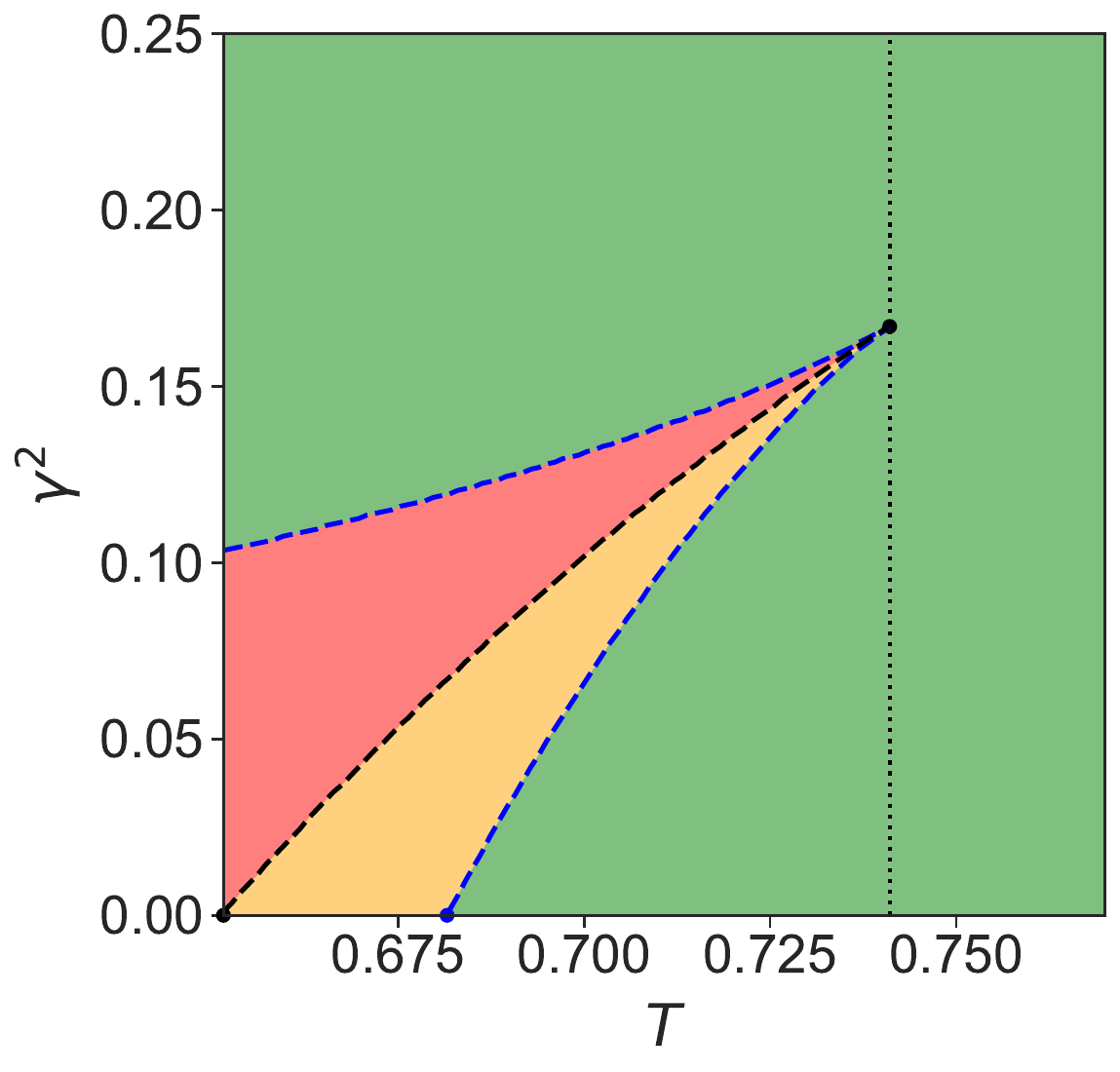}
    \includegraphics[width=0.252\textwidth]{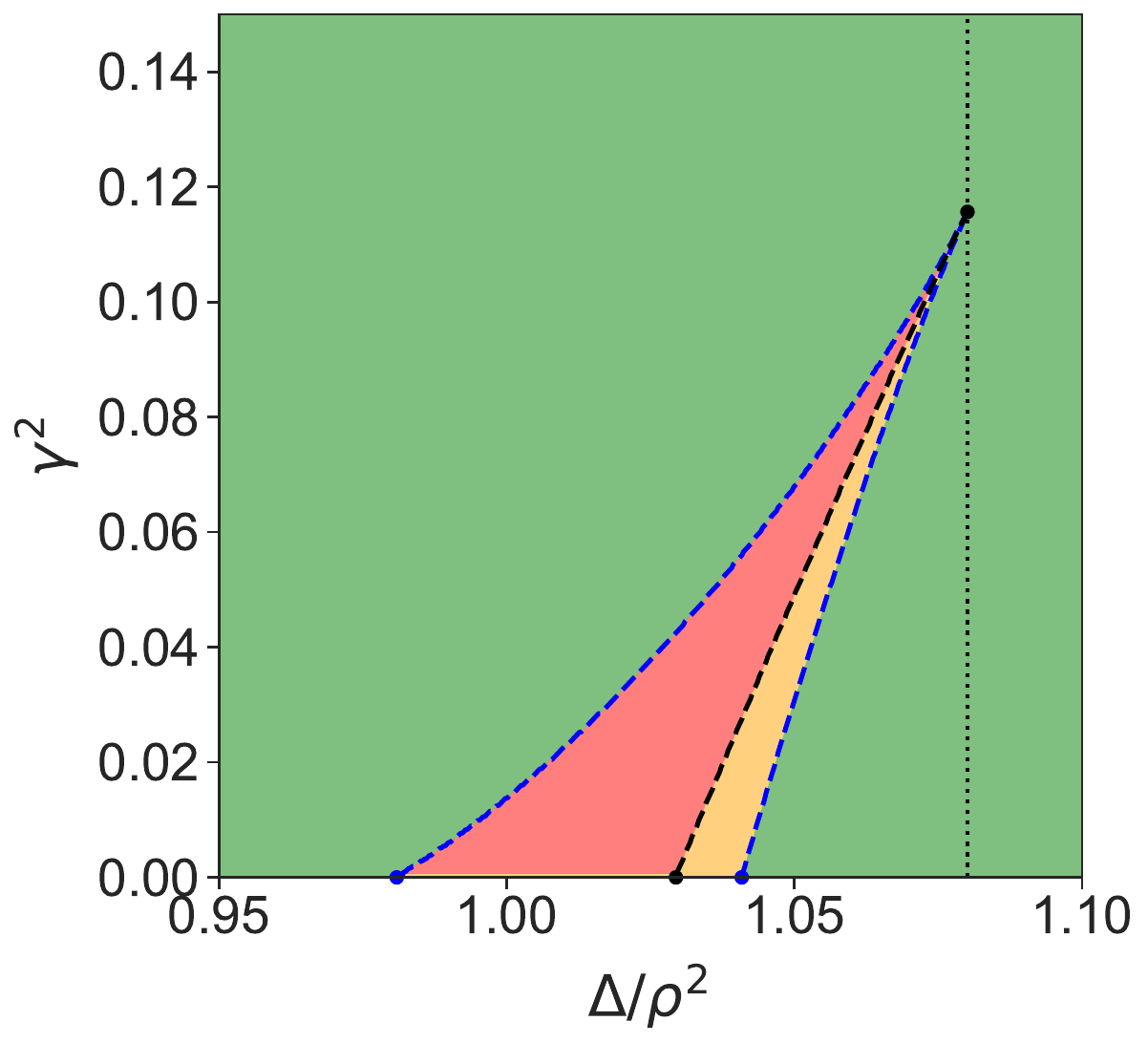}
    \includegraphics[width=0.243\textwidth]{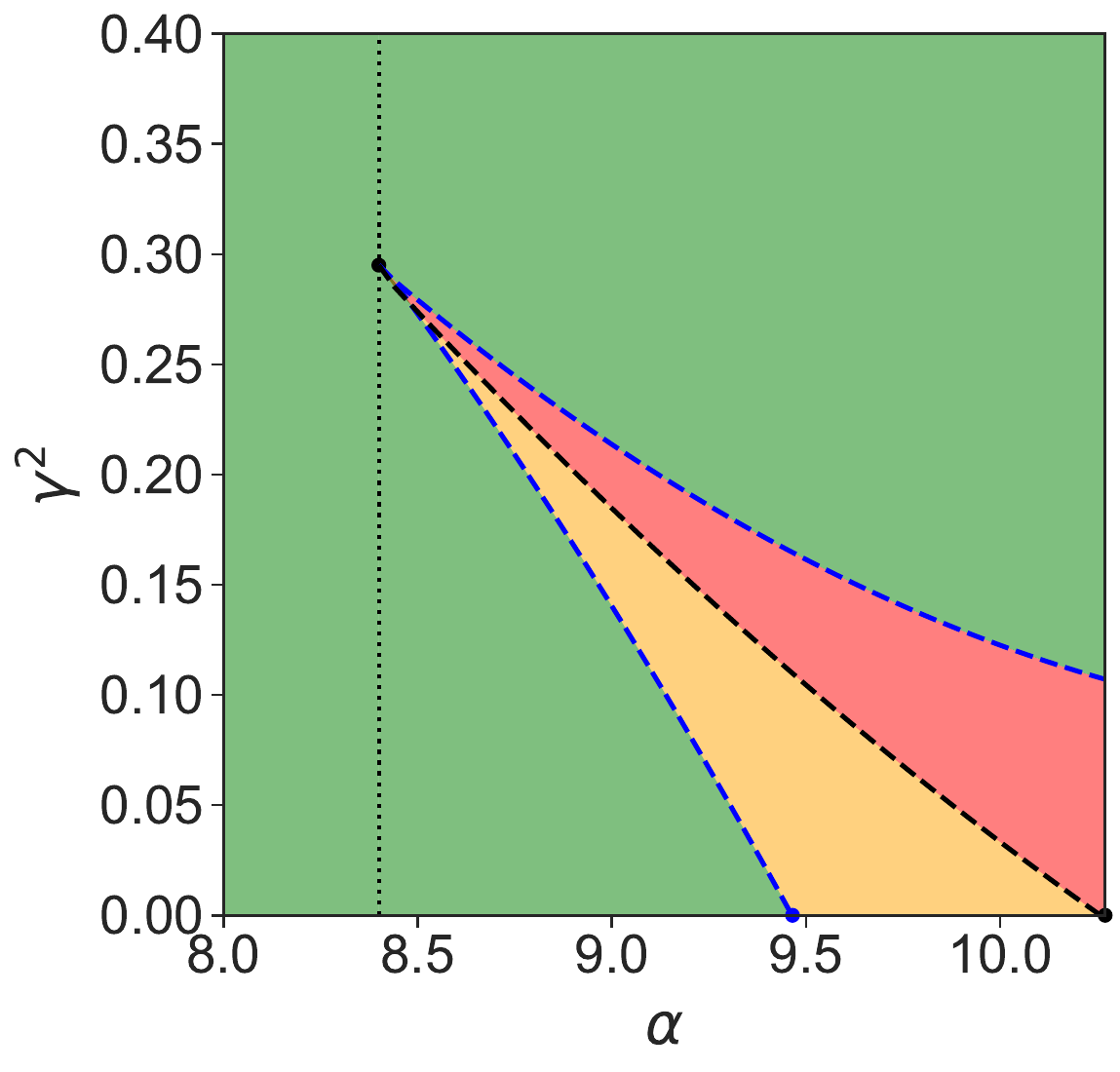}
  \par\medskip
       \vspace*{-0.25cm}
       \textbf{Autoregressive-based}
       \vspace*{-0.25cm}
      \par\medskip
    \includegraphics[width=0.241\textwidth]{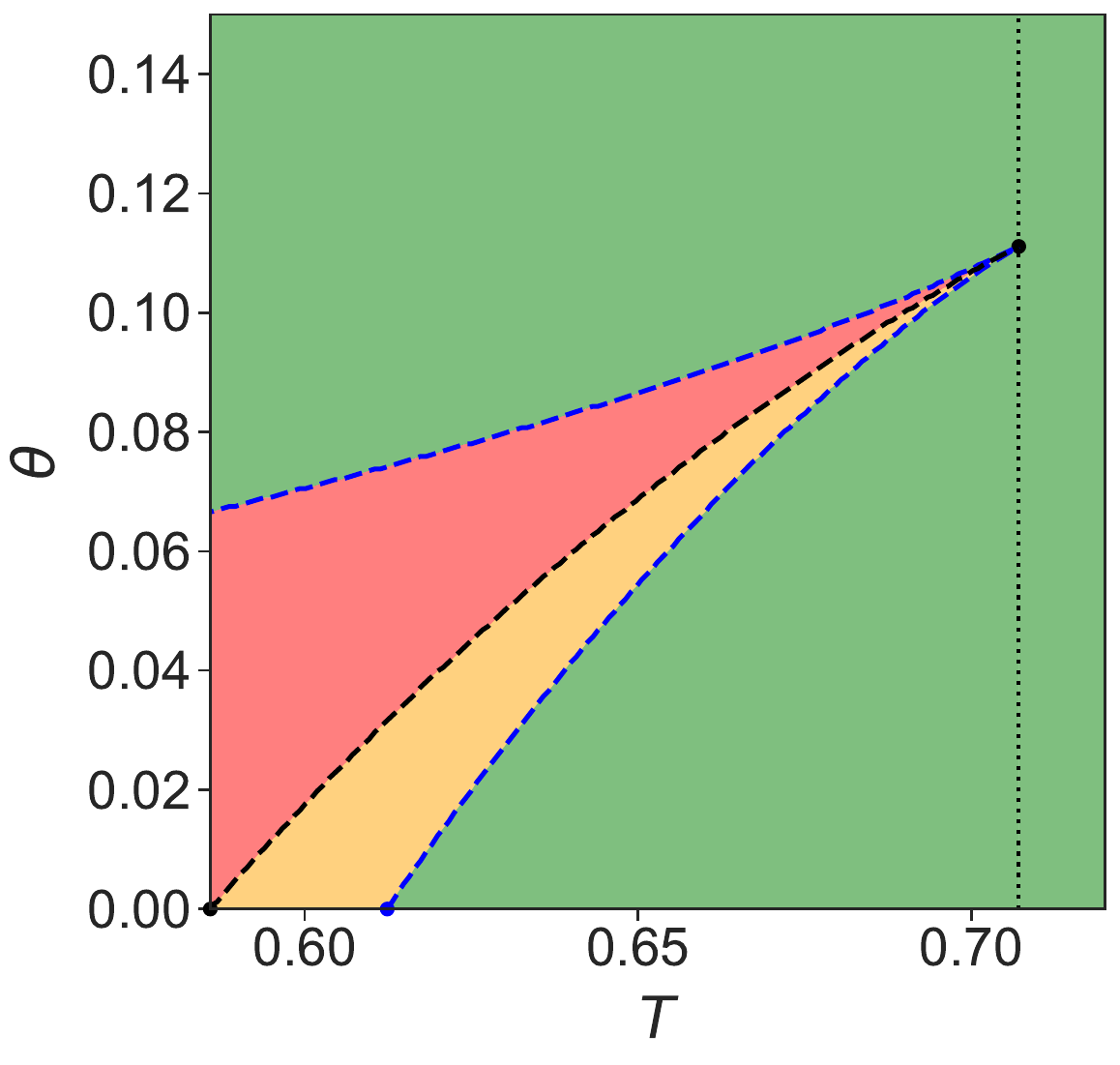}
    \includegraphics[width=0.246\textwidth]{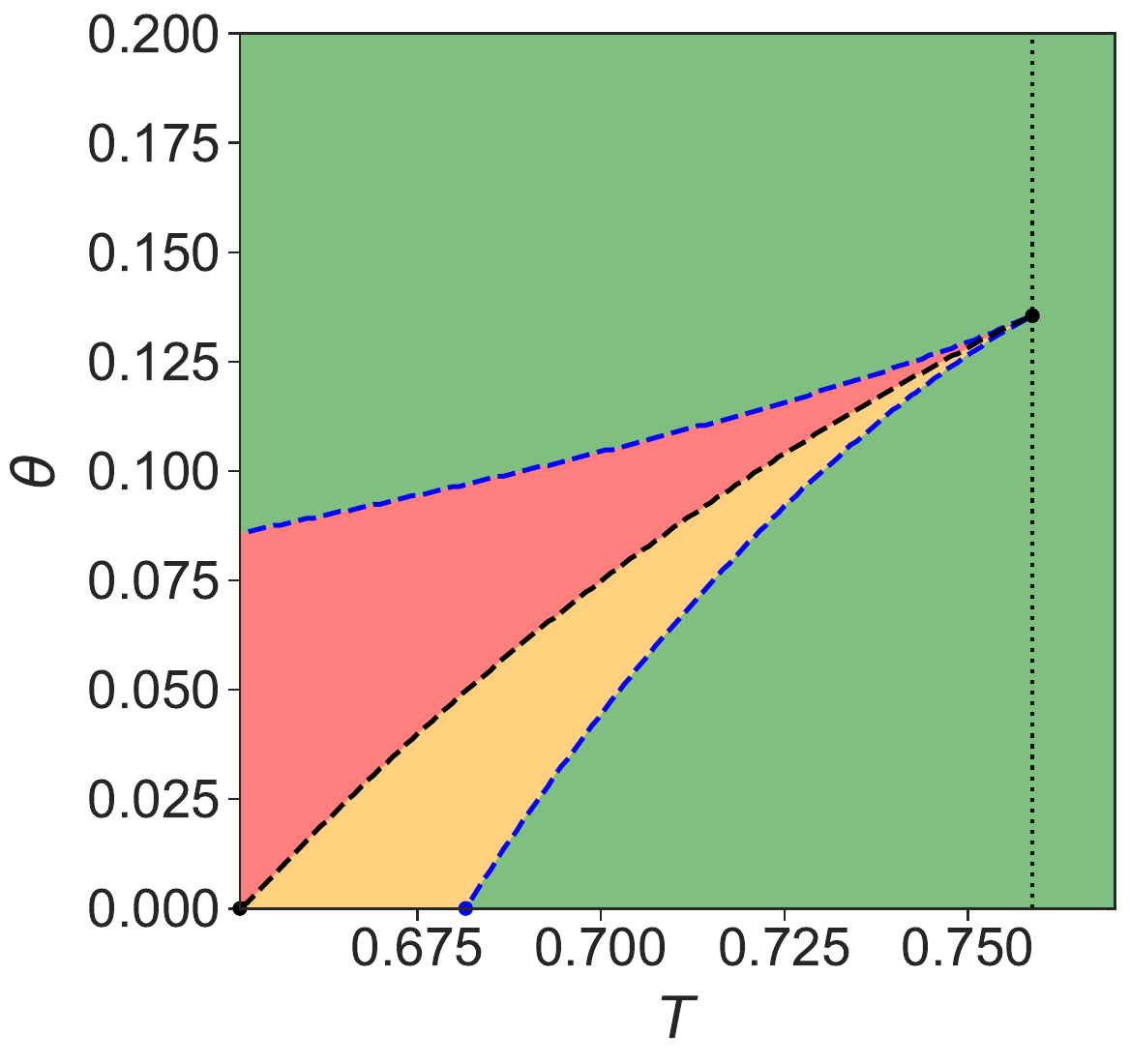}
    \includegraphics[width=0.253\textwidth]{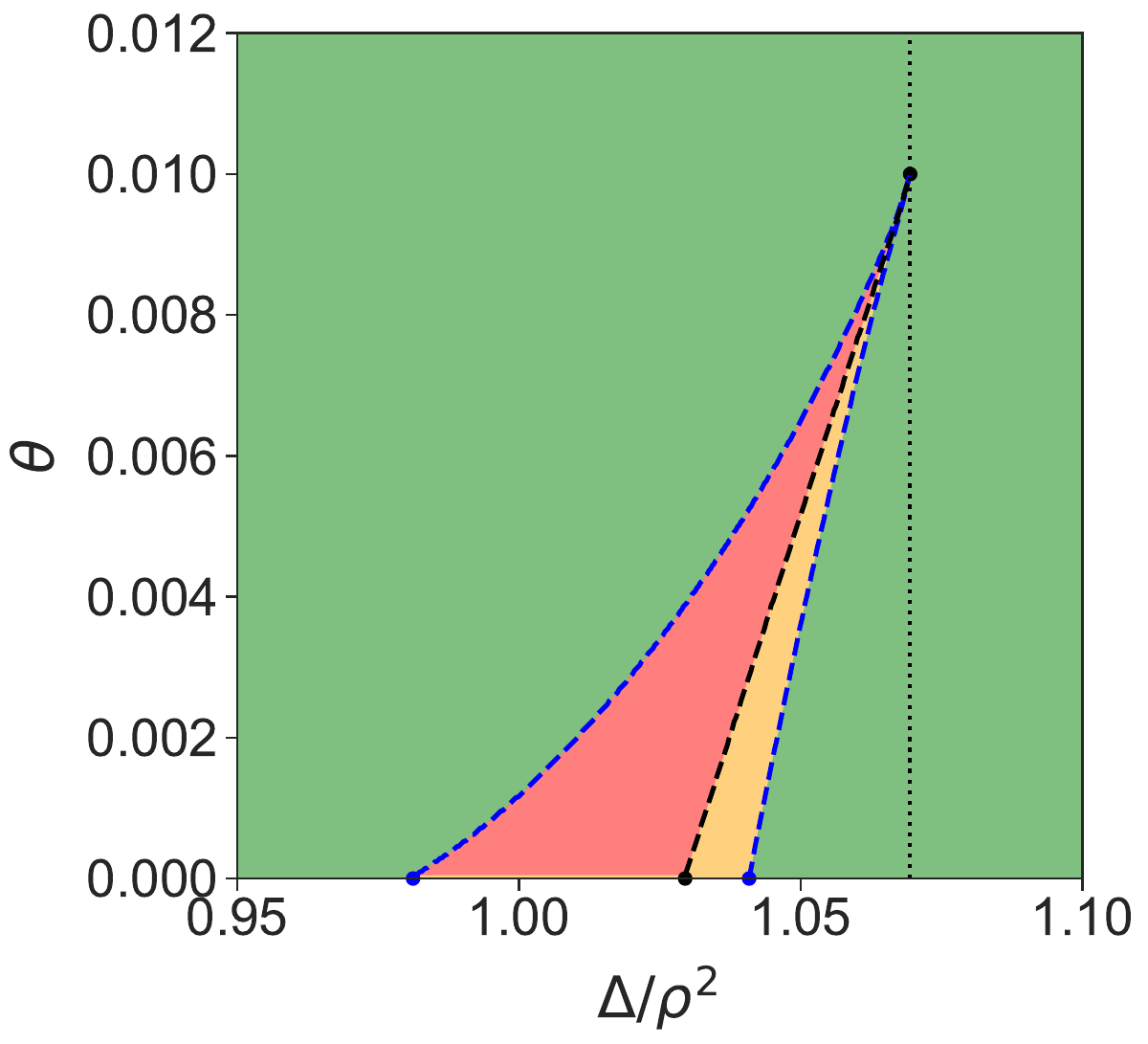}
    \includegraphics[width=0.241\textwidth]{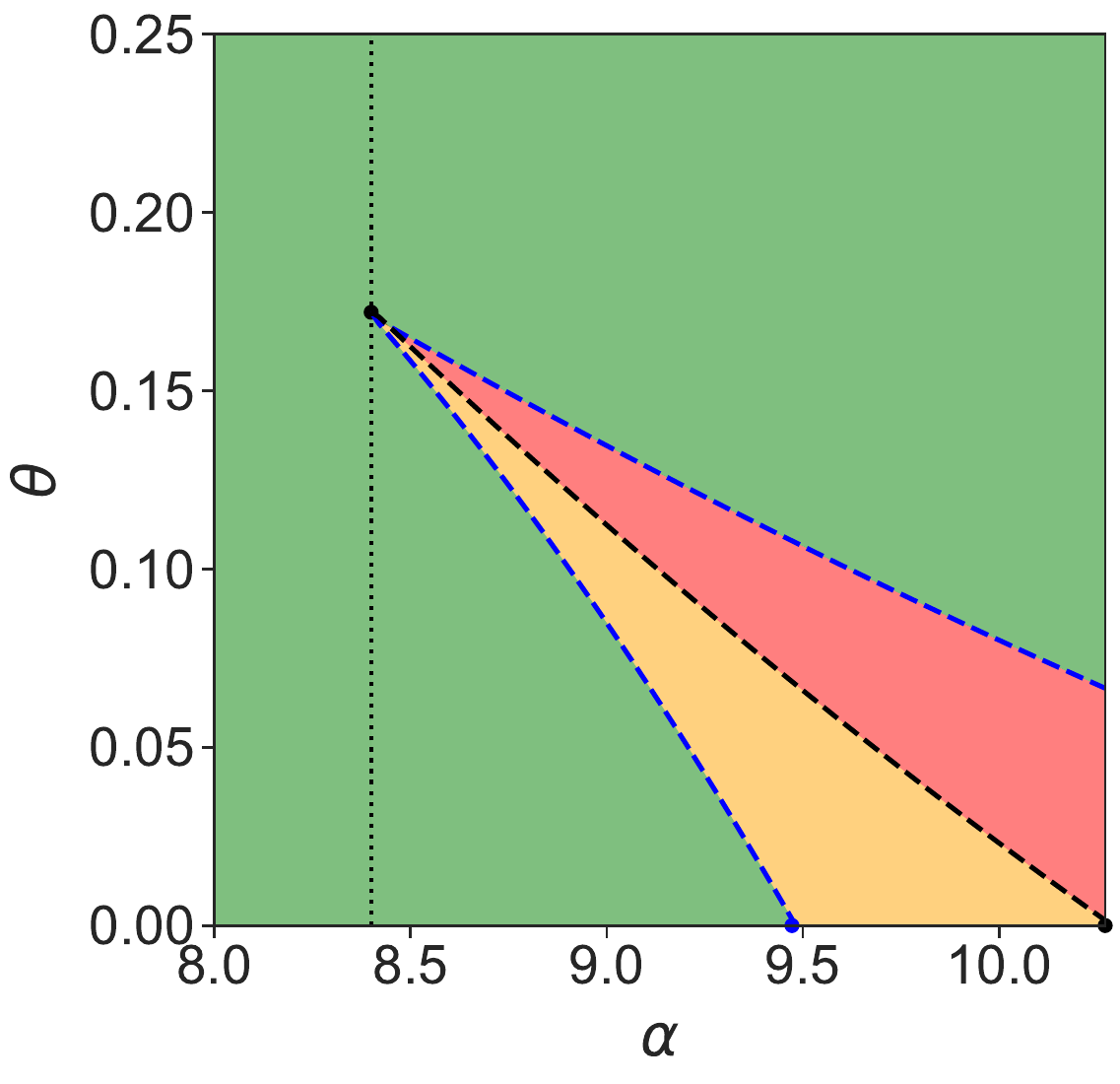}  
        \text{  (a) Spherical $p$-spin}\hspace{1.7cm} \text{(b) Ising $p$-spin} \hspace{1.65cm}\text{(c) Rank-one estimation} \hspace{1.4cm} \text{(d) Bicoloring}  
        \vspace{0.5cm}
    \caption{Phase diagrams for the tilted measure $P_\gamma$ (top), and the pinning measure $P_\theta$ (bottom): (a) The spherical $p$-spin model, b) the Ising $p$-spin model, c) the sparse rank-one matrix estimation d) the bicoloring problem on hypergraphs (NAE-SAT). The x-axis is the temperature $T=1/\beta$ in (a) and (b), the inverse-SNR $\Delta/\rho^2$ in (c) and the clauses-to-variables ratio $\alpha$ in (d), while the y-axis shows the SNR ratio $\gamma^2 = \alpha^2/\beta^2$ (top) and the decimated ratio $\theta$ (bottom). 
    In the green phase there is a single maxima to the free entropy. The red and orange regions display a phase coexistence with two maxima. In the red region efficient denoising is predicted to be algorithmically hard. 
    In (a) the spherical $p$-spin at $\gamma=\theta=0$ the dynamical threshold $T_d=\sqrt{3/8}$, the Kauzmann transition $T_{\rm K}\approx 0.58$, while the tri-critical point is at
    $T_{\rm tri}=2/3$ for diffusion and $T_{\rm tri}=\sqrt{1/2}$ for autoregressive.  
    In (b) the Ising $p$-spin the values are: $T_d\approx0.682$, $T_{\rm K}\approx 0.652$, $T_{\rm tri}\approx0.741$ for diffusion, and $T_{\rm tri}\approx0.759$ for autoregressive. 
    In (c) for the sparse rank-one matrix estimation at $\rho=0.08$ we have $\Delta_d/\rho^2 \approx 1.041$, $\Delta_{\rm K}/\rho^2 \approx 1.029$, and $\Delta_{\rm alg}/\rho^2 \approx 0.981$. The tri-critical points are at $\Delta_{\rm tri}/\rho^2\approx 1.08$ for diffusion, and $\Delta_{\rm tri}/\rho^2\approx 1.069$ for autoregressive. 
    In (d) the bicoloring the values are $\alpha_d \approx 9.465$, $\alpha_{\rm K} \approx 10.3$. The tri-critical points are $\alpha_{\rm tri}\approx 8.4$ for both diffusion and autoregressive. The curves for bicoloring were obtained by a polynomial fit, while in all the other cases we represent directly the data points.}
    \label{fig:ALL}
\end{figure*}

\section{Phase diagrams of the tilted and pinning measures}

We first discuss sampling in the spherical $p$-spin model, Eq.~(\ref{def:pspin}). We focus in particular on sampling in its paramagnetic phase, i.e. $T>T_d$, where MCMC and Langevin algorithms are predicted to work efficiently. In order to analyse flow-based, diffusion-based and autoregressive sampling, we need to consider the corresponding denoising problems. For flow-based and diffusion-based sampling, this leads to the tilted measure Eq.~(\ref{tilted-measure}), with $P_0$ now given by Eq.~(\ref{def:pspin}). Taken together, this defines a variant of the $p$-spin model with a particular random field. 

While the tilted measure may look complicated at first sight, because of the random field in the direction of a particular ``equilibrium" direction ${\bf x}_0$, it turns out that measures of this type were already studied in the literature. We refer the reader to Appendix~\ref{sec:Plant-app} for details and only briefly sketch the reasoning here. The trick (also used in e.g.~\cite{zdeborova2010generalization,el2022sampling}) is to notice that $\forall T>T_K$ the original $p$-spin model is contiguous to its ``planted" version, where a vector ${\bf x}_0$ has been hidden beforehand as an equilibrium configuration (the spike tensor model \cite{Lesieur2017Statistical}). The model is thus equivalent for all practical purposes to its planted version, with the tilted field now acting in the direction of the planted configuration. 

Computing the phase diagram of the model associated with the tilted measure thus requires computing the free entropy of the planted model with an additional side Gaussian information. This is the same computation as needed in various mathematical works based on Guerra's interpolation technique where the planted model is observed together with an additional Gaussian channel, e.g. in \cite{barbier2019optimal}. This, and the single-state property of Bayes-optimal denoising discussed in the former section (that ensures replica symmetry), allows us to obtain the phase diagram of the tilted measure.

For the present spherical $p$-spin model, we show in Appendix~\ref{sec:PhD-app} that the equilibrium properties for $T>T_{\rm K}$ are given by
\begin{align}
    \C^* &= \underset{\chi}{\rm argmax} \, \Phi_{\rm RS}(\C) \nonumber\\
 \Phi_{\rm RS}(\C)    &= \frac{\widetilde{\C}}{2} + \frac{1}{2}\log\left( \frac{2\pi}{\widetilde{\C} + 1} \right) - \frac{1}{2T^2}\C^3\,,\label{fnrg:p-spin-spherical}\\
   \,
    \widetilde{\C} &= \frac{3}{2T^2}\C^2 + \gamma^2\,. \nonumber
\end{align}
Solving the above maximization problem is easily done. One observes that depending on the range of parameters $T$ and $\gamma$, up to two local maxima of $\Phi_{\rm RS}$ can be found for this model. In Fig.~\ref{fig:ALL} (a, top) we depict in green the regions of $T,\gamma$ where $\Phi_{\rm RS}$ has a unique maximizer. 
The orange region is where two maximizers co-exist with the global one having a smaller value of the order parameter $\C$, and in the red region, the global maximizer has a larger value of $\C$. Such a phenomenology is familiar in first-order phase transitions, where the red and orange phases in Fig.~\ref{fig:ALL} correspond to the phase coexistence region. 

Now, a crucial point for the follow-up discussion of the flow- and diffusion-based sampling is that it is widely believed to be {\it algorithmically hard} to obtain the Bayes-optimal denoiser in the red region for all efficient algorithms (even if $P_0$ is known) \cite{zdeborova2016statistical,bandeira2022franz,celentano2020estimation,bandeira2022franz,gamarnik2022disordered}. This is nothing else than the well-known metastability problem in thermodynamics, as one is "trapped" in the wrong maxima of \eqref{fnrg:p-spin-spherical}. 

The evidence for denoising being algorithmically hard in the red region goes beyond mere physical analogies and is an intense subject of studies in computer science, with the study of a variety of techniques such as low degree polynomial or message passing algorithms (see e.g.~\cite{bandeira2022franz}). Note that while denoising is hard in the red region, direct sampling with MCMC is hard already in the orange one.

To explain this further, let us discuss a concrete implementation of a denoiser. Computing the marginals for the tilted measure is a classical topic in spin glass and estimation theory, and the best-known algorithm to do it efficiently is the so-called mean-field Thouless-Anderson-Palmer equations \cite{thouless1977solution}, or ---to use their modern counterpart--- the iterative approximate message passing (AMP) \cite{donoho2010message,Richard2014}. The AMP algorithm is an iterative update procedure on the estimates of the posterior means $\hat{x}_i$ and the co-variances $\sigma_i$, and for the spherical 3-spin model it reads
\begin{equation}
\begin{cases}
    B_i^t = \frac{\sqrt{3}\beta}{ N}\sum_{j<k} J_{ijk} \widehat{x}_j^{t}\widehat{x}_k^{t} - \frac{3}{N}\beta^2\widehat{x}_i^{t-1} \sigma^t {\bf \widehat{x}}^{t} \cdot {\bf \widehat{x}}^{t-1} \\
     \widehat{x}_i^{t+1} = \frac{B_i^t + \frac{\alpha(t)}{\beta(t)^2}[{\bf Y}_t]_i}{3\beta^2\|\hat {\bf x^t}\|_2^2/(2 N) + \gamma^2+1} \,,     \sigma^{t+1} = \frac{1}{3\beta^2\|\hat {\bf x^t}\|_2^2/(2 N) + \gamma^2+1}\nonumber\,. \\
\end{cases}
\end{equation}
The virtue of AMP is that its performance can be tracked rigorously over iteration time, and in fact, one can show that the overlap $\C^t$ defined by the AMP estimates obeys the following {\it state evolution}:
\begin{equation}
    \C^{t+1} = \frac{\widetilde{\C^t}}{1 + \widetilde{\C^t}}\,,\quad \widetilde{\C}^{t} \equiv \frac{3}{2T^2} (\C^{t})^2+ \gamma^2\,,
\end{equation}
which is nothing but the fixed point equation of Eq.~(\ref{fnrg:p-spin-spherical}). We see now how the presence of multiple fixed points can trap the algorithm in the wrong maximizer.

Turning an AMP-denoiser into a sampler was precisely the idea introduced in \cite{el2022information,montanari2023posterior} under the framework of stochastic localization \cite{eldan2013thin,eldan2020taming,chen2022localization}. While we derived the equation for a flow-based approach, the conclusion of \cite{el2022information,montanari2023posterior} remains. In particular, leveraging the rigorous analysis of the asymptotic error obtained by AMP \cite{bayati2011dynamics,montanari2021estimation}, they prove that AMP can approximate optimal denoising throughout the interpolation path in the regime where  the global maximum for $\Phi_{\rm RS}$ is the first one reached by the state evolution \cite{montanari2023posterior}. Furthermore, using the local convexity of the TAP free energy \cite{el2022sampling,celentano2023local}, they show that the AMP iterates satisfy Lipschitz-continuity w.r.t the observation ${\bf Y}_t$. This is crucial for the SDE-based sampling in \cite{el2022sampling,montanari2023posterior} as well as the continuous flow-based sampler in our case, since both require control of the discretization errors. We discuss this further in Appendix~\ref{sec:AnalysisAlgo}. 
However, AMP and, conjecturally \cite{gamarnik2022disordered}, any other polynomial-time denoiser (and in particular any neural-network denoiser learned from data) fail to return the correct marginal whenever the global maxima of $\Phi_{\rm RS}$ is not the first one encountered by the state evolution, i.e. in the red region in Fig.~\ref{fig:ALL}.

\subsection{Failure of generative models while MCMC succeeds}
Recall now how the flow- or diffusion-based methods use the denoiser (the marginal of the tilted measure) to produce samples from $P_0$. They start with a Gaussian noise at $\gamma=0$ and increase $\gamma$ gradually while transforming the Gaussian noise towards the direction of the marginal of the tilted measure. 
To do this for a specific temperature~$T$, one needs to be able to denoise optimally on a vertical line at {\it all} values of $\gamma \in [0,\infty[$. However, for temperatures below the tri-critical points $T_{\rm tri}=2/3$ (this value is for diffusion, in the spherical $p$-spin model) we see that we encounter the first-order phase transition, and the metastable region (red in Fig.~\ref{fig:ALL}) where optimal denoising is computationally hard and hence the uniform sampling using this strategy is as well. Another, less critical, problem is the fact that the measure will change drastically at the phase transition (as the value $\C^*$ jumps discontinuously) which means one has to be very careful with the discretization of the diffusion process. 

Based on this reasoning, we interpret the phase diagram of the tilted measure as a representation of the presence of a fundamental barrier for sampling by flow- and diffusion-based methods: The sampling scheme corresponds to starting from $\gamma=0$ and going upwards. If this path intersects the hard phase (red) at any time, it means that the Bayesian denoising cannot be performed optimally in an efficient way, thus preventing the sampling scheme from working efficiently in consequence. In particular, for the spherical $p$-spin model, we computed the curves analytically, see Appendix~\ref{sec:PhD-app}, and found that this hurdle is present at temperatures up to the tri-critical point $T_{\rm tri} = 2/3$, strictly larger than the dynamical temperature $T_d=\sqrt{3/8}$ (threshold between the orange and red region at $\gamma=0$) down to which MCMC and Langevin algorithms are predicted to sample efficiently in the literature. 

Indeed, one can write exact equations describing the Langevin dynamics \cite{cugliandolo1993analytical} (and prove them \cite{ben2006cugliandolo}). The analysis shows that Langevin is efficient at sampling for all $T>T_d$ \cite{crisanti1993spherical_dyn}. On the other hand, for $T<T_d$, we are in the so-called "dynamical" spin glass phase and Langevin fails to sample in linear time, this is the ageing regime \cite{cugliandolo1993analytical,bouchaud1998out}. In fact, all polynomial algorithms are conjectured to fail to sample efficiently, as can be proven for any "stable" algorithm~\cite{alaoui2023shattering}. 

The situation for sampling with autoregressive network, analysed via the phase diagram of the pinning measure, is very analogous, see the lower row of Fig.~\ref{fig:ALL}. 
The pinning measure again defines a variation of the original $p$-spin model, and as shown in the appendix we can compute properties at equilibrium by solving $\C^* = {\rm argmax}_\C \, \Phi_{\rm RS}(\C)$ where the RS free entropy reads
\begin{align}
    \Phi_{\rm RS}(\C) &= \frac{\widetilde{\C}}{2} + \frac{1-\theta}{2}\log\left( \frac{2\pi}{\widetilde{\C} + 1}\right)-\theta\log(\sqrt{2\pi }e) -\frac{\C^3}{2T^2}\,, \nonumber \\ \widetilde{\C} &= \frac{3\C^2}{2T^2}\,.
\end{align}
The same reasoning can be applied to this phase diagram (and in fact the so-called decimated versions of message-passing algorithms were proposed as early as in \cite{mezard2002analytic,chavas2005survey}). Concerning the efficiency of sampling, the same phenomenology appears again when interpreting this phase diagram. In fact, for the spherical $p$-spin model, it turns out to be even worse because the temperature where the difficulty arises for the autoregressive method is $T_{\rm tri}=\sqrt{1/2}$, which is larger than the one for flow-based and diffusion models $T_{\rm tri}=2/3$. So in this case, autoregressive-based sampling algorithms perform worse than flow- or diffusion-based. 

\subsection{Other models} Fig.~\ref{fig:ALL} then evaluates the phase diagrams of the tilted and pinning measures for three other models -- the Ising $p$-spin (b), the rank-one matrix estimation (c), and the bicoloring problem on sparse hypergraphs (d). 
For the bicoloring problem, defined on random sparse hypergraphs, we can use the belief propagation equations as Bayesian denoiser \cite{Castellani2003, Gabrié_2017, ricci2019typology}. The resulting equations are quite long, and we defer their presentation in the appendix, along with their derivation using the cavity method. 

We observe the same phenomenology with tri-critical points causing hurdles for flow-, diffusion- and autoregressive methods reaching out to the phase where traditional approaches based on MCMC or Langevin work efficiently. 
In fact, we expect this picture to always appear for any model with the RFOT phenomenology where the dynamical temperature is distinct from the ideal glass or Kauzmann temperature. Such a phenomenology was described in many problems far beyond those we picked to study, and we can hence anticipate follow-up studies identifying analogous phase diagrams in many other problems of interest. 

Finally, we notice that depending on the model, the position of the tri-critical point for flow- and diffusion-based methods is better (e.g. for the spherical $3$-spin, Ising 3-spin) or worse (e.g. for the sparse rank-one matrix estimation) than for the autoregressive methods. In any case, the position of the tri-critical point does depend on the noise channel to which the generative model maps. This leaves open the question of which channel one should use, for each model, to minimize the range of values for which generative model-based sampling is suboptimal compared to MCMC techniques that fail at the dynamical threshold $T_d$. It is not inconceivable that it is possible to reduce $T_{\rm tri}$ very close (or maybe even up to) $T_d$ by optimizing over different distributions in linear interpolant \cite{albergo2023stochastic}, or using non-linear maps \cite{montanari2023posterior}. This is left for future studies.

\subsection{Outperforming MCMC in inference models with a hard phase}
We now discuss a situation more advantageous for modern techniques. Column (c) of Fig.~\ref{fig:ALL} depicts the phase diagram of the sparse rank-one matrix estimation problem that presents an additional interesting feature: here there is a planted ``hidden" signal that we seek to recover. At large values of the noise the signal is hidden and the RFOT-type phenomenology reappears so that the high noise behaviour is identical to the high-temperature models of the other models. 

At low noise values, however, there is another phase transition at $\gamma=\theta=0$ denoted as $\Delta_{\rm alg}$ below which the AMP algorithm solves the estimation problems optimally and above which it does not up to the value $\Delta_{\rm IT}$ (the hard region, in red, is delimited by these two values). Going vertically up in the phase diagram in $\gamma$ or $\theta$ for $\Delta<\Delta_{\rm alg}$ does not cause any encounter of the hard (red) phase, and thus sampling based on flow or diffusion or autoregressive networks works. This has been proven rigorously recently in \cite{wu2019solving} (with some technical regularity assumptions on the denoiser \cite{celentano2023local}).

Yet existing literature collects evidence that in inference problems that present such a hard phase, local dynamics algorithms such as MCMC and Langevin are {\it not able} to sample efficiently until some yet lower values of noise $\Delta_{\rm MCMC}$. In particular, this suggestion was put forward indirectly in \cite{antenucci2019glassy} by arguing that the metastable phase in the hard region is glassy, and this glassy nature extends well beyond the region that is hard for message passing algorithms such as AMP. This was then shown explicitly in follow-up works starting with an analysis of the dynamics in a mixed spiked matrix-tensor model in \cite{mannelli2020marvels}, in the phase retrieval problem \cite{sarao2020complex}, the planted coloring \cite{angelini2023limits} and on a rigorous basis in the planted clique problem \cite{chen2023almost}. Works such as \cite{sarao2020optimization} suggest that over-parametrization observed in modern neural networks mitigates those hurdles, and this may be one of the reasons why over-parametrization is beneficial. 

In light of these works, it is interesting to note that sampling based on flow, diffusion or autoregressive networks also avoids hurdles stemming from the glassiness of the hard phase and rather effortlessly so by working in the space of marginals rather than configurations directly. The phase diagram presented in Fig.~\ref{fig:ALL} indicates that both diffusion and autoregressive networks sample the $P_0$ efficiently for any $\Delta < \Delta_{\rm alg}$. This poses an intriguing question for future work of whether over-parametrization of neural networks that are learning the denoisers from data would still be so beneficial in these methods. 

Finally, we comment on the relevance of these findings beyond the specific model discussed here, and in particular for the study of physical objects in finite dimension with short-range interaction. Do we expect a problem embedded in finite dimension to suffer the same fate as the one discussed here? While similar phase diagrams as Fig.~\ref{fig:ALL} have been observed in finite dimension in e.g.~\cite{biroli2008thermodynamic,cammarota2013random}, the phenomenology of first-order transition is different. From nucleation arguments, the exponentially hard denoising phase is not expected to exist in finite dimension~\cite{krzakala2011melting}. Indeed, it can be proven rigorously that for any graph that can be embedded in a finite-dimensional lattice, an efficient algorithm exists~\cite{el2021computational}. In this case, the analysis of whether a good denoiser can be learned with a neural network will require a finer study, depending on the discretization of the ODE close to the transition, and on the number of points in the data set (perhaps in the vein of~\cite{biroli2023generative}). These are interesting potential new directions of research.

\section*{Conclusion}
Our investigation into the efficiency of sampling with modern generative models, in comparison with traditional methods, reveals distinct strengths and weaknesses for each. By examining a specific class of probability distributions from statistical physics, we identified parameters where either method excels or falls short. Significantly, our approach highlighted challenges stemming from a first-order discontinuous transition for generative models-based sampling techniques even in regions of parameters where traditional samplers work efficiently. While generative models have shown promise across various applications, it is crucial to understand their potential pitfalls and advantages in specific contexts, and our paper makes a key step in this direction. 

\section*{Acknowledgments}
We acknowledge funding from the Swiss National Science Foundation grant SNFS OperaGOST and SMArtNet. We also thank Ahmed El Alaoui, Hugo Cui, and Eric Vanden-Eijnden for enlightening discussions on these problems.

\newpage
\appendix

\section*{Appendices}

\section{\label{sec:Plant-app}Computing free entropies: the planting trick}

The main technical difficulty is the study of the tilted or pinned measure and its correlation with an equilibrium configuration. Here, we explain how this difficulty is avoided. Given a probability measure $P_0$, the tilted measure reads
\begin{equation}
    P_\gamma({\bf x}|{\bf x}_0,{\bf z}) \propto  e^{\gamma(t)^2\langle {\bf x}, {\bf x}_0 \rangle + \gamma(t)\langle z,{\bf x} \rangle - \frac{\gamma(t)^2}{2}\|{\bf x}\|^2} P_0({\bf x})\,.
    \label{tilted-measure-app}
\end{equation}
We shall restrict ourselves to ``planted" models with a hidden assignment in the following. While this is the case for actual inference statistical models (such as the sparse Wigner model we consider), it turns out that the $p$-spin \cite{lesieur2015phase} and the NAESAT model~\cite{ding2014satisfiability} are contiguous to their planted version as long as $T>T_K$ (for $p$-spin) or $\alpha<\alpha_K$ for NAESAT. For the $p$-spin models, such contiguity has also been rigorously established in certain setups \cite{perry2016statistical,Lesieur2017Statistical,jagannath2020statistical}.

For all practical purposes, we shall thus study planted models. In this case, the difficulty is greatly simplified, as the joint distribution of the planted vector and the disorder equals the joint distribution of the disorder and an equilibrium configuration. 
This equivalence, along with the contiguity with the planted models for non-inference problems, allows us to replace the equilibrium configuration ${\bf x_0}$ in the tilted measure with the planted vector.

For the case of the Sherrington-Kirkpatrick (SK) model, such an equivalence was rigorously proven and utilized in \cite{el2022sampling} using the contiguity between the planted and unplanted models. Concretely, Proposition 4.2 in \cite{el2022sampling} shows the equivalence between the following two methods for generating the tilted measure:
\begin{enumerate}
    \item Sample ${\bf x_0}$ uniformly. Then sample the interaction matrix $J$ as $J=\frac{\beta}{n}{\bf x_0}{\bf x_0}^\top+W$ where $W \sim \text{GOE}$.
    \item Sample $J \sim \text{GOE}$. Then sample ${\bf x_0} \sim P_{\rm SK}(J)$,
\end{enumerate}
where $P_{\rm SK}(J)$ denotes the Boltzmann Sherrington-Kirkpatrick measure with interaction matrix $J$.
Similarly, based on the contiguity between planted and unplanted models, we assume that the above equivalence holds for $T>T_K$ (for $p$-spin) and for $\alpha<\alpha_K$ (for NAESAT).
We note that for inference problems, i.e. when $J$ is planted for both points (1), (2) above, the equivalence holds directly based on the definition of the posterior measure.

Additionally, the tilted measure can be then seen, in terms of free entropy, as the measure associated now to an inference problem. Consider for concreteness the planted $p$-spin, also called the spike-tensor model \cite{Richard2014,Lesieur2017Statistical}: One extracts a random vector $X^0 \in \mathbb R^N$ from a Gaussian or a Rademacher distribution, and then one aims at recovering $X^0$ from the measure of a) a Gaussian measurement $Y=\alpha X^0 + \beta Z$ and b) a noisy tensor measurement $J_{ijk} = X^0_iX^0_jX^0_k + \frac{1}{T} \xi_{ijk} \forall i<j<k$. The Bayesian posterior of this model is nothing but the tilted measure~(\ref{tilted-measure-app}) applied to the tensor spike model.

Interestingly, such considerations are not new. Nishimori \cite{nishimori2001statistical} and later Iba \cite{iba1999nishimori} already used a similar trick, and \cite{montanari2006rigorous,krzakala2009hiding,zdeborova2010generalization} used it to discuss the dynamics starting from equilibrium conditions, while \cite{wong2000error} used it in the context of error correction, all for the $p$-spin model. Recently, the same technique was used to prove the clustering property in the $p$-spin model \cite{alaoui2023shattering}.

Adding a Gaussian measurement to an inference problem is also a classical trick used when proving free energies in the mathematical physics literature, especially in the context of Guerra interpolation \cite{guerra2014interpolation,panchenko2013parisi} for Bayes optimal models, see e.g. \cite{korada2009exact,barbier2019adaptive,barbier2019optimal,barbier2021overlap}, and thus such free energies have been solved rigorously as well.

How does this change for the decimated problem? In this case, the additional Gaussian channel is replaced by an erasure channel. This is precisely one of the alternatives used as well in the mathematical physics literature! Indeed, the pinning lemma \cite{abbe2013conditional} (see also \cite{coja2017information}) is often used instead of the Gaussian channel.

These considerations are really helpful, as we can now solve these problems as a simple variant of problems already solved in the literature, often rigorously (in the case of the $p$-spin we refer to \cite{lesieur2017constrained} and \cite{mannelli2020marvels}, and for the spike Wigner model to \cite{deshpande2014information,dia2016mutual,lelarge2019fundamental}. 

For the sparse case, a rigorous control is harder, and thus we shall simply stay at the level of rigour of the cavity method \cite{mezard1987spin} and use the results of \cite{Gabrié_2017}. Note however that the method is trustworthy \cite{ding2014satisfiability}.

Finally, since these are variations of known inference problems, we can leverage on the existing work on approximate message passing \cite{donoho2010message} for the $p$-spin model \cite{Richard2014,lesieur2017constrained} and the spike model \cite{deshpande2014information}, see in particular \cite{Lesieur2017Statistical} for a detailed presentation.

\section{\label{sec:PhD-app}Asymptotic solutions and Phase Diagrams}
In this section, we present how each of the phase diagrams we presented in the main text can be produced. Specifically, we first remind our definition of tilted and pinning measures and of the order parameters we are going to use for the analysis. 
After defining the probability distribution associated to each problem, 
we give the expression of the RS free entropy, from which one can compute the values of the parameters of the problem at which the \textit{spinodal points} are located, i.e. the points at which the potential develops a second maxima by continuous deformation, and also the point corresponding to the \textit{IT transition}, i.e. where the two maxima exchange the role of global and local maxima.

We then report the expression of the denoisers (AMP/BP) used in the sampling schemes illustrated in the main text, with their associated self-consistent asymptotic equations (state evolution for AMP and Cavity equations for BP). Looking at the fixed points of these equations, starting from an uninformed and informed initialization, we can plot the difference between the values of the order parameter reached at the fixed point in these two cases, which allows detecting the phases in which multiple fixed points are present.

In Eq.~(\ref{tilted-measure-app}) we have recalled the expression for the tilted measure characterizing the flow-based sampling scheme. As we already mentioned, $\gamma$ is nothing but a rescaled sampling time, such that studying the properties of the tilted measure varying $\gamma$ allows us to characterize the properties of the Bayesian denoising problem at all times during sampling.

In the same way, let us remind the pinning measure, which we use to analyse the autoregressive-based sampling procedure:
\begin{equation}
    P_{\theta}({\bf x}| {\bf x}_0, S_\theta) \propto  P_0({\bf x}) \prod_{i \in S_\theta}     \delta([{x}_i-[{\bf x}_0]_{i}) 
    \label{pinning-measure-app}
\end{equation}
where $\theta$ is the fraction of pinned variables.

Finally, let us remind that we shall study the evolution in time (or equivalently in $\gamma \in [0,\infty[$ (AWGN) and $\theta \in [0,1]$ (BEC)) of the following order parameters:
\begin{eqnarray}
\label{eq:Mu-app}
\M(\gamma) &\equiv& \frac{1}{N}\mathbb{E}[ \hat {\bf x}(\gamma) \cdot {\bf x}_0]\,,\\ 
\label{eq:Chi-app}
\C(\gamma) &\equiv& \frac{1}{N}\mathbb{E}[\norm{ \hat {\bf x}(\gamma)}^2]\,,
\end{eqnarray}
and analogously we can define $\M(\theta)$ and $\C(\theta)$. Concretely, we will consider only cases in which the Nishimori identities hold, such that these two quantities always coincide, and thus we will restrict our analysis to $\chi$. 

Let us now go through each one of the models mentioned in the main.

\subsection{\label{sec:SpSK}Sparse rank-one matrix factorization}
We consider the Bayes-Optimal rank-one matrix estimation (or rank-one matrix factorization) problem: 

Given a hidden vector ${\bf x^*}$, sampled from the so-called \textit{Rademacher-Bernoulli} prior distribution
    \begin{equation}
        P_X(x) = (1-\rho)\delta_{x,0} + \frac{\rho}{2}\left( \delta_{x,+1} + \delta_{x,-1} \right)\,, \quad x^*_i \sim P_X\;\forall\,i\,. \nonumber
    \end{equation}
one has access to noisy observations, that is a matrix $J_{ij}$ is composed by a rank-one spike plus i.i.d. Gaussian noise:
    \begin{align}
   & J_{ij} = \frac{x^*_ix^*_j}{\sqrt{N}} + \widetilde{z}_{ij}\,, \quad \widetilde{z}_{ij}= \widetilde{z}_{ji}\sim\Ncl(0,\Delta)\,;
    \nonumber
    \end{align}
and the goal is to infer ${\bf x^*}$ in the best way possible. There are many important problems in statistics and machine learning that can be expressed in this way \cite{donoho2018optimal}, and this model has been the subject of many works both from the statistics \cite{baik2006eigenvalues,rangan2012iterative,deshpande2014information,krzakala2016mutual,lelarge2019fundamental} and the statistical physics communities \cite{lesieur2015mmse,lesieur2017constrained}. The presentation of this problem closely follows the study presented in \cite{lesieur2017constrained}.

From the Bayesian point of view, the problem amounts to sampling from the posterior.
One way to introduce the model is  through the following probability distribution:
\begin{equation}
    P_0({\bf x)} \propto \prod_i P_X(x_i)\prod_{i<j}\exp\left(\frac{1}{\Delta\sqrt{N}}J_{ij}x_ix_j - \frac{1}{2\Delta N}x_i^2x_j^2 \right)\,.
    \end{equation}
With this distribution, ${\bf x}$ will be a random vector with, on average, a fraction $\rho$ of components that are Ising spin variables (i.e. each $x_i$ takes values $\pm1$) and the rest of the entries that are put to zero, in such a way that the parameter $\rho$ controls the sparsity of the vector we want to retrieve.

As discussed in Appendix~\ref{sec:Plant-app}, the \textbf{tilted measure} with an equilibrium vector ${\bf x}_0$ is equivalent to the one where the vector ${\bf x}_0$ is planted. Therefore, in what follows, we shall assume that ${\bf x}_0$ corresponds to the planted configuration.
We thus obtain the following tilted measure for diffusion and flow-based models:
\begin{eqnarray}\label{eq:tilted_plant}
    P_{\gamma}({\bf x)} &=& \frac 1{Z_{\gamma}}  \left(\prod_i P_X(x_i)\right) e^{\gamma(t)^2\langle {\bf x}, {\bf x}_0 \rangle + \gamma(t)\langle {\bf z},{\bf x} \rangle - \frac{\gamma(t)^2}{2}\|{\bf x}\||^2} 
  \left( \prod_{i<j}e^{\frac{1}{\Delta\sqrt{N}} \tilde z_{ij}x_ix_j + 
 \frac 1{\Delta\N}x_ix_j[{\bf x}_0]_i[{\bf x}_0]_j
  - \frac{1}{2\Delta N}x_i^2x_j^2 }\right)
  \end{eqnarray}

 With respect to the original model, the tilted measure simply includes in addition, a field in the planted direction, a random field and a renormalization of the constant in front of the quadratic part. As mentioned in Appendix~\ref{sec:Plant-app}, this corresponds to an equivalent inference problem with an additional Gaussian measurement.

The \textbf{pinned measure} is slightly different. In this case, it modifies the original problems as (denoting the pinned list as $S_\theta$):
\begin{eqnarray}
    P_{\theta}({\bf x)} &=& \frac 1{Z_{\theta}}  
    \left(\prod_{i \notin S_\theta} P_X(x_i)\right) 
        \left(\prod_{i \in S_\theta}  \delta(x_i-x_i^*)\right) 
  \left( \prod_{i<j}e^{\frac{1}{\Delta\sqrt{N}} \tilde z_{ij}x_ix_j + 
 \frac 1{\Delta\N}x_ix_j[{\bf x}_0]_i[{\bf x}_0]_j
  - \frac{1}{2\Delta N}x_i^2x_j^2 }\right)
  \end{eqnarray}

\subsubsection{Replica free entropy} The replica formula for such problems can be found in many places and we refer to the mathematical literature for the detailed rigorous statements: \cite{deshpande2014information,krzakala2016mutual,dia2016mutual,lelarge2019fundamental,barbier2019adaptive,el2018estimation}. In particular, \cite{barbier2019adaptive} gives a generic proof using the adaptive interpolation methods in the presence of a Gaussian channel is added, which turns out to give the same measure as the tilted one, while \cite{lelarge2019fundamental} uses instead a pinned measure. In both case, we can thus adapt the results in the literature:
The asymptotic free entropy is given by the maximum of the so-called replica symmetric potential\cite{mezard1987spin} :
\begin{eqnarray}
    \frac 1N {\mathbb E}_{{\bf x}_0,{\bf z},\tilde {\bf z}} \log Z_{\gamma} &\xrightarrow[N \to \infty]{}& {\rm argmax  }\;\Phi_{RS}(m)\\
    \Phi_{RS}(m) &=& \mathbb{E}_{w,x_0}\left[ \log Z_x\left(\frac{m}{\Delta},\frac{m}{\Delta} x_0 + \sqrt{\frac{m}{\Delta}}w \right)\right] - \frac{m^2}{4\Delta}
    \label{eq:PhiRS_SK}
\end{eqnarray}
where $m$ is the order parameter of the problem, $x_0\sim P_X$, $w\sim \mathcal{N}(0,1)$ while $Z_x$ depends on the specific measure considered. The same is valid for the pinned measure, as long as one substitutes $Z_{\gamma}$ with $Z_{\theta}$.

\paragraph*{Tilted measure:}
In the case of the {\bf tilted measure}~(\ref{tilted-measure-app}), we have
\begin{equation}
\begin{split}
    Z_x(A,B;x_0) &= \int \dd x P_X(x)\exp(\gamma^2xx_0 + \gamma wx - \gamma^2x^2/2) \exp(Bx - Ax^2/2)\\  
    &=\rho e^{ -(A+\gamma^2)/2}\cosh(B + \gamma w + \gamma^2 x_0) + (1-\rho)
\end{split}
\end{equation}
that we can put into Eq.~(\ref{eq:PhiRS_SK}) to get
\begin{equation}
\begin{aligned}
    \Phi_{\rm RS}(\C) = &\rho \E_w \left[ \log\left( (1-\rho) + \rho e^{-\widetilde{\C}/2}\cosh\left(\widetilde{\C} +\sqrt{\widetilde{\C}}w\right)\right) \right] \\ &+ (1-\rho) \E_w \left[ \log\left( (1-\rho) + \rho e^{-\widetilde{\C}/2}\cosh\left(\sqrt{\widetilde{\C}}w\right)\right) \right] - \frac{\C^2}{4\Delta}\,, \quad \widetilde{\C} = \frac{\C}{\Delta} + \gamma^2\,.
\end{aligned}
\end{equation}

\paragraph*{Pinning measure:}
Meanwhile, considering the {\bf pinning measure}~(\ref{pinning-measure-app}) leads to
\begin{equation}
\begin{split}
    Z_x(A,B;x_0) &= 
    \begin{cases}
        \int \dd x P_X(x)\delta_{x,x_0} \exp(Bx - Ax^2/2) & \text{with probability } \theta \\
        \int \dd x P_X(x) \exp(Bx - Ax^2/2) & \text{with probability } 1 - \theta
    \end{cases}
    \\  
    &=
    \begin{cases}
        P_X(x_0)\exp\left( Bx_0 - Ax_0^2/2\right) & \text{with probability } \theta \\
        \rho \exp(-A/2)\cosh(B) + (1-\rho) & \text{with probability } 1 - \theta
    \end{cases}
\end{split}
\end{equation}
which in turn gives
\begin{equation}
\begin{aligned}
    \Phi_{\rm RS}(\C) = &\theta \left( \rho \log \rho + (1-\rho)\log(1-\rho) + \frac{\rho}{2} \widetilde{\C}\right) \\ &+(1-\theta)\bigg[\rho \E_w \left[ \log\left( (1-\rho) + \rho e^{-\widetilde{\C}/2}\cosh\left(\widetilde{\C} +\sqrt{\widetilde{\C}}w\right)\right) \right] \\ &+ (1-\rho) \E_w \left[ \log\left( (1-\rho) + \rho e^{-\widetilde{\C}/2}\cosh\left(\sqrt{\widetilde{\C}}w\right)\right) \right]\bigg] - \frac{\C^2}{4\Delta}\,, \quad \widetilde{\C} = \frac{\C}{\Delta}\,.
\end{aligned}
\end{equation}

\subsubsection{Message-Passing algorithm} The derivation of the AMP algorithm for this problem has a long history, and is connected to the Thouless-Anderson-Palmer (TAP) equations \cite{thouless1977solution}. For this problem, the introduction of TAP as an iterative algorithm is due to Bolthausen \cite{bolthausen2014iterative} and has been adapted to the present situation in \cite{rangan2012iterative,deshpande2014information}. 

\paragraph*{Tilted measure:}
The equivalence between the tilted and the planted measure with external field allows us to reduce the AMP iterations for the tilted measure to the ones for an associated inference problem. In \cite{lesieur2017constrained}, the authors provided a framework for deriving the AMP iterates for such an inference problem involving pair-wise interactions between spins. While the tilted measure defined by Eq.~(\ref{eq:tilted_plant}) involves additional random and planted fields, the generality of the derivation in \cite{lesieur2017constrained} allows them to be straightforwardly incorporated into the single-site factors $P_X(x_i)$ in \cite{lesieur2017constrained}. Through an adaptation of the derivation of Equations 66, 67 in  \cite{lesieur2017constrained}, we obtain:

\begin{equation}\label{eq:AMP_SK_diff}
\begin{cases}
    \widehat{x}_i^{t+1} = \frac{\rho \tanh(B_i^t + \frac{\alpha(t)}{\beta(t)^2}[{\bf Y}_t]_i)}{\rho + \frac{(1-\rho)\exp((A^t + \gamma^2)/2)}{\cosh(B_i^t + \frac{\alpha(t)}{\beta(t)^2}[{\bf Y}_t]_i)}}, \quad
    \sigma_i^{t+1} = \rho \frac{\rho + (1-\rho)e^{(A^t+\gamma^2)/2}\cosh(B_i^t +\frac{\alpha(t)}{\beta(t)^2}[{\bf Y}_t]_i)}{\left( \rho \cosh(B_i^t + \frac{\alpha(t)}{\beta(t)^2}[{\bf Y}_t]_i) + (1-\rho) \exp((A+\gamma^2)/2)\right)^2}\\
    A^t = \frac{  \|{\bf\widehat{ x}^t}\|_2^2}{\Delta N}\,; \quad
    B_i^t = \frac{1}{\Delta\sqrt{N}}{\bf J_i}\cdot {\bf \widehat{x}^{t}} - \frac{1}{N\Delta} \widehat{x}_i^{t-1}\sum_k \sigma_k^t
\end{cases}
\end{equation}
where $\alpha(t)$ and $\beta(t)$ are the functions defining the interpolant process, fixed at the start, and ${\bf Y}_t$ is the value of the noisy observation at time $t$.

\paragraph*{Pinning measure:}
When considering the pinning measure~(\ref{pinning-measure-app}), the AMP equations are only a slight variation of the ones presented for the flow-based case.

Specifically, in autoregressive-based sampling we choose a fraction $\theta$ of the variables, for which we fix $\widehat{x}_i^{t} = [{\bf x}_0]_i,\; \sigma_i^{t} = 0$; this is due to the fact that their posterior means are completely polarized on the solution. For the rest of the variables, a fraction $1-\theta$, the AMP equations are exactly the ones for diffusion in Eq.~(\ref{eq:AMP_SK_diff}), at $\gamma=0$.

The resulting algorithm is thus
\begin{equation}\label{eq:AMP_SK_deci}
\begin{cases}
    &\widehat{x}_i^{t+1} = 
    \begin{cases}
        [{\bf x}_0]_i & \text{if } i \in {S_{\theta}} \\
        \frac{\rho \tanh(B_i^t)}{\rho + \frac{(1-\rho)\exp(A^t/2)}{\cosh(B_i^t)}}, & \text{otherwise}
    \end{cases}\,,\quad\sigma_i^{t+1} = 
    \begin{cases}
        0 & \text{if } i \in {S_{\theta}} \\
        \rho \frac{\rho + (1-\rho)e^{A^t/2}\cosh(B_i^t)}{\left( \rho \cosh(B_i^t) + (1-\rho) \exp(A/2)\right)^2} & \text{otherwise}\\
    \end{cases}\\
    &A^t = \frac{  \|{\bf\widehat{ x}^t}\|_2^2}{\Delta N}\,; \quad
    B_i^t = \frac{1}{\Delta\sqrt{N}}{\bf J_i}\cdot {\bf \widehat{x}^{t}} - \frac{1}{N\Delta} \widehat{x}_i^{t-1}\sum_k \sigma_k^t
\end{cases}
\end{equation}
\subsubsection{State Evolution equations}
The advantage of AMP is that it can be rigorously tracked by the State Evolution equations \cite{bolthausen2014iterative,rangan2012iterative,donoho2010message} that turn out to be nothing but the fixed point equations of the associated replica free entropy. Again, we can use the generic results reported in \cite{lesieur2017constrained} to get
\begin{eqnarray}
    m^{t+1} = \mathbb{E}_{x_0,w}\left[ f_{\rm in}\left(\frac{m^t}{\Delta},\frac{m^t}{\Delta} x_0 + \sqrt{\frac{m^t}{\Delta}}w \right) x_0 \right]
\end{eqnarray}

where $x_0\sim P_X$, $w\sim\mathcal{N}(0,1)$ and $f_{\rm in}$ is the input channel and depends on the specific problem.

\paragraph*{Tilted measure:} For the \textbf{tilted measure}~(\ref{tilted-measure-app}) we get

\begin{equation}
\begin{split}
    f_{\rm in}(A,B;x_0) &= \frac{\int \dd x x P_X(x)\exp(\gamma^2xx_0 + \gamma wx - \gamma^2x^2/2) \exp(Bx - Ax^2/2)}{\int \dd x P_X(x)\exp(\gamma^2xx_0 + \gamma wx - \gamma^2x^2/2) \exp(Bx - Ax^2/2)}
    \\
    &=\frac{\rho\tanh(B + \gamma w + \gamma^2 x_0)}{\rho + \frac{(1-\rho)\exp\left((A+\gamma^2)/2\right)}{\cosh(B + \gamma w + \gamma^2 x_0)}}
\end{split}
\end{equation}

which leads to the state evolution equations:
\begin{equation}
    \C^{t+1} = \rho^2 \E_w\left[\frac{\tanh(\widetilde{\C}^{t} + \sqrt{\widetilde{\C}^{t}} w)}{\rho + \frac{(1-\rho)\exp(\widetilde{\C}^{t}/2)}{\cosh(\widetilde{\C}^{t} + \sqrt{\widetilde{\C}^{t}} w)}}\right]\,,\quad \widetilde{\C}^{t} \equiv \frac{\C^t}{\Delta} + \gamma^2\,,\quad w\sim \Ncl(0,1)
\end{equation}
\paragraph*{Pinning measure:} In the same way, the \textbf{pinning measure}~(\ref{pinning-measure-app}) is associated to
\begin{equation}
\begin{split}
    f_{\rm in}(A,B;x_0) &= 
    \begin{cases}
        \frac{P_X(x_0)x_0\exp\left( Bx_0 - Ax_0^2/2\right)}{P_X(x_0)\exp\left( Bx_0 - Ax_0^2/2\right)} & \text{with probability } \theta \\
        \frac{\int \dd x x P_X(x) \exp(Bx - Ax^2/2)}{\int \dd x P_X(x) \exp(Bx - Ax^2/2)} & \text{with probability } 1-\theta
    \end{cases} \\
    &= 
    \begin{cases}
        x_0 & \text{with probability } \theta \\
        \rho\tanh(B)/(\rho + (1-\rho)\exp\left(A/2\right)/\cosh(B)) & \text{with probability } 1-\theta
    \end{cases}
\end{split}
\end{equation}
and consequently to the fixed point equations
\begin{equation}
    \frac{\C^{t+1}}{\rho} = \theta + (1-\theta) \E_z\left[\frac{\rho\tanh(\widetilde{\C}^{t} + \sqrt{\widetilde{\C}^{t}} w)}{\rho + \frac{(1-\rho)\exp(\widetilde{\C}^{t}/2)}{\cosh(\widetilde{\C}^{t} + \sqrt{\widetilde{\C}^{t}} w)}}\right]\,,\quad \widetilde{\C}^{t} \equiv \frac{\C^t}{\Delta}\,,\quad w\sim \Ncl(0,1)\,.
\end{equation}

\subsubsection{Phase diagrams}
\begin{figure}[t!]
    \centering
    \includegraphics[width=0.49\textwidth]{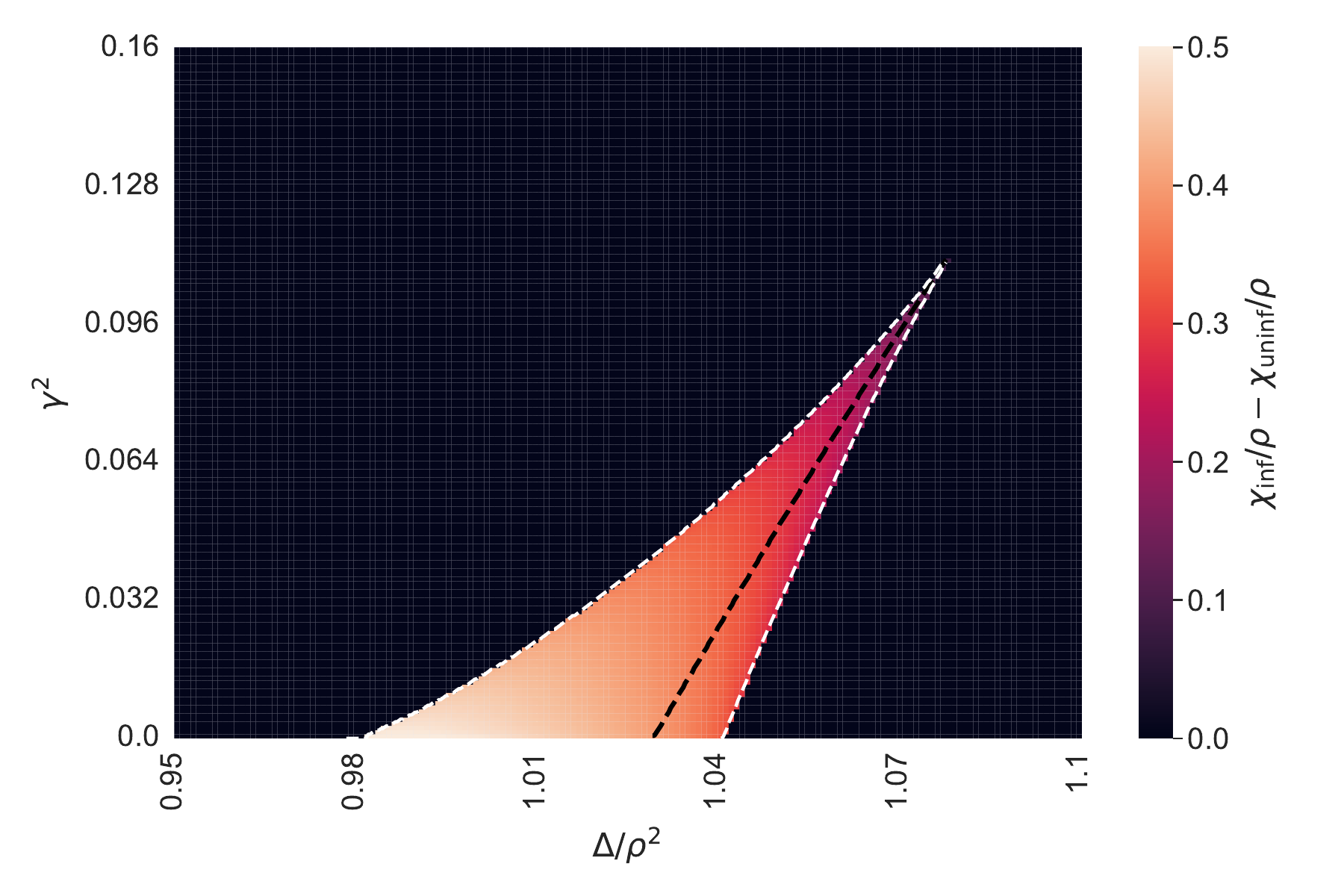}
    \includegraphics[width=0.49\textwidth]{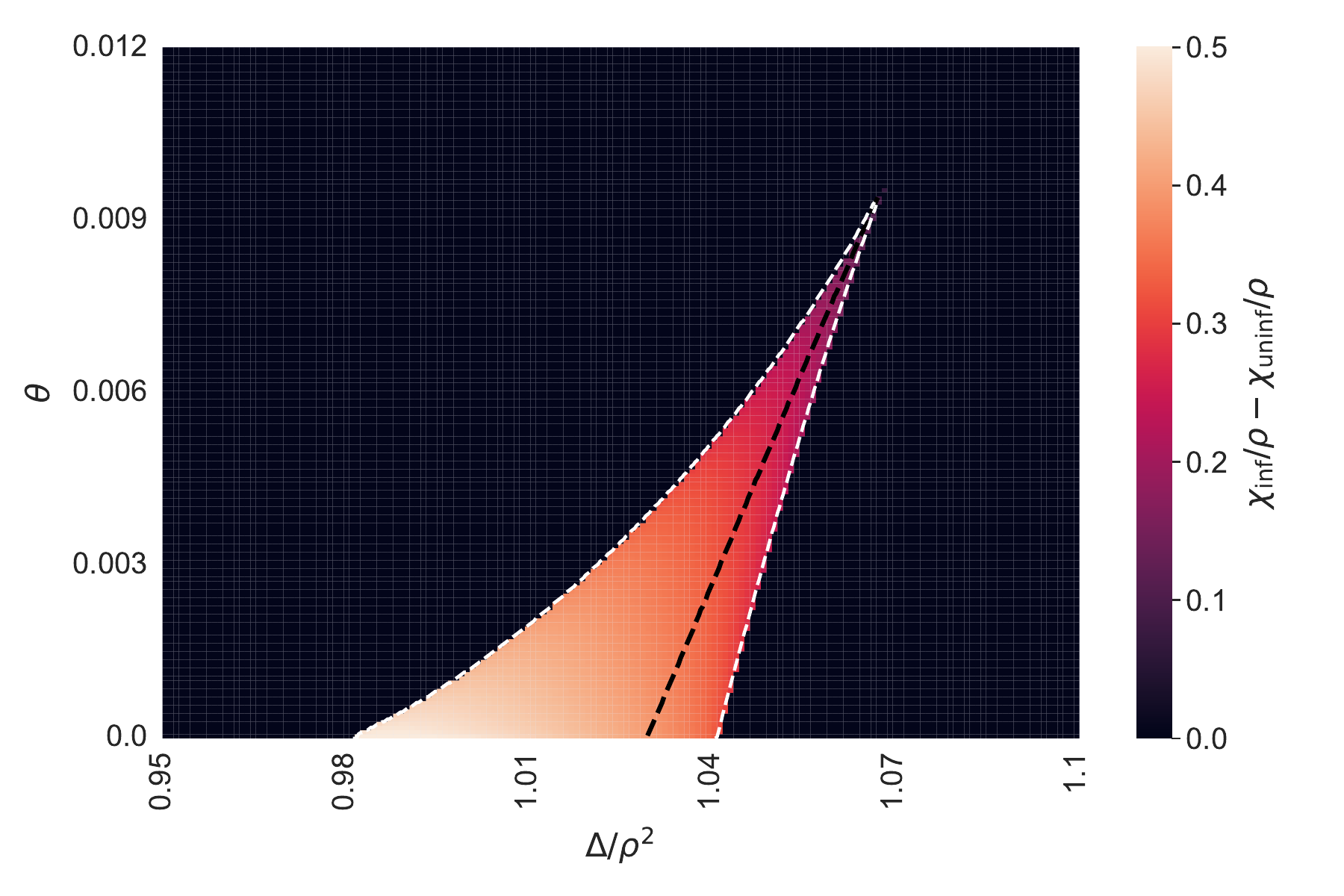}
    \caption{Phase diagrams for flow-based sampling (left) and autoregressive-based sampling (right) for the \textit{sparse rank-one} model, with Rademacher-Bernoulli prior and sparsity $\rho=0.08$. On the x-axis we put the rescaled signal-to-noise-ratio $\Delta/\rho^2$ and on the y-axis the ratio $\gamma^2 = \alpha^2/\beta^2$ (left) and the decimated ratio $\theta$ (right). 
    We compute the order parameter $\chi/\rho$, defined in~(\ref{eq:Chi-app}), both from an uninformed and an informed initialization, and we plot the difference between the two. The dashed white lines are the \textit{spinodal lines}, while the dashed black one is the \textit{IT threshold}. Note that for the flow-based case (left panel), we show explicitly in Fig.~\ref{fig:CUT-1.05} and Fig.~\ref{fig:CUT-0.98} the behaviour of the free entropy functional for $\Delta/\rho^2=1.05$ and $\Delta/\rho^2=0.98$.     
    Here in both plots (left and right) we have that the dynamical transition is at $\Delta_d/\rho^2\approx 1.041$, the IT/Kauzmann transition is at $\Delta_{IT}/\rho^2 \approx 1.029$, while the tri-critical points are at $\Delta_{\rm tri}/\rho^2\approx1.08$ for flow-based and $\Delta_{\rm tri}/\rho^2\approx1.069$ for autoregressive based sampling.}

    \label{fig:RB_SK}
\end{figure}
In Fig.~\ref{fig:RB_SK} we present the phase diagrams for the sparse rank-one matrix factorization problem, choosing $\rho=0.08$ as value for sparsity. We remind that $\rho$ must be small enough to observe a first-order phenomenology \cite{krzakala2016mutual}.

The section of the parameter space displayed is the same as in the plots presented in the main text, but here we display directly the difference $\C_{\rm inf} - \C_{\rm uninf}$, so that the coloured zones of the plots are the ones displaying multiple fixed points, as opposed to the black ones. We furthermore draw as white dashed lines the spinodal points, and as a black dashed line the IT threshold, both defined at the beginning of Appendix~\ref{sec:PhD-app}.
For the flow-based plot, we also show explicitly in Fig.~\ref{fig:CUT-1.05} and Fig.~\ref{fig:CUT-0.98} the behaviour of the free entropy functional for $\Delta/\rho^2=1.05$ and $\Delta/\rho^2=0.98$. For Fig.~\ref{fig:CUT-1.05} the first order transition is apparent, while for Fig.~\ref{fig:CUT-0.98} the transition is found to be continuous.

\begin{figure}[ht!]
    \centering
   
    \textbf{Spike Wigner model: $\Phi_{\rm RS}(\chi)$ for $ \Delta_d/\rho^2<\Delta/\rho^2 = 1.05 <{\Delta}_{\rm tri}/\rho^2$ }\par\medskip
    \vspace*{-0.25cm}
\includegraphics[width=0.24\textwidth]{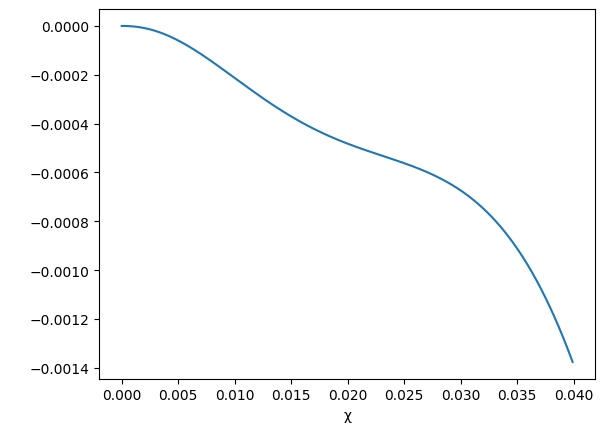}    
    \includegraphics[width=0.24\textwidth]{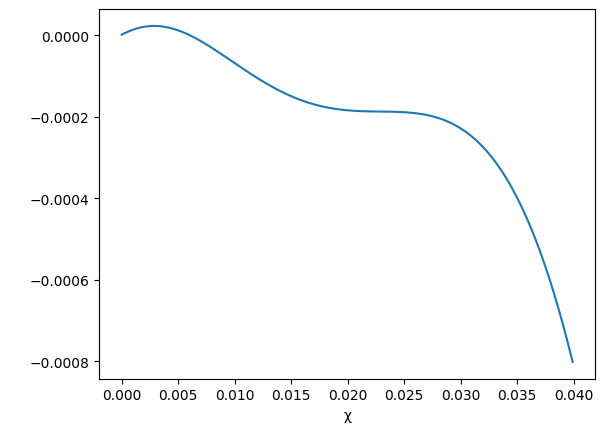}  
    \includegraphics[width=0.24\textwidth]{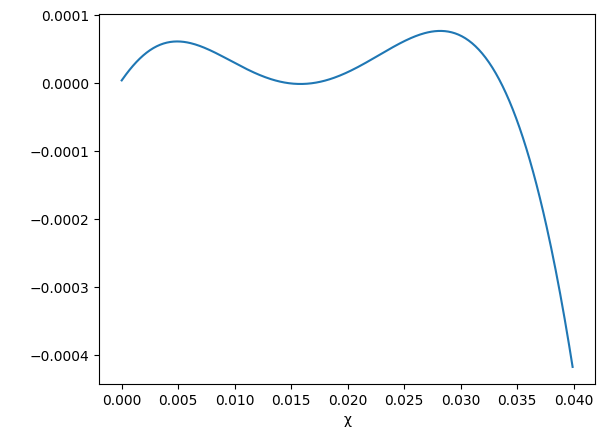}
    \includegraphics[width=0.24\textwidth]{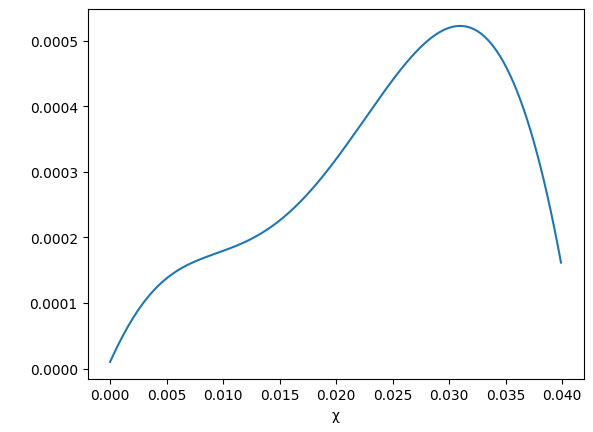}
      \text{(a) $\gamma^2=0$}\hspace{2.cm} \text{(b) $\gamma^2=0.03$} \hspace{2.cm}\text{(c) $\gamma^2=0.05$} \hspace{1.5cm} \text{(d) $\gamma^2=0.08$}  
        \vspace{0.5cm}
          \caption{The free entropy function $\Phi_{\rm RS}(\chi)$ in the region where the flow-based model fails because of the jump around panel(c). }
 \label{fig:CUT-1.05}
\end{figure}
\begin{figure}[ht!]
    \centering   
    \textbf{Spike Wigner model: $\Phi_{\rm RS}(\chi)$ for $ \Delta/\rho^2 = 0.98 <\Delta_{\rm alg}/\rho^2$ }\par\medskip
    \vspace*{-0.25cm}
\includegraphics[width=0.24\textwidth]{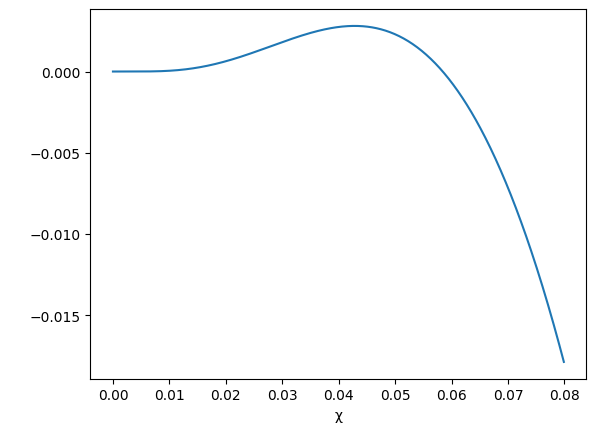}   
    \includegraphics[width=0.24\textwidth]{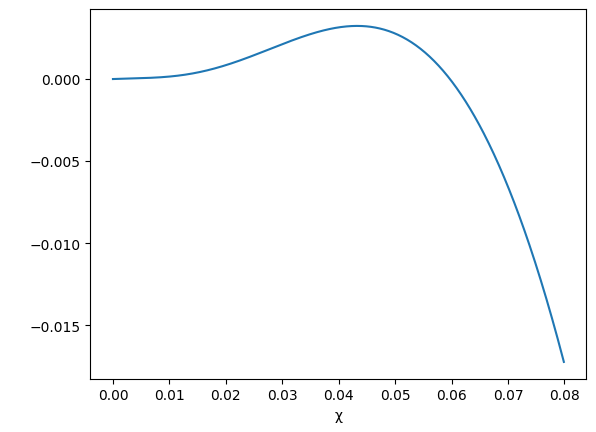}
    \includegraphics[width=0.24\textwidth]{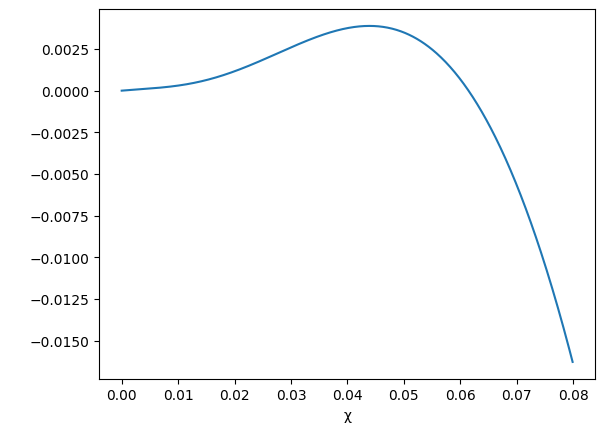}
    \includegraphics[width=0.24\textwidth]{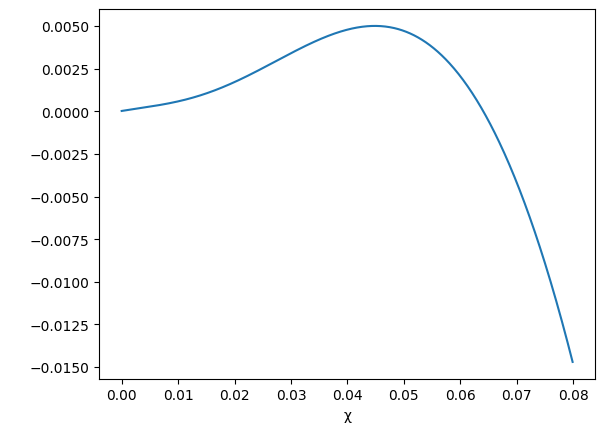}
      \text{(a) $\gamma^2=0$}\hspace{2.cm} \text{(b) $\gamma^2=0.02$} \hspace{2.cm}\text{(c) $\gamma^2=0.05$} \hspace{1.5cm} \text{(d) $\gamma^2=0.1$}  
        \vspace{0.5cm}
          \caption{The free entropy function $\Phi_{\rm RS}(\chi)$ in the region where the flow-based model succeeds as the position of the maxima is always unique.}
           \label{fig:CUT-0.98}
\end{figure}

\paragraph*{Autoregressive networks vs Flows} As we can see from the plots, for this model the tri-critical point for flow-based sampling is at $\Delta_{\rm tri}/\rho^2 \approx 1.08$, while for autoregressive-based sampling is at $\Delta_{\rm tri}/\rho^2 \approx 1.069$, meaning that the gap with MCMC and Langevin sampling is smaller in this latter case. In other words, this means that there is a range of values of inverse-SNR $\Delta$ for which autoregressive networks based sampling (along with MCMC and Langevin) is efficient, while flow-based sampling is not.

This appears not to be specific to this particular value of $\rho$, but in all the range of sparsity in which the model presents a first order phase transition, autoregressive-based sampling appears to be more efficient than flow-based sampling, in the sense explained above, as shown in Fig.~\ref{fig:TRIvsRHO}.

\begin{figure}[th!]
    \centering
    \includegraphics[width=0.7\textwidth]{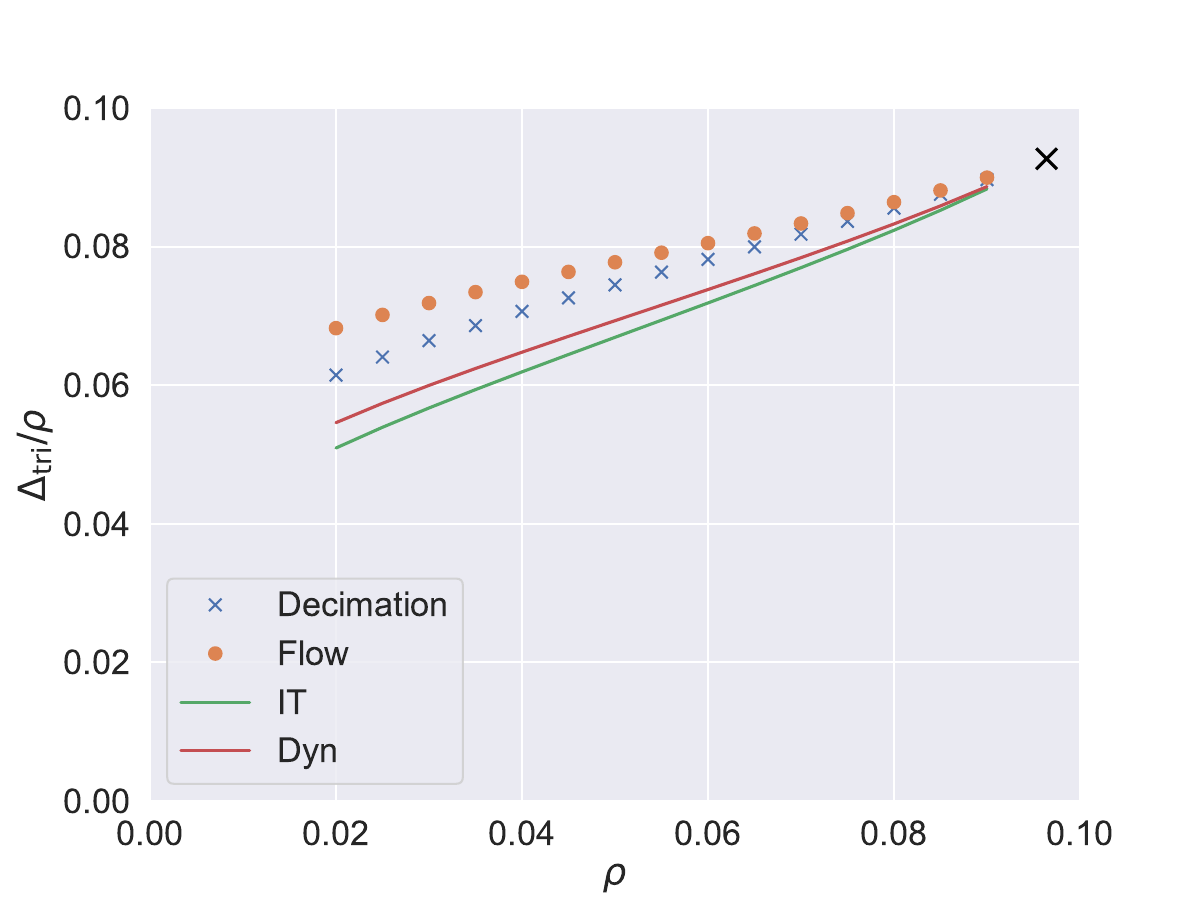}
    \caption{\textbf{Flows vs autoregressive in the sparse rank-one problem.} We plot the values of the tri-critical point for flow-based (orange dots) and autoregressive-based (blue crosses) sampling on the sparse rank-one model, when varying the sparsity parameter $\rho$. Specifically, we plot $\Delta_{\rm tri}/\rho$ for the two cases, comparing them also to the Dynamical transition value $\Delta_d/\rho$ (red line) and the IT transition value $\Delta_{\rm IT}/\rho$ (green line). Finally, we also put a black cross at the point $(\rho_{\rm max},\Delta_{\rm max}/\rho_{\rm max})$, taken from \cite{lesieur2017constrained}, which corresponds to the maximum value of $\rho$ at which there is a first order phase transition.}
    \label{fig:TRIvsRHO}
\end{figure}

\paragraph*{The easy phase (and a subtle point)} For the expert reader, it is worth pointing a subtle difference between what happens at $\Delta=\rho^2$ and $\Delta=\Delta_{\rm alg}<\rho^2$ (see also \cite{lesieur2017constrained}). Indeed, we use here the definition of the easy phase as the absence of a metastable maxima that is trapping the dynamics. In other words, the state evolution of AMP initialized at the uninformed fixed point finds the global maxima. This is indeed what is going on for $\Delta=\Delta_{\rm alg}$, as illustrated in Figure~\ref{fig:CUT-0.98}.

However, it is worth mentioning that already below $\Delta=\rho^2$, the fixed point at zero is unstable, so that there are two fixed point: the "correct" one at large $\chi$, and the one found by AMP at low but non-zero $\chi$. This phenomenon, usually called the Baik-Ben Arous-Peche (BBP) transition, is illustrated in Figure~\ref{fig:ALGvsBBP}. In this phase, while AMP, for instance (but also a standard spectral method such as PCA \cite{baik2005phase,lesieur2017constrained}) would find an estimator correlated with the ground truth (and thus $\chi^*>0$ even at $\gamma=0$), there is still a first-order phase transition and AMP would not be optimal. Since there is a discontinuity, the flow-based method suffers from the same problem, and fails to sample as well. 

\begin{figure}[h!]
    \centering   
    \textbf{The BBP transition  at $\gamma^2=0$:  $\Phi_{\rm RS}(\chi)$ for $ \Delta/\rho^2 = 0.99$ (left) and $\Delta/\rho^2 = 1.01$ (right)}\par\medskip
    \vspace*{-0.25cm}
\includegraphics[width=0.24\textwidth]{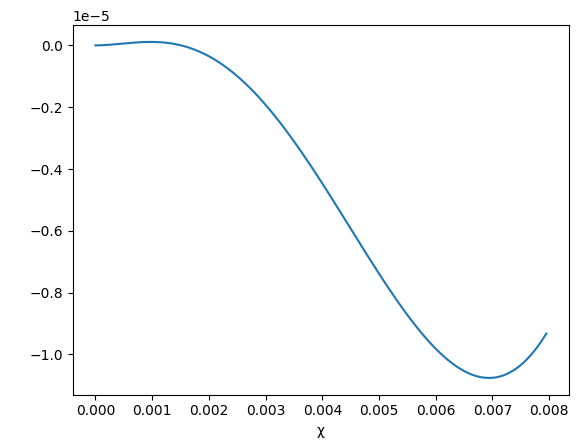}  
    \includegraphics[width=0.24\textwidth]{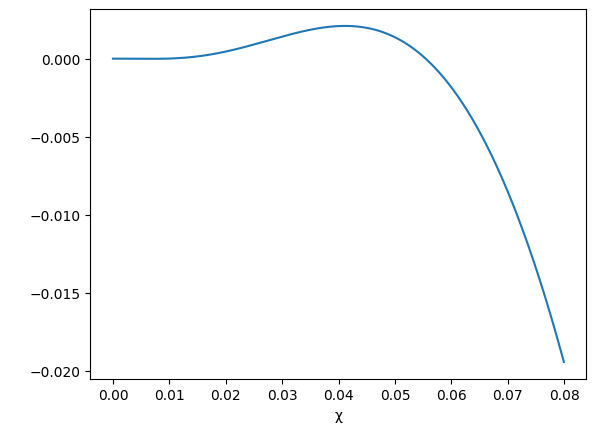}
    \includegraphics[width=0.24\textwidth]{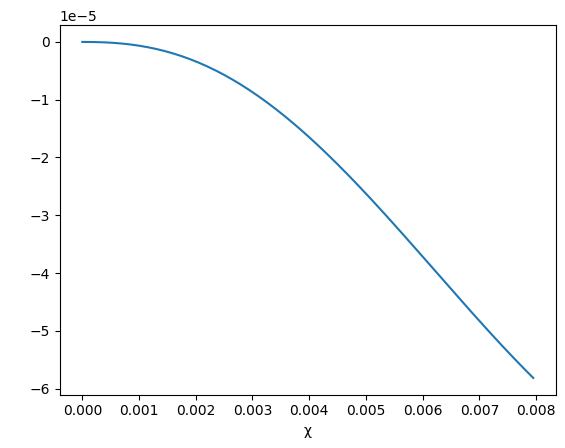}
    \includegraphics[width=0.24\textwidth]{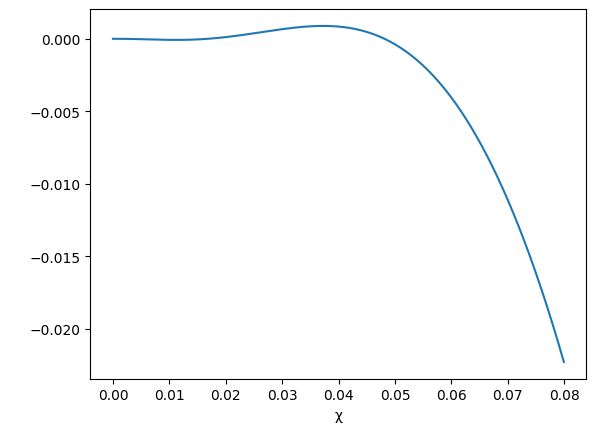}
      \text{(a) $\Delta/\rho^2=0.99$ (zoom)}\hspace{0.5cm} \text{(b) $\Delta/\rho^2=0.99$ (unzoom)}\hspace{0.5cm}\text{(c) $\Delta/\rho^2=1.01$ (zoom)} \hspace{0.5cm} \text{(d) $\Delta/\rho^2=1.01$ (unzoom)} 
        \vspace{0.5cm}
          \caption{Illustration of the BBP\cite{baik2005phase} transition around $\Delta=\rho^2$ where the spurious minima stop to be at zero (see the difference between the zoomed (a) and (c)), but there the correct non-spurious minima is still at larger value of $\chi$ (see the unzoomed (b) and (d)). This is still a hard phase for inference and sampling, and the problem remains so until $\Delta<\Delta_{\rm alg}$}
           \label{fig:ALGvsBBP}
\end{figure}

\subsection{Ising \texorpdfstring{$p$}{p}-spin model}
We consider now the $p$-spin model (here with $p=3$), which is one of the most important statistical physics problems in spin glass theory. First, let us look at the Ising version, defined by the following Hamiltonian:
\begin{align}
    \label{def:pspin-isi}
    {\cal H}({\bf x}) = -\frac{\sqrt{3}}{N} \sum_{i<j<k} J_{ijk}x_ix_jx_k \,,
    \end{align}
where $J_{ijk} \sim \mathcal{N}(0,1)$ and the variables $x_i$ are constrained to be $\pm1$. This version of the model was introduced in \cite{derrida1980random}, and solved with the replica method in \cite{gross1984simplest}. It is hard to underestimate its importance in the spin glass and glass theory as the prototype of the mean-field random first-order model \cite{kirkpatrick1987dynamics,kirkpatrick1987connections,biroli2008thermodynamic,biroli2012random,charbonneau2017glass}.

For temperatures higher than the spin glass, or Kauzmann, temperature, the model can be proven \cite{perry2016statistical,Lesieur2017Statistical,jagannath2020statistical} to be contiguous to its ``planted version", the tensor factorization problem \cite{Richard2014}. This is nothing but the tensor generalization of the former matrix estimation problem of the preceding section. We shall report the derivations presented in \cite{Lesieur2017Statistical}, which considers the planted version of the problem, but that (thanks to contiguity) also describe the unplanted model, and the mapping between the two can be done using $\Delta = \frac{2}{3}T^2$.

As before, from a Bayesian perspective, the problem boils down to sampling the following posterior
\begin{equation}
    P_0({\bf x}) \propto  \prod_i P_X(x_i) \prod_{i<j<k}e^{\frac{\sqrt{3}\beta}{N}J_{ijk}x_ix_jx_k }\,,
    \label{eq:Ppspin-app}
\end{equation}
where the prior $P_X(x) = \delta_{x,-1}/2 + \delta_{x,+1}/2$ constrains the variables to be $\pm1$.

As we have reminded for the previous model, when considering the tilted measure~(\ref{tilted-measure-app}) we will exploit again the fact that considering an equilibrium configuration ${\bf x}_0$ is equivalent to take as ${\bf x}_0$ the planted configuration, as long as we are beyond the Spin Glass temperature.

We shall thus consider the following \textbf{tilted measure} for diffusion and flow-based sampling:
\begin{equation}
    P_{\gamma}({\bf x}) = \frac{1}{Z_{\gamma}}  \left(\prod_i P_X(x_i)\right)e^{\gamma(t)^2\langle {\bf x}, {\bf x}_0 \rangle + \gamma(t)\langle z,{\bf x} \rangle - \frac{\gamma(t)^2}{2}\|{\bf x}\||^2} \prod_{i<j<k}e^{\frac{\sqrt{3}\beta}{N}J_{ijk}x_ix_jx_k }
\end{equation}
As before, we can notice that with respect to the original measure in~(\ref{eq:Ppspin-app}), the tilting adds a field in the planted direction, a random field and a normalization factor depending on the l2 norm.

For the \textbf{pinned measure} defined in~(\ref{pinning-measure-app}) the original problem becomes (denoting as $S_{\theta}$ the pinned list):
\begin{equation}
    P_{\theta}({\bf x}) = \frac{1}{Z_{\theta}}  \left(\prod_{i \notin S_\theta} P_X(x_i)\right)\left(\prod_{i \in S_\theta}  \delta(x_i-x_i^*)\right) \prod_{i<j<k}e^{\frac{\sqrt{3}\beta}{N}J_{ijk}x_ix_jx_k }
\end{equation}

\subsubsection{Replica free entropy}
The replica formula for this model has been studied extensively in the literature, and can also be proven rigorously, see \cite{Lesieur2017Statistical}.
Again, the asymptotic free entropy is given by the maximum of the so-called replica symmetric potential \cite{mezard1987spin} :
\begin{eqnarray}\label{eq:PhiRS_IsiT}
    \frac 1N {\mathbb E}_{{\bf x}_0,{\bf z},J} \log Z_{\gamma} &\xrightarrow[N \to \infty]{}& {\rm argmax  }\; \Phi_{RS}(m)\\
    \Phi_{RS}(m) &=& \mathbb{E}_{w,x_0}\left[ \log Z_x\left(\frac{m^2}{\Delta},\frac{m^2}{\Delta} x_0 + \sqrt{\frac{m^2}{\Delta}}w \right) \right] - \frac{m^3}{3\Delta}
\end{eqnarray}
where $m$ is the order parameter of the problem, $x_0\sim P_X$, $w\sim \mathcal{N}(0,1)$ while $Z_x$ depends on the specific measure considered. The same is valid for the pinned measure, as long as one substitutes $Z_{\gamma}$ with $Z_{\theta}$.

In the following, since we will be interested on the unplanted model, we will use the temperature $T = \sqrt{3\Delta/2}$ as signal-to-noise parameter.

\paragraph*{Tilted measure: }
In the case of the \textbf{tilted measure}~(\ref{tilted-measure-app}) we have
\begin{equation}
\begin{split}
    Z_x(A,B;x_0) &= \int \dd x P_X(x)\exp(\gamma^2xx_0 + \gamma wx - \gamma^2x^2/2) \exp(Bx - Ax^2/2)\\  
    &=e^{ -(A+\gamma^2)/2}\cosh(B + \gamma w + \gamma^2 x_0)
\end{split}
\end{equation}
that can be put into Eq.~(\ref{eq:PhiRS_IsiT}) to get
\begin{equation}
    \Phi_{\rm RS}(\C) = -\frac{\widetilde{\C}}{2} + \E_w \left[ \log\cosh\left(\widetilde{\C} +\sqrt{\widetilde{\C}}w\right)\right] - \frac{\C^3}{2T^2}\,, \quad \widetilde{\C} = \frac{3\C^2}{2T^2} + \gamma^2\,.
\end{equation}

\paragraph*{Pinning measure:} Meanwhile, considering the \textbf{pinning measure}~(\ref{pinning-measure-app}) leads to
\begin{equation}
\begin{split}
    Z_x(A,B;x_0) &= 
    \begin{cases}
        \int \dd x P_X(x)\delta_{x,x_0} \exp(Bx - Ax^2/2) & \text{with probability } \theta \\
        \int \dd x P_X(x) \exp(Bx - Ax^2/2) & \text{with probability } 1 - \theta
    \end{cases}
    \\  
    &=
    \begin{cases}
        P_X(x_0)\exp\left( Bx_0 - Ax_0^2/2\right) & \text{with probability } \theta \\
        \exp(-A/2)\cosh(B) & \text{with probability } 1 - \theta
    \end{cases}
\end{split}
\end{equation}
which in turn gives
\begin{equation}
\begin{aligned}
    \Phi_{\rm RS}(\C) = \frac{2\theta-1}{2} \widetilde{\C}+(1-\theta) \E_w \left[ \log \cosh\left(\widetilde{\C} +\sqrt{\widetilde{\C}}w\right) \right] - \frac{\C^3}{2T^2}\,, \quad \widetilde{\C} = \frac{3\C^2}{2T^2}\,.
\end{aligned}
\end{equation}

\subsubsection{Message-passing algorithm}

As for the SK model and the sparse rank-one matrix factorization problem, the AMP iterates and the associated
Thouless-Anderson-Palmer (TAP) equations \cite{thouless1977solution} have been widely studied for the $p$-spin models \cite{crisanti1992spherical}. Here we consider the formalism of \cite{Lesieur2017Statistical}, which presents the AMP iterates for the Spike Tensor model, and we use contiguity to derive equations valid for the unplanted model with the tilting (or pinning) field. We remind again that the mapping between the two models is given by $\Delta=\frac{2T^2}{3}$.

\paragraph*{Tilted measure:}
Due to contiguity with the planted model and the equivalence of the tilted measure and an associated measure with a planted field described in Section~\ref{sec:Plant-app}, the AMP iterations for the tilted measure can be obtained through the ones for the associated inference problem described in \cite{lesieur2017constrained}.This yields to

\begin{equation}\label{eq:AMP_IsiT_diff}
\begin{cases}
    \widehat{x}_i^{t+1} = \tanh(B_i^t + \frac{\alpha(t)}{\beta(t)^2}[{\bf Y}_t]_i), \quad
    \sigma_i^{t+1} =\cosh(B_i^t + \frac{\alpha(t)}{\beta(t)^2}[{\bf Y}_t]_i)^{-2}\\
    B_i^t = \frac{\sqrt{3}\beta}{ N}\sum_{j<k} J_{ijk} \widehat{x}_j^{t}\widehat{x}_k^{t} - \frac{3}{N}\beta^2 \widehat{x}_i^{t-1}{\bf \widehat{x}}^{t} \cdot {\bf \widehat{x}}^{t-1} \sum_k\sigma_k^t/N \\
\end{cases}
\end{equation}
where $\alpha(t)$ and $\beta(t)$ are the functions defining the interpolant process, fixed at the start, and ${\bf Y}_t$ is the value of the noisy observation at time $t$. 

\paragraph*{Pinning measure:} The AMP equations for the pinning measure~(\ref{pinning-measure-app}) are a slight variation of the ones just presented for the tilted measure.

Specifically, in the autoregressive-based sampling scheme for a fraction $\theta$ of the variables we fix $\widehat{x}_i^{t} = x_0,\; \sigma_i^{t} = 0$, since their posterior means are totally polarized on the solution. For the rest of the variables, a fraction $1-\theta$, the AMP equations are the same as the one presented for diffusion in Eq.~(\ref{eq:AMP_IsiT_diff}), provided that we fix $\gamma=0$.

The resulting equations are:
\begin{equation}\label{eq:AMP_IsiT_deci}
\begin{cases}
    &\widehat{x}_i^{t+1} = 
    \begin{cases}
        [{\bf x}_0]_i & \text{if } i \in {S_{\theta}} \\
         \tanh(B_i^t), & \text{otherwise}
    \end{cases}\,,\quad\sigma_i^{t+1} = 
    \begin{cases}
        0 & \text{if } i \in {S_{\theta}} \\
        \cosh(B_i^t)^{-2} & \text{otherwise}\\
    \end{cases}\\
    &B_i^t = \frac{\sqrt{3}\beta}{ N}\sum_{j<k} J_{ijk} \widehat{x}_j^{t}\widehat{x}_k^{t} - \frac{3}{N}\beta^2 \widehat{x}_i^{t-1}{\bf \widehat{x}}^{t} \cdot {\bf \widehat{x}}^{t-1} \sum_k\sigma_k^t/N
\end{cases}
\end{equation}
\subsubsection{State evolution equations}
The advantage of AMP is that it can be rigorously tracked by the State Evolution equations \cite{donoho2010message}, that can be proven to be nothing but the fixed point of the replica potential~(\ref{eq:PhiRS_IsiT}). We can use the formalism in \cite{Lesieur2017Statistical} to get

\begin{eqnarray}
    m^{t+1} = \mathbb{E}_{x_0,w}\left[ f_{\rm in}\left(\frac{(m^{t})^2}{\Delta},\frac{(m^{t})^2}{\Delta} x_0 + \sqrt{\frac{(m^{t})^2}{\Delta}}w \right) x_0 \right]
\end{eqnarray}
where $x_0 \sim P_X$, $w\sim \mathcal{N}(0,1)$ and $f_{in}$ is the input channel and depends on the specific problem. Again, we will state our results using the temperature $T = \sqrt{3\Delta/2}$.

\paragraph*{Tilted measure:} For the \textbf{tilted measure}~(\ref{tilted-measure-app}) we have

\begin{equation}
\begin{split}
    f_{\rm in}(A,B;x_0) &= \frac{\int \dd x x P_X(x)\exp(\gamma^2xx_0 + \gamma wx - \gamma^2x^2/2) \exp(Bx - Ax^2/2)}{\int \dd x P_X(x)\exp(\gamma^2xx_0 + \gamma wx - \gamma^2x^2/2) \exp(Bx - Ax^2/2)}
    \\
    &=\tanh(B + \gamma w + \gamma^2 x_0)
\end{split}
\end{equation}
which leads to the State Evolution equations
\begin{equation}
    \C^{t+1} = \E_w\left[\tanh(\widetilde{\C}^{t} + \sqrt{\widetilde{\C}^{t}} w)\right]\,,\quad \widetilde{\C}^{t} \equiv \frac{3(\C^{t})^2}{2T^2} + \gamma^2\,,\quad w\sim \Ncl(0,1)
\end{equation}

\paragraph*{Pinning measure:} In the same way, for the \textbf{pinning measure}~(\ref{pinning-measure-app}) we get
\begin{equation}
\begin{split}
    f_{\rm in}(A,B;x_0) &= 
    \begin{cases}
        \frac{P_X(x_0)x_0\exp\left( Bx_0 - Ax_0^2/2\right)}{P_X(x_0)\exp\left( Bx_0 - Ax_0^2/2\right)} & \text{with probability } \theta \\
        \frac{\int \dd x x P_X(x) \exp(Bx - Ax^2/2)}{\int \dd x P_X(x) \exp(Bx - Ax^2/2)} & \text{with probability } 1-\theta
    \end{cases} \\
    &= 
    \begin{cases}
        x_0 & \text{with probability } \theta \\
        \tanh(B) & \text{with probability } 1-\theta
    \end{cases}
\end{split}
\end{equation}
and thus the fixed point equations
\begin{equation}
    \C^{t+1} = \theta + (1-\theta) \E_w\left[\tanh(\widetilde{\C}^{t} + \sqrt{\widetilde{\C}^{t}} w)\right]\,,\quad \widetilde{\C}^{t} \equiv \frac{3(\C^{t})^2}{2T^2}\,,\quad w\sim \Ncl(0,1)
\end{equation}

\subsubsection{Phase diagrams}
The phase diagrams for the Ising $p$-spin are presented in Figure~\ref{fig:IsiT-app}, in the same style as the previous section. In this case, notably, we observe that the flow-based method is advantageous with respect to the autoregressive one, in the sense that the former shows a smaller gap $T_{\rm tri} - T_{d}$ compared to the latter. This is the opposite situation compared to the previous problem.
\begin{figure}[t]
    \centering
    \includegraphics[width=0.45\textwidth]{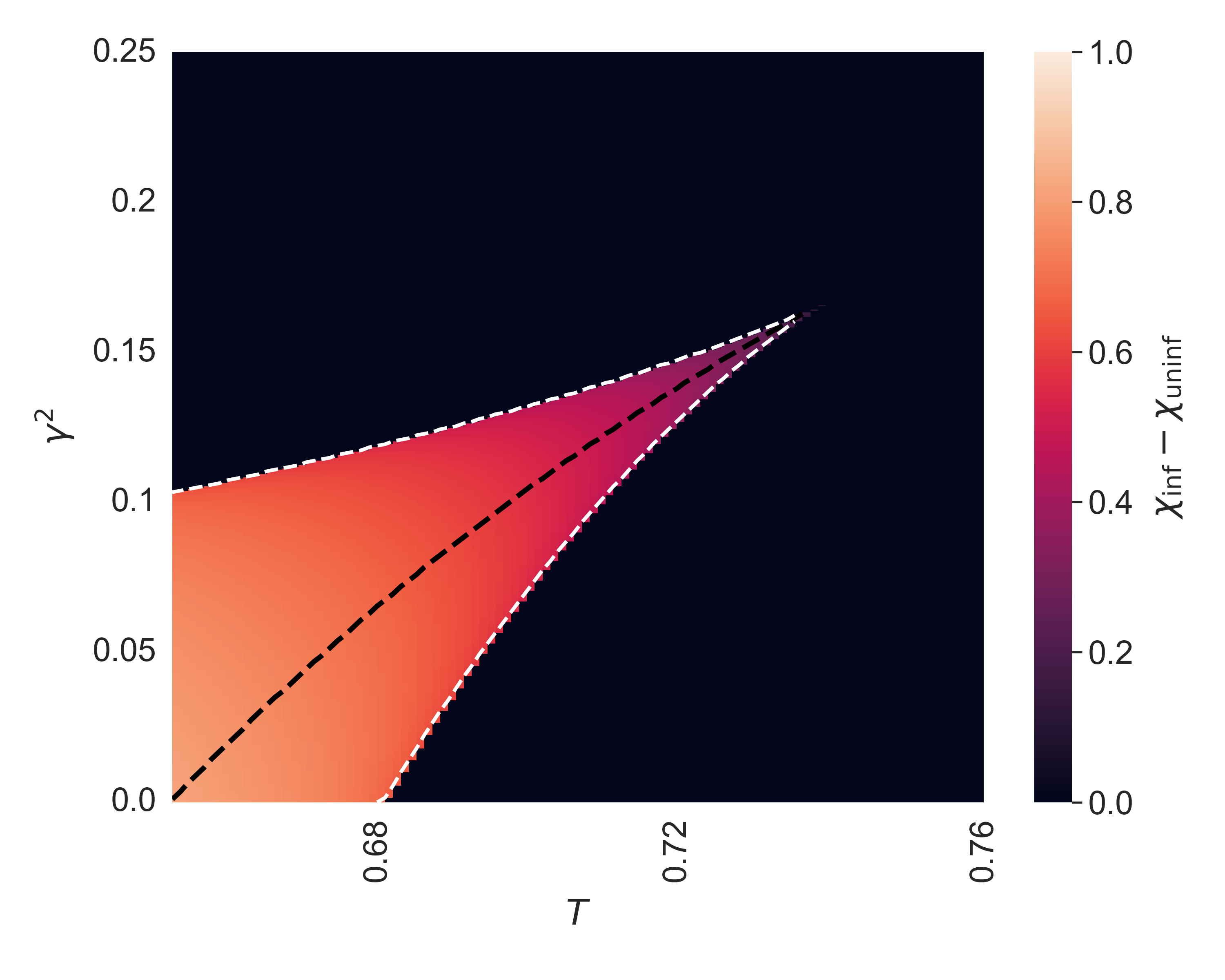}
    \includegraphics[width=0.45\textwidth]{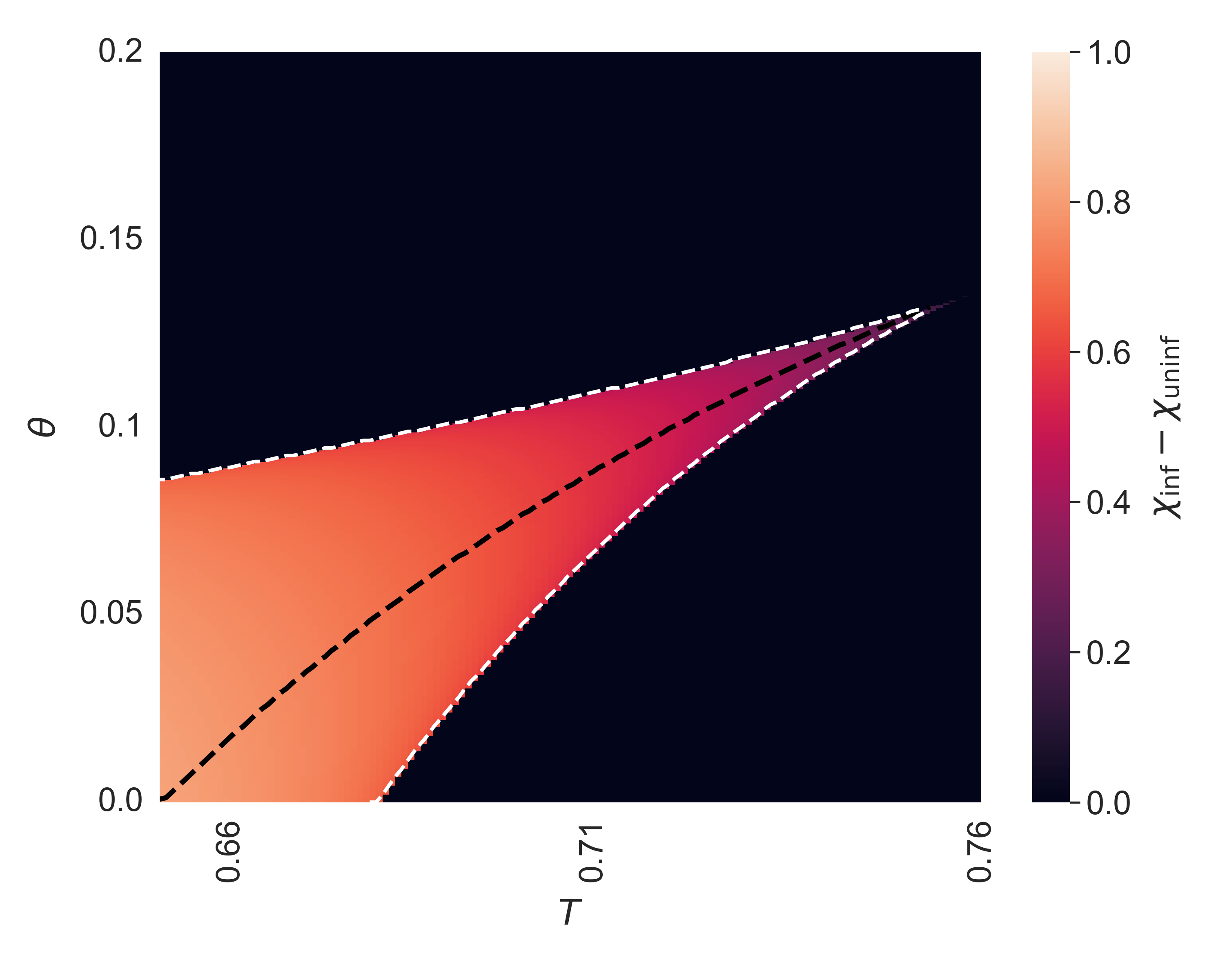}
    \caption{Phase diagrams for flow-based sampling (left) and autoregressive-based sampling (right) for the \textit{Ising p-spin} model. On the x-axis we put the temperature $T$ and on the y-axis the ratio $\gamma^2 = \alpha^2/\beta^2$ (left) and the decimated ratio $\theta$ (right). We compute the order parameter $\chi$, defined in~(\ref{eq:Chi-app}), both from an uninformed and an informed initialization, and we plot the difference between the two. The dashed white lines are the \textit{spinodal lines}, while the dashed black one is the \textit{IT threshold}, both defined at the beginning of the section. In both plots we have that the dynamical transition is at $T_d\approx 0.682$, the Kauzmann transition is at $T_K \approx 0.652$, while the tri-critical points are at $T_{\rm tri}\approx0.741$ for flow-based and $T_{\rm tri}\approx0.759$ for autoregressive-based sampling.}
    \label{fig:IsiT-app}
\end{figure}

\subsection{Spherical p-spin model}
The spherical $p$-spin model is a variation of the Ising model where the configurations are constrained to lie on the sphere ${\bf x} \in {\cal S}^{N-1}$ \cite{crisanti1992spherical}. Again, we shall use the planted version, which is asymptotically equivalent at high temperature with the unplanted problem having independent Gaussian entries. The Hamiltonian thus still reads:
\begin{align}
    \label{def:pspin-sph}
    {\cal H}({\bf x}) = -\frac{\sqrt{3}}{N} \sum_{i<j<k} J_{ijk}x_ix_jx_k \,,
    \end{align}
where $J_{ijk} \sim \mathcal{N}(0,1)$ and now $P_X(x_i)={\cal N}(0,1)$.

Same as its Ising version, for temperatures higher than the Kauzmann temperature the model can be proven \cite{Lesieur2017Statistical} to be contiguous to its ``planted version", the tensor factorization problem \cite{Richard2014}. Here we report the derivations presented in \cite{Lesieur2017Statistical}, and the mapping between the two is again $\Delta = \frac{2}{3}T^2$.

From a Bayesian point of view, the posterior measure we want to sample is the following:
\begin{equation}
    P_0({\bf x}) \propto  \prod_i P_X(x_i) \prod_{i<j<k}e^{\frac{\sqrt{3}\beta}{N}J_{ijk}x_ix_jx_k }\,,
    \label{eq:PpspinS-app}
\end{equation}
where $P_X(x) = {\cal N}(0,1)$ in high-dimension constrains the variables to be on the sphere.

We first consider the following \textbf{tilted measure}, arising in diffusion and flow-based sampling:
\begin{equation}
    P_{\gamma}({\bf x}) = \frac{1}{Z_{\gamma}}  \left(\prod_i P_X(x_i)\right)e^{\gamma(t)^2\langle {\bf x}, {\bf x}_0 \rangle + \gamma(t)\langle z,{\bf x} \rangle - \frac{\gamma(t)^2}{2}\|{\bf x}\||^2} \prod_{i<j<k}e^{\frac{\sqrt{3}\beta}{N}J_{ijk}x_ix_jx_k }
\end{equation}
Again, with respect to the original measure in Eq.~(\ref{eq:PpspinS-app}), this measure presents an additional field in the planted direction, a random field and a constant factor depending on the l2 norm.

For the \textbf{pinned measure} defined in~(\ref{pinning-measure-app}) the original problem becomes (denoting as $S_{\theta}$ the pinned list):
\begin{equation}
    P_{\theta}({\bf x}) = \frac{1}{Z_{\theta}}  \left(\prod_{i \notin S_\theta} P_X(x_i)\right)\left(\prod_{i \in S_\theta}  \delta(x_i-x_i^*)\right) \prod_{i<j<k}e^{\frac{\sqrt{3}\beta}{N}J_{ijk}x_ix_jx_k }
\end{equation}
\subsubsection{Replica free entropy:}
The replica free entropy for this model has been studied extensively in the spin-glass literature, and can also be proven rigorously, see e.g. \cite{Lesieur2017Statistical}.

As for the previous models, the asymptotic free entropy is given by the maximum of the so-called replica symmetric potential \cite{mezard1987spin} :
\begin{eqnarray}\label{eq:PhiRS_SphT}
    \frac 1N {\mathbb E}_{{\bf x}_0,{\bf z},J} \log Z_{\gamma} &\xrightarrow[N \to \infty]{}& {\rm argmax  }\; \Phi_{RS}(m)\\
    \Phi_{RS}(m) &=& \mathbb{E}_{w,x_0}\left[ \log Z_x\left(\frac{m^2}{\Delta},\frac{m^2}{\Delta} x_0 + \sqrt{\frac{m^2}{\Delta}}w \right) \right] - \frac{m^3}{3\Delta}
\end{eqnarray}
where $m$ is the order parameter of the problem, $x_0\sim P_X$, $w\sim \mathcal{N}(0,1)$ and $Z_x$ depends on the specific measure considered. The same is valid for the pinned measure, as long as one substitutes $Z_{\gamma}$ with $Z_{\theta}$.

While using the same derivation, in the following we will use $T=\sqrt{3\Delta/2}$ as signal-to-noise parameter, since we are interested in studying the unplanted model.

\paragraph*{Tilted measure: }
In the case of the \textbf{tilted measure}~(\ref{tilted-measure-app}) we have

\begin{equation}
\begin{split}
    Z_x(A,B;x_0) &= \int \dd x P_X(x)\exp(\gamma^2xx_0 + \gamma wx - \gamma^2x^2/2) \exp(Bx - Ax^2/2)\\  
    &=\sqrt{\frac{2\pi}{A+\gamma^2+1}}\exp\left( \frac{(B+\gamma^2 x_0 + \gamma w)^2}{2(A+\gamma^2+1)}\right)
\end{split}
\end{equation}
which leads to
\begin{equation}
    \Phi_{\rm RS}(\C) = \frac{\widetilde{\C}}{2} + \frac{1}{2}\log\left( \frac{2\pi}{\widetilde{\C} + 1} \right) - \frac{1}{2T^2}\C^3\,, \quad \widetilde{\C} = \frac{3\C^2}{2T^2} + \gamma^2\,.
\end{equation}
\paragraph*{Pinning measure:} Meanwhile, considering the \textbf{pinning measure}~(\ref{pinning-measure-app}) leads to
\begin{eqnarray}
    Z_x(A,B;x_0) = 
    \begin{cases}
        P_X(x_0)\exp\left( Bx_0 - \frac{A}{2}x_0^2\right) & \text{with probability } \theta \\
        \sqrt{\frac{2\pi}{A+1}}\exp\left( \frac{B^2}{2(A+1)}\right) & \text{with probability } 1 - \theta
    \end{cases}
\end{eqnarray}
which in turn gives
\begin{equation}
    \Phi_{\rm RS}(\C) = \frac{\widetilde{\C}}{2} + \frac{1-\theta}{2}\log\left( \frac{2\pi}{\widetilde{\C} + 1}\right)-\theta\log(\sqrt{2\pi }e) -\frac{\C^3}{2T^2}\,, \quad \widetilde{\C} = \frac{3\C^2}{2T^2}\,.
\end{equation}

\subsubsection{Message-passing algorithm}
As already mentioned for the Ising version of the model, the Thouless-Anderson-Palmer (TAP) equations \cite{thouless1977solution} have been widely studied for the $p$-spin models \cite{crisanti1992spherical}. We use the results of \cite{Lesieur2017Statistical}, which presents the AMP algorithm for the Spike Tensor model, and we use contiguity to derive equations valid for the unplanted model with the tilting (or pinning) field. The mapping between the two models is again given by $\Delta=\frac{2T^2}{3}$.
\paragraph*{Tilted measure:} The equivalence between the tilted and the planted measure with external field allows us to
map the AMP iterations for the tilted measure to the ones for an associated inference problem. Here we consider the formalism of \cite{Lesieur2017Statistical}, such that the resulting equations are:
\begin{equation}\label{eq:AMP_SphT_diff}
\begin{cases}
    B_i^t = \frac{\sqrt{3}\beta}{ N}\sum_{j<k} J_{ijk} \widehat{x}_j^{t}\widehat{x}_k^{t} - \frac{3}{N}\beta^2 \sigma^t\widehat{x}_i^{t-1} {\bf \widehat{x}}^{t} \cdot {\bf \widehat{x}}^{t-1} \\
     \widehat{x}_i^{t+1} = \frac{B_i^t + \frac{\alpha(t)}{\beta(t)^2}[{\bf Y}_t]_i}{3\beta^2\|\hat {\bf x^t}\|_2^2/(2 N) + \gamma^2+1} \,,     \sigma^{t+1} = \frac{1}{3\beta^2\|\hat {\bf x^t}\|_2^2/(2 N) + \gamma^2+1}\\
\end{cases}
\end{equation}
where $\alpha(t)$ and $\beta(t)$ are the functions defining the interpolant process, fixed at the start, and ${\bf Y}_t$ is the value of the noisy observation at time $t$. 
\paragraph*{Pinning measure:} 
When considering the pinning measure~(\ref{pinning-measure-app}), the AMP equations are only a slight variation of the ones presented for the flow-based case.

Specifically, in autoregressive-based sampling we choose a fraction $\theta$ of the variables, for which we fix $\widehat{x}_i^{t} = x_0,\; \sigma_i^{t} = 0$, which stems from the fact that their posterior means are completely polarized on the solution. For the rest of the variables, a fraction $1-\theta$, the AMP equations are exactly the ones reported in Eq.~(\ref{eq:AMP_SphT_diff}), provided we fix $\gamma=0$.

The resulting algorithm is the following:
\begin{equation}\label{eq:AMP_SphT_deci}
\begin{cases}
    &\widehat{x}_i^{t+1} = 
    \begin{cases}
        [{\bf x}_0]_i & \text{if } i \in {S_{\theta}} \\
         \frac{B_i^t}{3\beta^2\|\hat {\bf x^t}\|_2^2/(2 N) +1}, & \text{otherwise}
    \end{cases}\,,\quad\sigma_i^{t+1} = 
    \begin{cases}
        0 & \text{if } i \in {S_{\theta}} \\
        \frac{1}{3\beta^2\|\hat {\bf x^t}\|_2^2/(2 N) +1} & \text{otherwise}\\
    \end{cases}\\
    &B_i^t = \frac{\sqrt{3}\beta}{ N}\sum_{j<k} J_{ijk} \widehat{x}_j^{t}\widehat{x}_k^{t} - \frac{3}{N}\beta^2\widehat{x}_i^{t-1} {\bf \widehat{x}}^{t} \cdot {\bf \widehat{x}}^{t-1}\sum_k\sigma_k^t/N
\end{cases}
\end{equation}
\subsubsection{State evolution equations}
As mentioned earlier, AMP iterations possess the salient property of being able to be rigorously tracked by the State Evolution equations, that turn out to be the
fixed point of the replica potential in Eq.~(\ref{eq:PhiRS_SphT}). Here, again we follow the results presented in \cite{Lesieur2017Statistical}, such that as the Ising version of the model, we have
\begin{eqnarray}
    m^{t+1} = \mathbb{E}_{x_0,w}\left[ f_{\rm in}\left(\frac{(m^{t})^2}{\Delta},\frac{(m^{t})^2}{\Delta} x_0 + \sqrt{\frac{(m^{t})^2}{\Delta}}w \right) x_0 \right]
\end{eqnarray}
where $x_0 \sim P_X$, $w\sim \mathcal{N}(0,1)$ and $f_{in}$ is the input channel and depends on the specific problem. We also remind the mapping $T = \sqrt{3\Delta/2}$, which we will use in the following presentation.

\paragraph*{Tilted measure:} For the \textbf{tilted measure}~(\ref{tilted-measure-app}) we have

\begin{equation}
\begin{split}
    f_{\rm in}(A,B;x_0) &= \frac{\int \dd x x P_X(x)\exp(\gamma^2xx_0 + \gamma wx - \gamma^2x^2/2) \exp(Bx - Ax^2/2)}{\int \dd x P_X(x)\exp(\gamma^2xx_0 + \gamma wx - \gamma^2x^2/2) \exp(Bx - Ax^2/2)}
    \\
    &= \frac{B + \gamma^2 x_0 + \gamma w}{A + \gamma^2 + 1}
\end{split}
\end{equation}
which leads to the State Evolution equations
\begin{equation}
    \C^{t+1} = \frac{\widetilde{\C^t}}{1 + \widetilde{\C^t}}\,,\quad \widetilde{\C}^{t} \equiv \frac{3}{2T^2}(\C^{t})^2 + \gamma^2\,.
\end{equation}
\paragraph*{Pinning measure:} In the same way, for the \textbf{pinning measure}~(\ref{pinning-measure-app}) we find
\begin{equation}
\begin{split}
    f_{\rm in}(A,B;x_0) &= 
    \begin{cases}
        \frac{P_X(x_0)x_0\exp\left( Bx_0 - Ax_0^2/2\right)}{P_X(x_0)\exp\left( Bx_0 - Ax_0^2/2\right)} & \text{with probability } \theta \\
        \frac{\int \dd x x P_X(x) \exp(Bx - Ax^2/2)}{\int \dd x P_X(x) \exp(Bx - Ax^2/2)} & \text{with probability } 1-\theta
    \end{cases} \\
    &=  
    \begin{cases}
        x_0 & \text{with probability } \theta \\
        \frac{B}{A + 1} & \text{with probability } 1- \theta
    \end{cases}
\end{split}
\end{equation}
and thus the fixed point equations are
\begin{equation}
    \C^{t+1} = \theta + (1-\theta)\frac{\widetilde{\C^t}}{1 + \widetilde{\C^t}}\,,\quad \widetilde{\C}^{t} \equiv \frac{3(\C^{t})^2}{2T^2}
\end{equation}

\subsubsection{Phase diagrams}
The phase diagrams for the spherical $p$-spin are presented in Figure~\ref{fig:SphT_deci}, and we plot the same quantities as for the previous models. We observe again that the flow-based method is advantageous with respect to the autoregressive one.
\begin{figure}[t]
    \centering
    \includegraphics[width=0.45\textwidth]{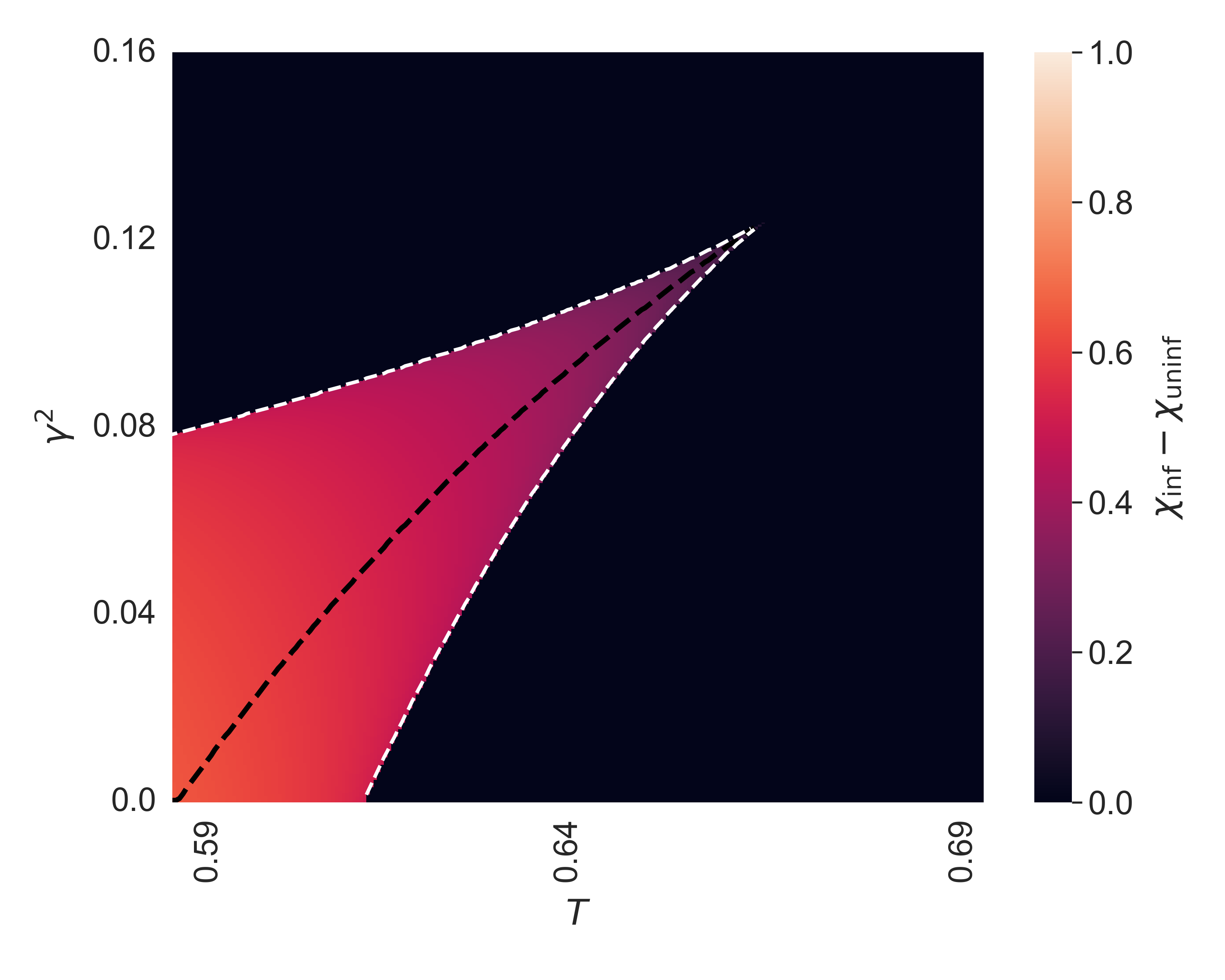}
    \includegraphics[width=0.45\textwidth]{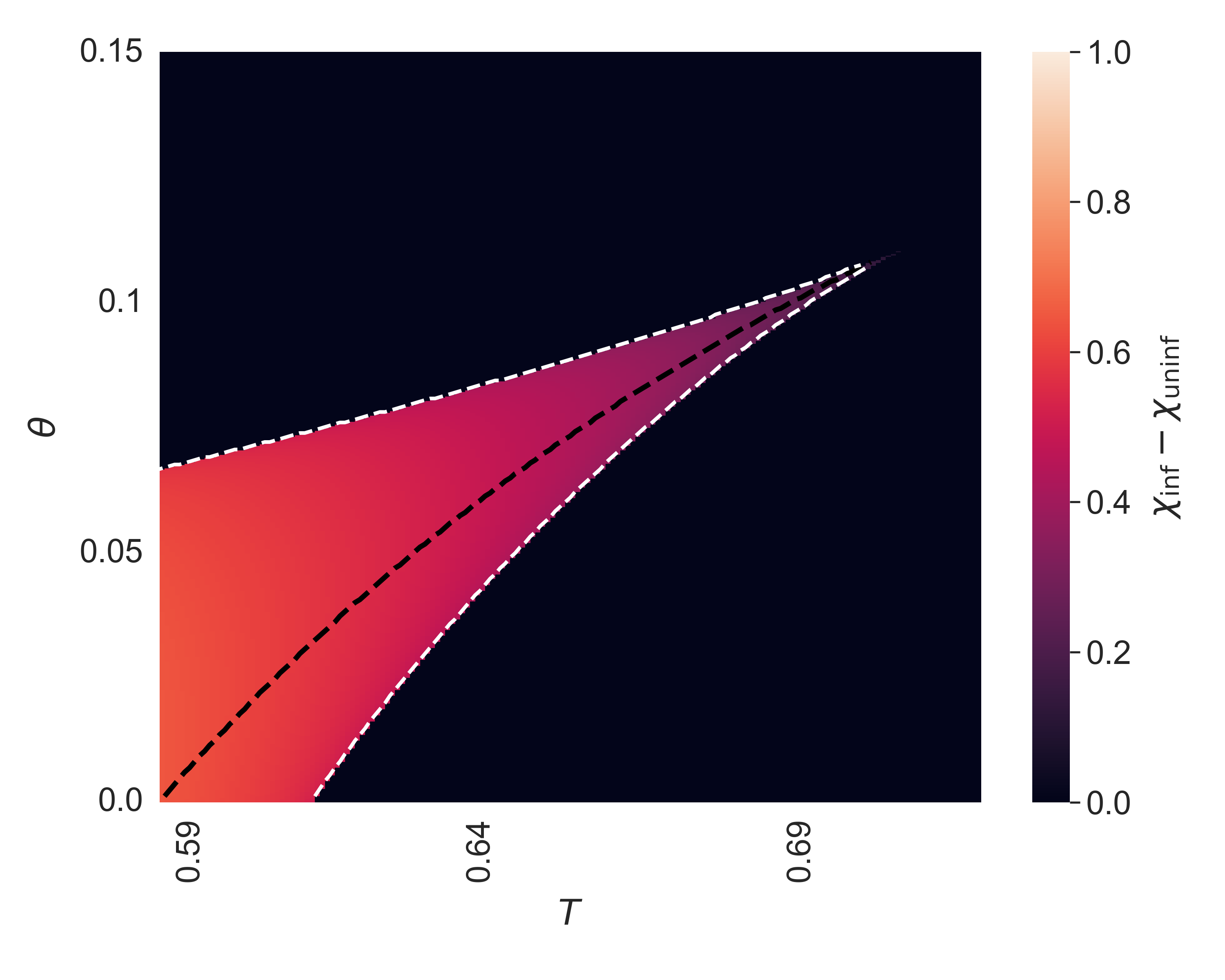}
    \caption{Phase diagrams for flow-based sampling (left) and autoregressive-based sampling (right) for the \textit{Spherical p-spin} model. On the x-axis we put the temperature $T$ and on the y-axis the ratio $\gamma^2 = \alpha^2/\beta^2$ (left) and the decimated ratio $\theta$ (right). We compute the order parameter $\chi$, defined in~(\ref{eq:Chi-app}), both from an uninformed and an informed initialization, and we plot the difference between the two. The dashed white lines are the \textit{spinodal lines}, while the dashed black one is the \textit{IT threshold}, both defined at the beginning of the section. In both plots we have that the dynamical transition is at $T_d= \sqrt{3/8}$, the Kauzmann transition is at $T_K\approx0.58$, while the tri-critical points are at $T_{\rm tri} = 2/3$ for flow-based and $T_{\rm tri}=\sqrt{1/2}$ for autoregressive based sampling.}
    \label{fig:SphT_deci}
\end{figure}
\newpage

\subsection{\label{sec:NAESAT}NAE-SAT model}

\subsubsection{Target model definition}
The $k$-hypergraph bicoloring (or $k$-NAESAT) problem is a prototypical model of constrained satisfaction problems defined on hypergraphs. An instance of the problem is determined by an hypergraph $G = (V,E)$, where $V$ is the set of $N$ vertices and $E$ the set of $M$ hyperedges, each containing exactly $k$ vertices.

Each vertex $i\in V$ is associated to an Ising-spin variable $x_i=\pm1$, while each hyperedge $a\in E$ is associated to a constraint involving the $k$ vertices entering in $a$ (in the following, we will use ${\bf x}_{\partial a}$ to indicate this set of variables).

For bicoloring, the $a$-th constraint is satisfied if there is at least one $+1$ and one $-1$ among the $k$ variables of ${\bf x}_{\partial a}$. In terms of probability distribution, this translates to
\begin{equation}\label{eq:P_NAESAT}
    P_0({\bf x}) = \frac{1}{Z(G)}\prod_{a=1}^{M}\omega({\bf x}_{\partial a}) \,,\quad \omega(x_1,\dots,x_k) = 
\begin{cases}
0 & \text{if} \;\sum_{i=1}^k x_i = \pm k\\
1 & \text{otherwise}
\end{cases}
\end{equation}
In the following, we will focus on the asymptotic limit where both $N$ and $M$ go to infinity, with constant rate $\alpha=M/N$.

This model has again the advantage to be contiguous to its planted version \cite{coja2017information}. We refer to \cite{Ding2016,Coja-Oghlan2012,coja2017information} for rigorous results, and to \cite{Castellani2003, Gabrié_2017, ricci2019typology,Budzynski_2019} for results within the cavity approach. 

\subsubsection{BP equations and Bethe free entropy:} 
To analyse the properties of the model, one can use the cavity method \cite{mezard2009information} from statistical physics, which allows to derive the BP equations for the problem (see e.g. \cite{Budzynski_2019}), which can be written in terms of \textit{cavity messages} as

\begin{equation}
    h_{i\rightarrow a} = f(\{ u_{b\rightarrow i}\}_{b\in\partial_i\backslash a})\,,\quad u_{a\rightarrow i} = g(\{ h_{j\rightarrow a}\}_{j\in\partial_a\backslash i})
\end{equation}
where
\begin{equation}\label{eq:f_NAESAT}
    f(u_1,\dots,u_d) = \frac{\prod_{i=1}^d(1+u_i) - \prod_{i=1}^d(1-u_i)}{\prod_{i=1}^d(1+u_i) + \prod_{i=1}^d(1-u_i)}
\end{equation}
is the function defining the messages going from factor nodes to variable nodes and
\begin{equation}\label{eq:g_NAESAT}
    g(h_1,\dots,h_{k-1}) = \frac{\sum_{x_1,\dots,x_k}\omega(x_1,\dots,x_k)x_k\prod_{i=1}^{k-1}(1+h_ix_i)}{\sum_{x_1,\dots,x_k}\omega(x_1,\dots,x_k)\prod_{i=1}^{k-1}(1+h_ix_i)}
\end{equation}
is the one defining messages going from variable nodes to factor nodes.

Once a fixed point of the BP equations is reached, one can compute the Free entropy from the resulting BP marginals. In its general form, it can be written as

\begin{align}
\label{BetheFreeEntropy}
\frac{1}{N}\ln Z(G) = 
\frac{1}{N}\sum_{i=1}^{N} \ln \Z_0^{\rm v} (\{u_{a \to i}\}_{a \in \di})
+ \frac{1}{N}\sum_{a=1}^{M} \ln \Z_0^{\rm c} (\{h_{i \to a}\}_{i \in \da})
- \frac{1}{N}\sum_{(i,a)} \ln \Z_0^{\rm e}(h_{i\to a}, u_{a\to i}) \ ,
\end{align}
where the last sum runs over the edges of the factor graph, and the local partition functions are defined as:
\begin{align}
\Z_0^{\rm v}(u_1,\dots,u_d) &= \sum_{x} \prod_{i=1}^d \left(\frac{1+x u_i}{2}\right) \ , \label{eq:Z0v}\\
\Z_0^{\rm c}(h_1,\dots,h_k) &= \sum_{x_1,\dots,x_k} \w(x_1,\dots,x_k) \prod_{i=1}^k \left(\frac{1+x_i h_i}{2}\right) \ , \label{eq:Z0c}\\
\Z_0^{\rm e}(h,u) &= \sum_{x} \left(\frac{1+x h}{2}\right)\left(\frac{1+x u}{2}\right) \label{eq:Z0e}\ .
\end{align}

\paragraph*{Replica symmetric cavity equations:} 
The simplest version of the cavity method implies the assumption of Replica Symmetry (RS). This is rigorously justified in our case thanks to the fact that we are considering a Bayes-Optimal model.

In such a case, the resulting self-consistent equations are relatively simpler:
\begin{align}
\label{RSeqn}
\mathcal{P}^{RS}(h) &= \sum_{d=0}^{\infty} p_d\int \left(\prod_{i=1}^{d} \dd u_i \widehat{\mathcal{P}}^{RS}(u_i) \right) 
\, \delta(h-f(u_1,\dots,u_d)) \ , \\
\nonumber
\widehat{\mathcal{P}}^{RS}(u) &= \int \left(\prod_{i=1}^{k-1} \dd h_i \mathcal{P}^{RS}(h_i) \right) \, \delta(u-g(h_1,\dots,h_{k-1})) \ . \\
\end{align}
where $\mathcal{P}^{RS}$ and $\widehat{\mathcal{P}}^{RS}$ are probability distributions defined on the messages $h$ and $u$ respectively.

Being self-consistency equations defined on probability distributions makes the computation of an analytical solution possible only in very restricted cases, while in practise one needs to use numerical techniques to find approximate solutions.

Still, \textit{population dynamics} techniques (see for example \cite{mezard2009information}) have been shown to provide very good approximations when one is interested in observables defined as average quantities over $\mathcal{O}(N)$ cavity messages.

When the density of interactions $\alpha$ becomes large, the hypothesis underlying the Replica Symmetric assumption breaks down, and we have to consider the Replica Symmetry Breaking (RSB) phenomenon~\cite{parisi1979infinite}. 

\subsubsection{1RSB and Tree reconstruction equations:} The 1RSB cavity method aims to compute the potential 
\begin{eqnarray}
    \Phi_1(m) = \lim_{N\rightarrow \infty}\frac{1}{N} \log\left(\sum_{\gamma}Z_{\gamma}^m\right)
\end{eqnarray}
where $m$ is the so-called Parisi parameter.

In their general formulation, the 1RSB equations are self-consistent equations defined on distributions over probability distributions. Focusing on the case $m=1$ allows simplifying considerably the equations, and the resulting formulas correspond to the so-called \textit{Tree Reconstruction} equations \cite{mezard2006reconstruction}. Moreover, the problem we are considering has a global spin-flip symmetry which allows us to simplify further the equations, which can be finally written as

\begin{align}
\label{1RSBeqnSimplifs}
Q_+^{(t+1)}(h) &= \sum_{d=0}^{\infty} p_d \int \left( \prod_{i=1}^{d} \dd u_i \hQ_+^{(t)}(u_i) \right) \, \delta(h-f(u_1,\dots,u_d)) \ , \\
\nonumber
\hQ_+^{(t)}(u) &= \sum_{x_1,...,x_{k-1}} \tilde{p}(x_1,...,x_{k-1}|+) 
\int \left(\prod_{i=1}^{k-1} \dd h_i Q_+^{(t)}(h_i) \right) 
\delta(u-g( x_1 h_1,\dots, x_{k-1} h_{k-1})) \ ,
\end{align}
where 
\begin{align}
\label{tildeP}
\tilde{p}(x_1, \dots ,x_{k-1}|+) = \frac{\w(x_1, \dots ,x_{k-1},+)}{\underset{x'_1,\dots,x'_{k-1}}{\sum} \w(x'_1, \dots ,x'_{k-1},+) } =
\frac{\overset{k-1}{\underset{p=0}{\sum}} \w_p \,  
\mathbb{I}\left[ \overset{k-1}{\underset{i=1}{\sum}} x_i=k-1-2p\right]}{\overset{k-1}{\underset{p=0}{\sum}} \binom{k-1}{p} \w_p } \ .
\end{align}
Similar as before, we can compute the RS free entropy for the planted problem as 
\begin{equation}
\begin{aligned}
    \Phi_{\rm RS} = & \sum_{d=0}^{\infty} p_d \int \left( \prod_{i=1}^{d} \dd u_i \hQ_+(u_i) \right) \ln \Z_0^{\rm v}(u_1,\dots,u_d) + \alpha \int \left( \prod_{i=1}^{k} \dd h_i Q_+(h_i) \right) \ln \Z_0^{\rm c}(h_1,\dots,h_k)
\\ & -\alpha k \int \dd h \dd u Q_+(h) \hQ_+(u) \ln \Z_0^{\rm e}(h,u) \ ,
\end{aligned}
\label{eq_Phi_RS}
\end{equation}
where $\Z_0^{\rm c}$, $\Z_0^{\rm e}$ and $\Z_0^{\rm v}$ are the local partitions reported in Eq.~(\ref{eq:Z0c}),~(\ref{eq:Z0e}) and~(\ref{eq:Z0v}) respectively.

Finally, the order parameter for the problem, corresponding to the definition in Eq.~(\ref{eq:Chi-app}), can be computed from the population dynamics simulations as
\begin{equation}\label{eq:chi_NAESAT}
    \C = \int \dd h Q_+(h) h^2\,.
\end{equation}
\subsubsection{Tilted and Pinned measures}
In the preceding section, we presented the classical $k$-NAESAT model and how its properties can be studied using the cavity method. Now, let us consider the tilted and pinning measures and see how this changes the previous equations.

\paragraph*{Tilted measure:} Starting with the \textbf{tilted measure}, the added tilting field results directly in the function defining the messages going from the factor nodes to the variable nodes, which is modified to
\begin{equation}
    f^{\rm tilt}(u_1,\dots,u_d) = \frac{e^{2\gamma(\gamma +z)}\underset{i=1}{\overset{d}{\prod}} (1+u_i) - \underset{i=1}{\overset{d}{\prod}} (1-u_i) }{e^{2\gamma(\gamma +z)}\underset{i=1}{\overset{d}{\prod}} (1+u_i) + \underset{i=1}{\overset{d}{\prod}} (1-u_i) }\,,\quad z\sim \Ncl(0,1)\,,
\end{equation}
while the function $g$, defined in~(\ref{eq:g_NAESAT}), is not modified. In terms of the partition function, the only local term that is modified is given by
\begin{equation}
    \Z_0^{\rm v}(u_1,\dots,u_d) = e^{\gamma(\gamma + z)}\prod_{i=1}^d \left(\frac{1+u_i}{2}\right) + e^{-\gamma(\gamma + z)}\prod_{i=1}^d \left(\frac{1-u_i}{2}\right)\,,\quad z\sim \Ncl(0,1)\,.
\end{equation}
while the other two terms remain the same.

\paragraph*{Pinning measure:} For the \textbf{pinning measure}, again the function $g$ remains the same, but now we have
\begin{equation}
    f^{\rm pinn}(u_1,\dots,u_d) = \begin{cases}
        1 & \text{with probability } \theta \\
        \frac{\underset{i=1}{\overset{d}{\prod}} (1+u_i) - \underset{i=1}{\overset{d}{\prod}} (1-u_i) }{\underset{i=1}{\overset{d}{\prod}} (1+u_i) + \underset{i=1}{\overset{d}{\prod}} (1-u_i) } & \text{with probability } 1 - \theta
    \end{cases}\,,
\end{equation}
and again the only contribution to the partition function that is modified is
\begin{equation}
    \Z_0^{\rm v}(u_1,\dots,u_d) = \begin{cases}
        \prod_{i=1}^d \left(\frac{1+u_i}{2}\right) & \text{with probability } \theta \\
        \prod_{i=1}^d \left(\frac{1+u_i}{2}\right) + \prod_{i=1}^d \left(\frac{1-u_i}{2}\right) & \text{with probability } 1 - \theta
    \end{cases}\,.
\end{equation}

\paragraph*{Phase diagrams:} In Fig.~\ref{fig:NEASAT} we present the phase diagrams for the $k$-NAESAT problem, considering the case $k=5$.
\begin{figure}
    \centering
    \includegraphics[width=0.45\textwidth]{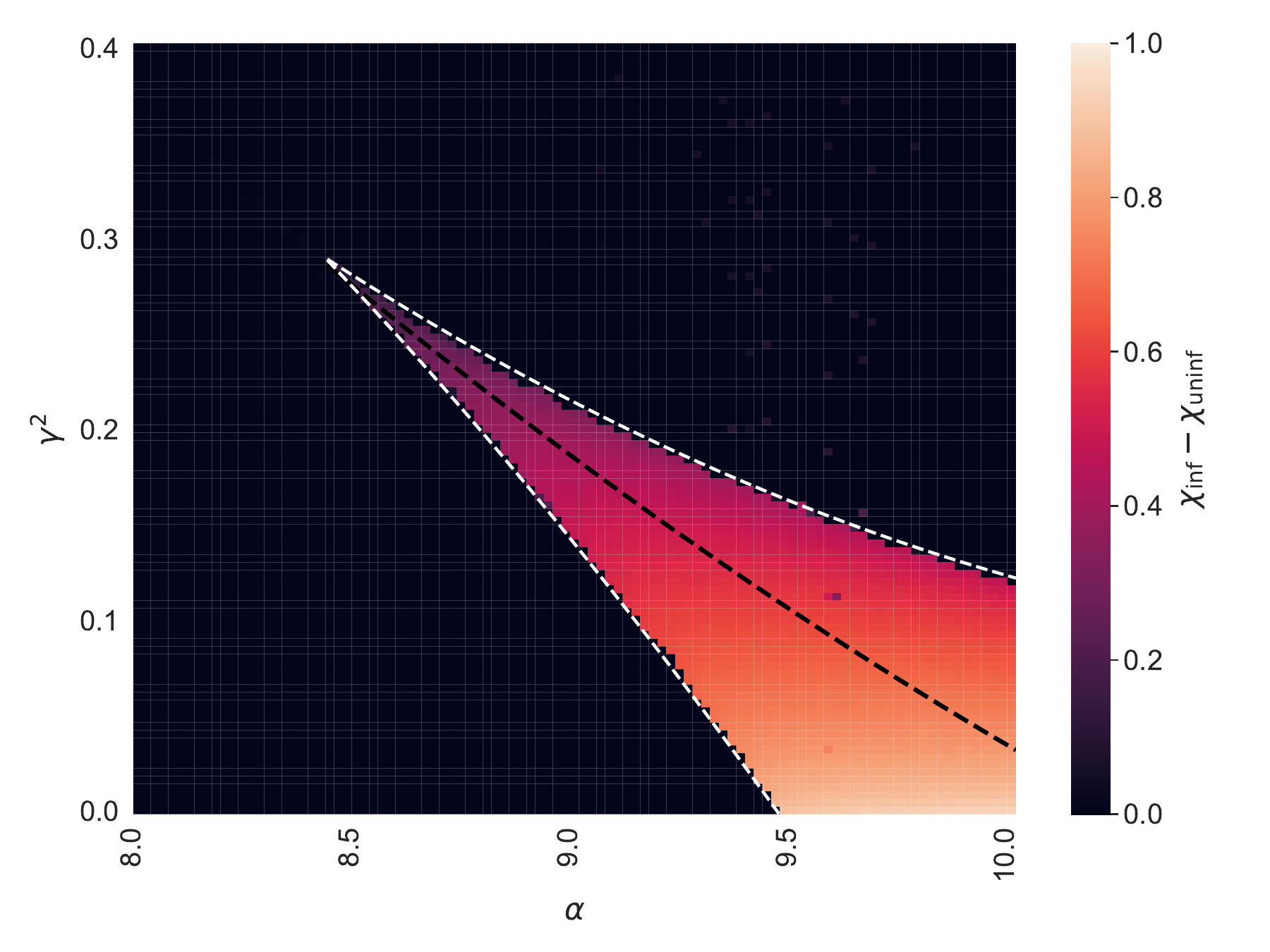}
    \includegraphics[width=0.45\textwidth]{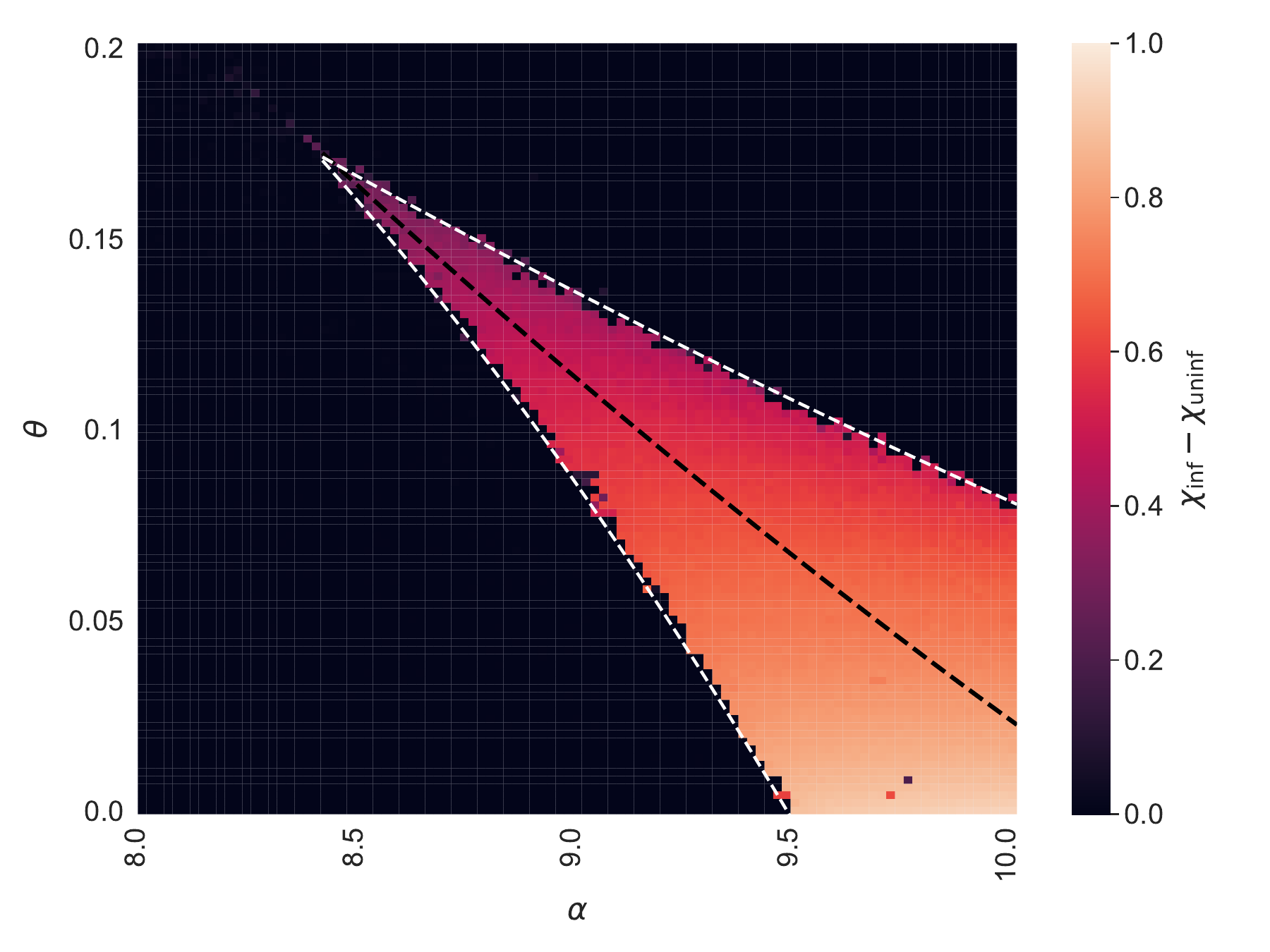}
    \caption{Phase diagrams for flow-based sampling (left) and autoregressive-based sampling (right) for the $k$-NAESAT model. On the x-axis we put the constraints to variables ratio $\alpha=M/N$ and on the y-axis the ratio $\gamma^2 = \alpha^2/\beta^2$ (left) and the decimated ratio $\theta$ (right). We compute the order parameter $\chi$, defined in~(\ref{eq:chi_NAESAT}), both from an uninformed and an informed initialization, and we plot the difference between the two. The dashed white lines are the \textit{spinodal lines}, while the dashed black one is the \textit{IT threshold}, both defined at the beginning of the section. In both plots we have that the dynamical transition is at $\alpha_d\approx 9.465$, the Kauzmann (or condensation) transition is at $\alpha_K \approx 10.3$, while the tri-critical points are at $\alpha_{\rm tri}\approx8.4$ for both flow-based and autoregressive based sampling. Thus, within our numerical precision, the two methods seems to be equally efficient.}
    \label{fig:NEASAT}
\end{figure}
Compared to the plots reported in the main text, here we display directly the difference $\C_{\rm inf} - \C_{\rm uninf}$, so that the coloured zones of the plots are the ones displaying multiple fixed points, as opposed to the black ones. We furthermore draw as white dashed lines the spinodal points, and as a black dashed line the IT threshold, both defined at the beginning of Appendix~\ref{sec:PhD-app}.

Compared to the previous models, defined on dense graphs, here we don't have some fixed point equations defined on an $O(1)$ number of scalar parameter, but instead some self-consistent equations on probability distributions. This clearly makes it harder to have a precise estimate of the tri-critical points, both for the tilted and the pinned measures. For the precision we were able to achieve, both for flow-based sampling and for autoregressive-based sampling the tri-critical point appears to be around $\alpha_{\rm tri} \approx 8.4$, and we are not able to state if one of the two methods has a smaller gap compared to the other and thus if there is a range in $\alpha$ where it is able to sample efficiently where the other can not. A more careful analysis is needed to clarify this point.

\section{Sampling simulations for Algorithm~\ref{Algo}}
We now show how the phenomena described through the previous phase diagram influences the performance of the sampling algorithm in practice.
\begin{figure}[t!]
    \centering
    \includegraphics[width=0.9\textwidth]{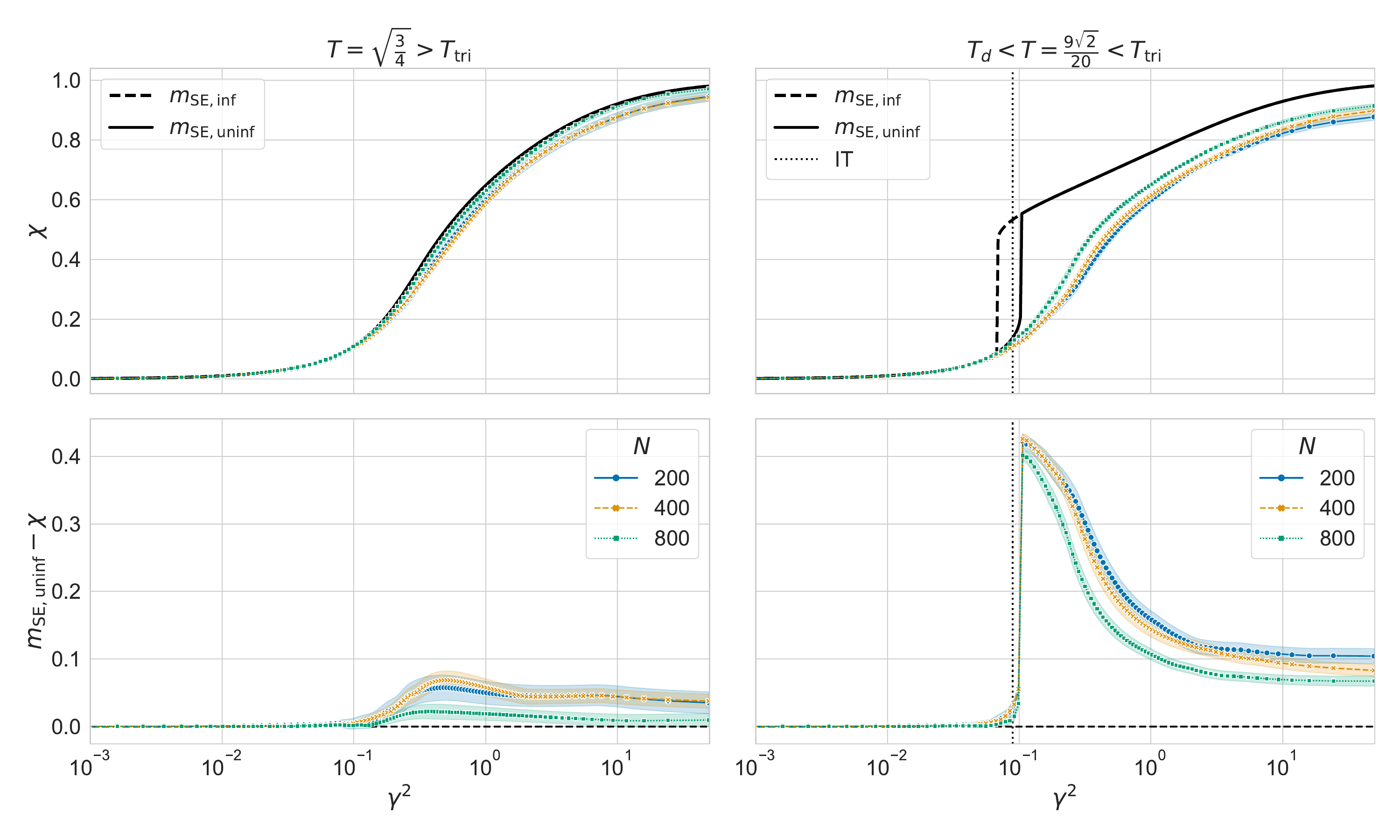}
    \caption{\textbf{Spherical p-spin: $\chi(\gamma)$ for flow-based sampling.} We compare the results for the order parameter $\chi$ computed from the State Evolution equations (black lines) to finite-size implementations of the sampling algorithm in two regimes. Left: $T = \sqrt{\frac{3}{4}}>\Delta_{\rm tri}$, for all values of $\gamma^2$, the SE equations have a unique fixed point, and thus the initialization plays no role. The resulting curve (black continuous line) is compared to algorithmic implementations of the flow-based sampling algorithm, for sizes $N=200,400,800$. These curves, shown in different colours above, match the asymptotic prediction. Right: $T_{d}<T=\frac{9\sqrt{2}}{20}<T_{\rm tri}$, there is a range of values of $\gamma^2$, for which the SE equations have two distinct fixed points. More precisely, for all values between the IT point (reported as a dotted line) and the informed spinodal point the model presents an algorithmically hard phase. The uninformed/informed state evolution curves (black continuous/dashed lines respectively) are compared to algorithmic implementations of the flow-based sampling algorithm, for sizes $N=200,400,800$. Importantly, we show that in this regime there is an evident mismatch with the asymptotic prediction. }
    \label{fig:SphT_N}
\end{figure}
In Fig.~\ref{fig:SphT_N} we report some numerical simulations where we compare the asymptotic curves derived through state evolution with some empirical simulations implementing the flow-based sampling for the spherical $p$-spin model.

Specifically, we compare what happens at a high value of temperature $T>T_{\rm tri}$, where we expect the sampling scheme to work, to a value of $T$ in the interval $[T_d,T_{\rm tri}]$, where, as previously explained, we predict the algorithm to fail.

In the first case, shown in the left part of the plot, we see that finite-size simulation follows very well the theoretical prediction, already at a considerably low number of variables. Conversely, in the second case, reported in the right part, the situation is different, and the curves show a gap.

Moreover, we check when the Nishimori conditions are satisfied for different values of temperature $T$ in the spherical $p$-spin. We do this by analysing the following observable:
\begin{eqnarray}\label{eq:overlap_alpha}
    \text{OV}(\gamma) \equiv \frac{1}{N}\mathbb{E} \left[ {\bf Y_t} \langle {\bf x} \rangle \right]
\end{eqnarray}
which is the overlap between the observation vector ${\bf Y_t}$ and the average magnetization $\langle {\bf x}\rangle$.

Indeed, by definition of the observation process we can write
\begin{eqnarray}
    \text{OV}(\gamma) = \frac{\alpha(t)}{N}\mathbb{E} \left[ {\bf x_0} \langle {\bf x} \rangle \right] + \frac{\beta(t)}{N}\mathbb{E} \left[ {\bf z} \langle {\bf x} \rangle \right]
\end{eqnarray}
and furthermore by using Stein's lemma \cite{stein1981estimation} we can write it as
\begin{eqnarray}\label{eq:OV_step1}
\begin{aligned}
    \text{OV}(\gamma) &= \frac{\alpha(t)}{N}\mathbb{E} \left[ {\bf x_0} \langle {\bf x} \rangle \right] + \frac{\beta(t)}{N}\mathbb{E} \left[  \partial_{\bf z}\langle {\bf x} \rangle \right] \\
    &= \frac{\alpha(t)}{N}\mathbb{E} \left[ {\bf x_0} \langle {\bf x} \rangle \right] + \frac{\beta(t)}{N}\mathbb{E} \left[  \frac{\alpha(t)}{\beta(t)}\langle \| {\bf x}\|_2^2 \rangle - \frac{\alpha(t)}{\beta(t)}\langle {\bf x} \rangle^2\right] \\ 
    &= \alpha(t)\frac{1}{N}\mathbb{E} \left[ {\bf x_0} \langle {\bf x}\rangle + \langle\| {\bf x}\|_2^2 \rangle - \langle {\bf x} \rangle^2 \right] \approx \alpha(t) + \alpha(t)\frac{1}{N}\mathbb{E} \left[ {\bf x_0} \langle {\bf x}\rangle - \langle {\bf x} \rangle^2 \right]\,.
\end{aligned}
\end{eqnarray}
where in the last step we used the fact that $\lim_{N\rightarrow \infty}\| {\bf x}\|_2^2/N = 1 $. Now, \textit{if} the Nishimori conditions are satisfied, we have that $\mathbb{E}\left[ {\bf x_0} \langle {\bf x}\rangle \right] = \mathbb{E}\left[ \langle {\bf x} \rangle^2 \right] $. Putting this equivalence back into~(\ref{eq:OV_step1}), we see that in this case $\text{OV}$ coincides with the function $\alpha(t)$ that defines the interpolant process.

In Fig.~\ref{fig:SphT_alpha}, we use this equivalence and compare with $\alpha(t)=1-t$, the behaviour of $\text{OV}$ for $T = \frac{9\sqrt{2}}{20} > T_{\rm tri}$ and for $T_{d} < T =\sqrt{\frac{3}{4}} < T_{\rm tri}$, showing that in the first case we are Bayes optimal for all the values of $\gamma$, and thus at each sampling step. Instead, in the second case, even if at the beginning the two curves coincide, around the IT threshold they develop a gap, and after this point they take two different paths. This behaviour is consistent with what we observed from the study of the order parameter $\chi$, for which also a gap with the state evolution equations was developed around the IT threshold.
\begin{figure}[t!]
    \centering
    \includegraphics[width=0.9\textwidth]{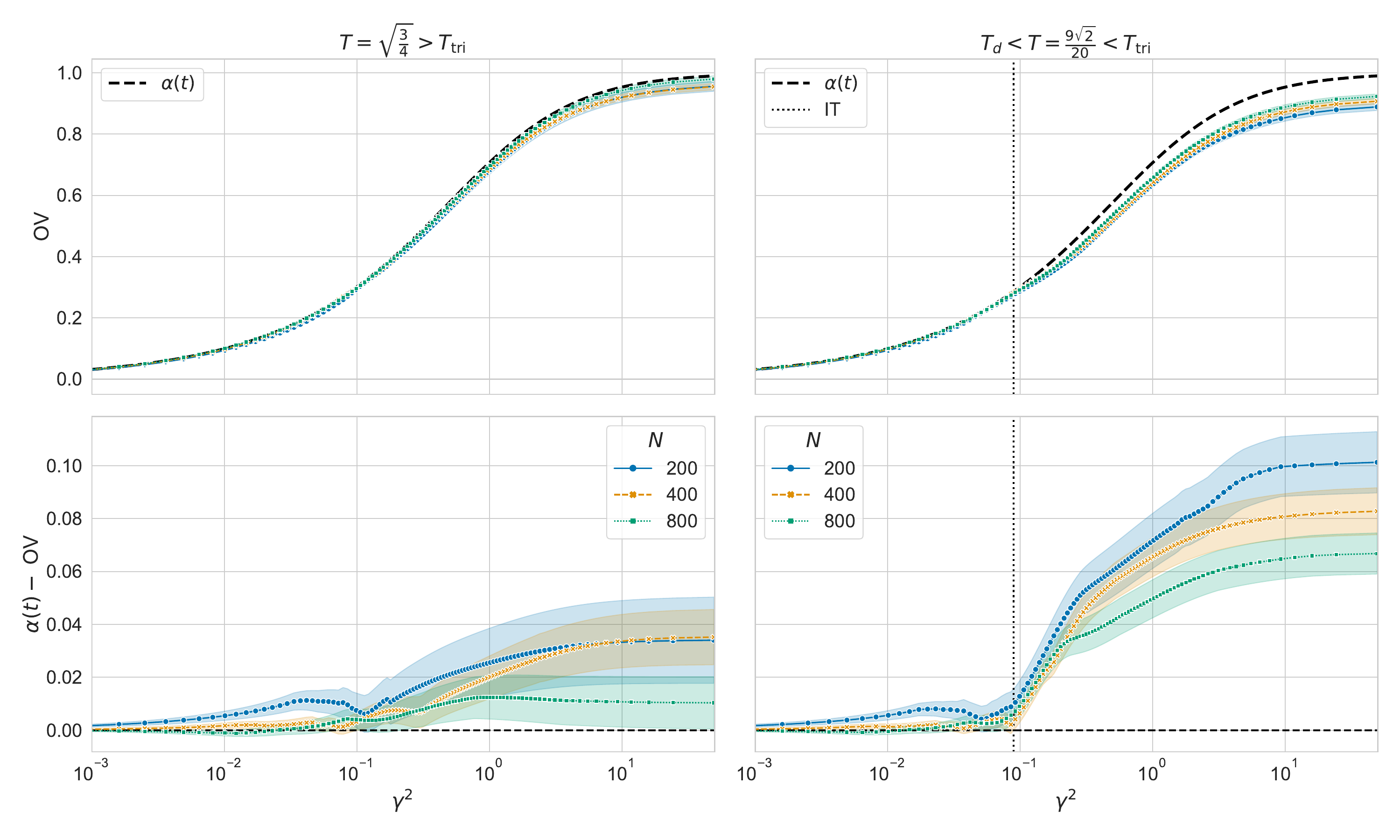}
    \caption{\textbf{Spherical p-spin: Checking the Nishimori conditions for flow-based sampling.} We compare the results for the overlap $\text{OV}$, defined in~(\ref{eq:overlap_alpha}) computed from finite-size $(N=200,400,800)$ implementations  of the sampling algorithm to the behaviour with $\gamma$ of the function $\alpha(t) = 1-t$ in two regimes. Left: $T = \sqrt{\frac{3}{4}}>\Delta_{\rm tri}$. The simulations, shown in different colours above, are very close to the asymptotic prediction, even if the sizes are relatively limited. Right: $T_{d}<T=\frac{9\sqrt{2}}{20}<T_{\rm tri}$. Here, after the IT threshold there is an evident mismatch between the two curves, due to the departure of the algorithm from the Bayes-optimality regime.
    }
    \label{fig:SphT_alpha}
\end{figure}

\section{Bayes optimal inference: Concentration, Replica Symmetry and Nishimori identities}

We shall briefly recall here some of the important properties of the measure associated with optimal Bayesian denoising, that we are using in the main text. All of these properties are well known in the literature, but we recall them for completeness

The first one are the well-known Nishimori identities \cite{nishimori2001statistical,iba1999nishimori,zdeborova2016statistical}, which are valid for any Bayesian posterior estimation problem where one observes a variable $Y$ sampled from $P(Y|X)$, and attempt to reconstruct $X$ by computing the posterior average. We reproduce the theorem and the proof here
\begin{theorem}[Nishimori Identity]
Let $ X^{(1)}, \ldots   , X^{(k)} $ be $ k $ i.i.d. samples (given $ Y $) from the distribution $ P \rdbrsv{ X = \cdot }{ Y } $. Denoting $ \agbrs{ \cdot } $ the ``Boltzmann" expectation, that is the average with respect to the $ P \rdbrsv{ X = \cdot }{ Y } $, and $ \mathbb{E}\sqbrs{ \cdot } $ the ``Disorder" expectation, that is with respect to $ (X^*,Y)$. Then for all continuous bounded function $ f $ we can switch one of the copies for $X^*$: 
\begin{equation}
    \mathbb{E} \sqbrs{ \agbrs{ f \rdbrs{ Y, X^{(1)}, \ldots, X^{(k-1)}, X^{(k)} } }_k } = \mathbb{E} \sqbrs{ \agbrs{ f \rdbrs{ Y, X^{(1)}, \ldots, X^{(k-1)}, X^* }}_{k-1} }
\end{equation}
\end{theorem}
\begin{proof}
    The proof is a consequence of Bayes theorem and of the fact that both $x^*$ and any of the copy $X^{(k)}$ are distributed from the posterior distribution. Denoting more explicitly the Boltzmann average over $k$ copies for any function $g$ as
  \begin{align}
 \agbrs{ g(X^{(1)},\ldots,X^{(k)}) }_{k}    
 \eqdef 
 \int \prod_{i=1}^k \dd{x}_i P(x_{i}|Y) g(X^{(1)},\ldots,X^{(k)})\,
 \end{align}  
    we have, starting from the right-hand side
\begin{align}
& \mathbb{E}_{Y,X^*} \sqbrs{ \agbrs{ f \rdbrs{ Y, X^{(1)}, \ldots, X^{(k-1)}, X^{*} } }_{k-1} }   \nonumber \\
&=  \int dx^* dy P(x^{*}|Y)P(Y)\agbrs{ f \rdbrs{ Y, X^{(1)}, \ldots, X^{(k-1)}, X^* } }_{k-1}  \nonumber \\
 &=  \mathbb{E}_Y \int \dd{x}^{k} P(x^{k}|y)  \agbrs{ f \rdbrs{ Y, X^{(1)}, \ldots, X^{(k-1)}, X^k } }_{k-1}  \nonumber \\
 &= \mathbb{E}_Y   \sqbrs{ \agbrs{ f \rdbrs{ Y, X^{(1)}, \ldots, X^{(k-1)}, X^{(k)} } }_k } \nonumber
 \end{align}
\end{proof}
In particular, we have the relation $\M=\C$, as stated in the main text.

We now move to the specific case of Gaussian denoising, and more specifically to the measure:
\begin{equation}
    P_{\gamma}({\bf x}) = \frac 1{Z_n} \exp \left( \gamma^2 \langle {\bf x_0},{\bf x} \rangle + \gamma\langle {\bf z},{\bf x} \rangle - \frac{\gamma^2}{2}\|{\bf x}\|^2\right)P_0({\bf x})\,
    \label{tilded-sup}
\end{equation}

The first identity is that the derivative of the free entropy associated with this problem is simply the expected overlap $\M$ (This is often called in slightly different context the I-MMSE theorem \cite{guo2013interplay}) and the second derivative the (Boltzmann) variance of the overlap (this is often called the Fluctuation-Dissipation theorem in statistical mechanics):
\begin{lemma}[First(I-MMSE theorem\cite{guo2005mutual}) and second (FDT theorem \cite{kubo1966fluctuation}) derivative of the free entropy] \label{lem:fdt_ov}
Consider the free entropy density associated with the measure~(\ref{tilded-sup}): 
\begin{equation}
f_N = \frac 1N {\mathbb E} [\log Z_n (\gamma)]
\end{equation}
then
\begin{equation}
\partial_{\gamma} f_n = \gamma \M(\gamma) \equiv \frac {\gamma}N \mathbb{E}[\langle {\bf x}(\gamma)\rangle {\bf x}_0]
\end{equation}
and
\begin{equation}
\partial^2_{\gamma} f_N = N\gamma^2
\mathbb{E}\left[
\left( 
\langle
(\frac { {\bf x}(\gamma)\cdot {\bf x}_0}{N})^2\rangle 
\right)
-
\left( 
\frac {\langle {\bf x}(\gamma)\rangle {\bf x}_0} N
\right)^2
\right] 
= N \gamma^2 {\mathbb E} [{\rm var}\left (\frac{{\bf x} \cdot {\bf x}_0}N \right)]
\end{equation}
\end{lemma}
\begin{proof}
The proof is a direct application of Nishimori identities together with Stein lemma (which states that
${\mathbb E}_Z zg(Z) = {\mathbb E} g'(z)$ for a Gaussian random variable $Z$), which we reproduce here:
\begin{eqnarray}
\partial_{\gamma} f_n &=& {\mathbb E} \int d{\bf x} P_{\gamma}({\bf x}) \frac {2\gamma {\bf x}_0 \cdot {\bf x} + {\bf x} \cdot {\bf z} - \gamma{\bf x} \cdot {\bf x} }{N {\bf z}(\gamma)} \\
&=& 2 \gamma \M - \gamma {\mathbb E} \langle \frac {{\bf x} \cdot {\bf x}}N \rangle- {\mathbb E}\int d{\bf x} P_{\gamma}(X) \frac {X \cdot Z}{N {\bf z}(\gamma)} \\
&=&2 \gamma \M - \gamma {\mathbb E} \langle \frac {{\bf x} \cdot {\bf x}}N \rangle +  \gamma {\mathbb E} \langle \frac {{\bf x} \cdot {\bf x}}N \rangle - \gamma \frac 1N  \langle {\bf x} \rangle \cdot \langle {\bf x} \rangle 
\label{weusedstein}\\
&=&2 \gamma \M - \C = \gamma \M 
\label{weusednish}
\ \end{eqnarray}
where we use Stein lemma on line~\ref{weusedstein} and Nishimori on line~\ref{weusednish}. The second derivative identity is obtained along the same line by deriving twice with respect to $\gamma$.
\end{proof}

This lemma in turn implies that the concentration of the overlap $ \M $, a trick often used in the literature of mathematical physics when proving the replica equation (see e.g. \cite{barbier2021overlap}):
\begin{theorem}[Concentration of overlaps]
Almost everywhere in $gamma$, we have for some $K(\gamma)$, that
$$
{\mathbb E} [{\rm var}\left (\frac{{\bf x} \cdot {\bf x}_0}N \right)] \to_{N \to \infty} 0 
$$
\end{theorem}
\begin{proof}
From the derivative, we have that 
\begin{eqnarray}
    \int_{\gamma_1}^{\gamma_2} {\rm var}_{\gamma} \left (\frac{{\bf x} \cdot {\bf x}_0}N \right) = \frac 1{\gamma^2 N}(\M(\gamma_2)-\M(\gamma_1)) \le \frac KN 
\end{eqnarray}
where we have assumed that $\M$ is bounded by some constant (which is the case for the discrete variables discussed in this paper where $-1<\M<1$). As a consequence, almost everywhere in $\gamma$, the variance of the overlap must vanish. 
\end{proof}
Additionally, one can also prove the concentration of the variance with respect to the disorder, see \cite{barbier2021overlap}.  

Finally, similar results exist in the case of ``pinning"
 \cite{abbe2013conditional,coja2017information}.

\section{\label{sec:AnalysisAlgo}Analysis of Algorithm~\ref{Algo}}

In this section, we provide a theoretical analysis of the performance of Algorithm~\ref{Algo} for the ``efficient" regime in Figure~\ref{fig:mainfig}. In these regimes, we conjecture being able to approximate perfect denoising during the interpolation path. Therefore, to simplify our analysis, we assume access to a perfect denoiser and primarily focus on the validity of the continuous-time limit and the effect of discretization. Additional approximation and discretization errors due to the denoiser, such as AMP can be incorporated in our analysis through a straightforward manner. For a rigorous analysis of the approximation error and Lipschitzness of the AMP iterates in the case of the SK and spiked matrix models, we refer the reader to \cite{el2022sampling,montanari2023posterior}.

Lastly, throughout the analysis, we shall assume that Algorithm~\ref{Algo} is run from time $t=1$ to $t=\epsilon$ for some $\epsilon > 0$. This ensures that we can avoid singularities at $t=0$, while choosing $\epsilon$ arbitrarily close to $0$ allows us to approximate the target measure up to arbitrary accuracy.

Recall the definition of the ``tilted" posterior measure at time $t$:
\begin{equation}
    P({\bf x}|{\bf y}(t)={\bf Y}_t)=\frac{1}{Z({\bf Y}_t)}\exp \left( \frac{\alpha(t)}{\beta(t)^2}\langle {\bf Y}_t,{\bf x} \rangle - \frac{\alpha(t)^2}{2\beta(t)^2}||{\bf x}||^2\right)P_0({\bf x}).
\end{equation}
Here the factor $\frac{1}{Z({\bf Y}_t)}\exp \left( \frac{\alpha(t)}{\beta(t)^2}\langle {\bf Y}_t,{\bf x} \rangle - \frac{\alpha(t)^2}{2\beta(t)^2}||{\bf x}||^2\right)$ is interpreted as the Radon-Nikodym derivative of $P({\bf x}|{\bf y}(t)$ w.r.t $P_0$. Throughout the present section, we shall denote $P({\bf x}|{\bf y}(t)={\bf Y}_t)$ by $P_{t,{\bf Y}_t}$. We shall assume that $\alpha(t),\beta(t)$ are continuously-differentiable.

\subsection{Fluctuation-dissipation and Lipschitzness}

A crucial quantity related to the well-posedness of the continuity equation (Eq.~(\ref{eq:transport})) and the validity of Algorithm~\ref{Algo} is the Lipschitz constant of the vector field. We therefore start by present a preliminary result, which can be interpreted as an instance of the Fluctuation-dissipation theorem:

\begin{lemma}\label{lem:fdt_vec}
    Suppose that the measure $P_0$ has compact support.
    Then the Jacobian of the vector $\frac{\partial b({\bf y},t)}{\partial {\bf y}}$ is related to the covariance of the tilted measure $P_{t,{\bf Y}_t}=P({\bf x}|{\bf y}(t)={\bf Y}_t)$ as follows:
    \begin{equation}
         \frac{\partial b({\bf y},t)}{\partial {\bf y}}\vert_{\bf y={\bf Y}_t} = \frac{\alpha(t)}{\beta(t)^2}(\dot{\alpha}(t)-\frac{\dot{\beta}(t)\alpha(t)}{\beta(t)})\operatorname{Covar}[P_{t,{\bf Y}_t}]+\frac{ \dot{\beta}(t)}{\beta(t)} 
    \end{equation}
\end{lemma}

\begin{proof}
We have, from Eq.~(\ref{eq:v_field}):
\begin{equation}\label{eq:v_field_pos}
\begin{split}
     b({\bf y},t) &= \mathbb{E}[\dot{\alpha}(t) {\bf x_0 }+ \dot{\beta}(t)z|{\bf y}(t)= {\bf y}]\\
     &=   \mathbb{E}[\dot{\alpha}(t){\bf x_0 }+ \frac{\dot{\beta}(t)}{\beta(t)}({\bf y}(t)-\alpha(t) {\bf x}_0)|{\bf y}(t)= {\bf y}]\\
     &=(\dot{\alpha}(t)-\frac{\dot{\beta}(t)\alpha(t)}{\beta(t)})\mathbb{E}[{\bf x_0 }|{\bf y}(t)= {\bf y}]+\frac{ \dot{\beta}(t)}{\beta(t)}{\bf y}.
\end{split}
\end{equation}
where we used the relation in Eq.~(\ref{eq:x(t)}).
Next, we have through the expression for the posterior measure $P_{t,{\bf Y}_t}$:
\begin{equation}
  \mathbb{E}[{\bf x_0 }|{\bf y}(t)= {\bf y}]=  \mathbb{E}_{P_0}\left[\frac{1}{Z({\bf y})}{\bf x}\exp \left( \frac{\alpha(t)}{\beta(t)^2}\langle {\bf y},{\bf x} \rangle - \frac{\alpha(t)^2}{2\beta(t)^2}||{\bf x}||^2\right)\right].
\end{equation}
Using the boundedness of the support of $P_0$ and dominated convergence theorem, we may differentiate the L.H.S inside the expectation. We obtain that for all $t \in (0,1]$, $\mathbb{E}[{\bf x_0 }|{\bf y}(t)= {\bf Y}_t]$ is differentiable w.r.t ${\bf Y}_t$ with the Jacobian given by $\frac{\alpha(t)}{\beta(t)^2}\operatorname{Covar}[P_{t,{\bf Y}_t}]$. Substituting into Eq.~(\ref{eq:v_field_pos}) completes the proof.

\end{proof}
\subsection{Derivation of the ODE (Eq.~(\ref{eq:ODE}))}
 We start by presenting an informal derivation of the continuity equation, borrowed from \citep{albergo2022building}. Recall that, by definition, ${\bf y}(t)=\alpha(t) {\bf x}_0 + \beta(t) {\bf z})$.
The density $\rho({\bf y},t)$ can be expressed as:
\begin{equation}\label{eq:delta}
    \rho({\bf y},t) = \int_{\R^N\times \R^N} \delta(\bf y-{\bf y}(t))  \rho_0({\bf x_0})\rho_\gamma({\bf z}) d{\bf x_0} d{\bf z},
\end{equation}
where $\rho_\gamma$ denotes the density of the standard Gaussian measure $\gamma=\mathcal{N}({\bf 0},\mathbb{I}_N)$. Differentiating both sides of Equation~(\ref{eq:delta}) yields:
\begin{equation}
\begin{split}
\partial_t\rho({\bf y},t)&= -\int_{\R^N \times \R^N} \nabla\delta(\bf y-{\bf y}(t))\cdot \partial_t {\bf y}(t) \rho_0({\bf x_0})\rho_\gamma({\bf z}) d{\bf x_0} d{\bf z}\\
    &=-\nabla\cdot \int_{\R^N \times \R^N} \delta(\bf y-{\bf y}(t))\partial_t {\bf y}(t) \rho_0({\bf x_0})\rho_\gamma({\bf z}) d{\bf x_0} d{\bf z}\\
    &=-\nabla\cdot \left(\int_{\R^N \times \R^N} \left(\delta({\bf y}-{\bf y}(t))\partial_t {\bf y}(t) \frac{\rho_0({\bf x_0})\rho_\gamma({\bf z})}{\rho({\bf y},t)}d{\bf x_0} d{\bf z}\right) \rho({\bf y},t)\right)
\end{split}    
\end{equation}
 Eq.~(\ref{eq:transport}) is then obtained by noticing that $\int_{\R^d \times \R^d} \left(\delta(\bf y-{\bf y}(t))\partial_t {\bf y}(t) \frac{\rho_0({\bf x_0})\rho_\gamma({\bf z})}{\rho({\bf y},t)}d{\bf x_0} d{\bf z}\right)$ equals $b({\bf y},t)$ defined by~(\ref{eq:v_field}).

We refer to \cite{albergo2022building,albergo2023stochastic}
for the complete derivation based on the above approach.

We next prove that the pushforward measure obtained after applying the flow defined by the ODE~(\ref{eq:ODE}) to initial Gaussian noise ${\bf z}$ results in a sample from the target measure $P_0$. For the sake of completeness, we include an alternative derivation for Eq.~(\ref{eq:transport}) in our proof, based on the weak form of the continuity equation. 
Our result allows $P_0$ to be a discrete measure, by restricting the time to $t \in (0,1)$.

\begin{lemma}\label{lem:cont}
Let $\epsilon$ be an arbitrarily-fixed real in $(0,1)$. Let $\bf{Y}_t({\bf z})$ denote the flow associated with the ODE in Eq.~(\ref{eq:ODE}) starting from $t=1$ to $t=\epsilon$ with ${\bf z} \sim \gamma$, where $\gamma$ denotes the standard Gaussian measure $\gamma=\mathcal{N}({\bf 0},\mathbb{I}_N)$. Suppose that $P_0$ has bounded support.
Then, the pushforward measure $\bf{Y}_{t_\#} \gamma$ at any time $t \in (0,1]$ equals the measure corresponding to the law $P_t$ of the interpolant ${\bf y}(t)$ defined by equation~(\ref{eq:x(t)}).
\end{lemma}

\begin{proof}
Let $\psi \in C^\infty_c(\R^d)$ be an arbitrary test function. Using the change of variables formula, the expectation of $\psi$ w.r.t the measure $\mu_t$ at time $t$ can be expressed as:
\begin{equation}
    \int_{\R^N} \psi({\bf y}) dP_t({\bf y}) = \int_{\R^N\times \R^N}\psi({\bf y}(t)) dP_0({\bf x}_0) d\gamma_{z}({\bf z}),
\end{equation}
where ${\bf y}(t)$ is defined as a measurable function of ${\bf x}_0,{\bf z}$ through Eq.~(\ref{eq:x(t)}).

Therefore, $P_t$ evolves in a distributional sense as follows:
\begin{equation}\label{eq:weak}
    \int_{\R^N} \psi({\bf y}) \partial_t dP_t({\bf y})=\frac{d\int_{\R^N} \psi({\bf y}) dP_t({\bf y})}{dt}= 
    \int_{\R^N\times \R^N} \nabla \psi({\bf y}(t))\cdot \partial_t {\bf y}(t) dP_0({\bf x}_0) d\gamma_{z}({\bf z}).
\end{equation}
Recall that:
\begin{equation}
    b({\bf y},t) = \E[\partial_t {\bf y}(t)|{\bf y}(t)={\bf y}].
\end{equation}
Using the definition of the conditional expectation, and the change of variables formula, we have:
\begin{align*}
     \int_{\R^N\times \R^N} \nabla \psi({\bf y}(t))\cdot \partial_t {\bf y}(t) dP_0({\bf x}_0) d\gamma_{z}({\bf z}) &=  \int_{\R^N} \nabla \psi({\bf y}).\E[\partial_t {\bf y}(t)|{\bf y}(t)={\bf y}] dP_t({\bf y})\\
    &=\int_{\R^N}\nabla \psi({\bf y})\cdot b({\bf y},t) dP_t({\bf y})\\
    &=- \int_{\R^N} \psi({\bf y}) \nabla.(b({\bf y},t) P_t({\bf y})) d{\bf y}, 
\end{align*}
where we used the compactness of the support of $\psi(x)$ and the distributional definition of the divergence operator $\nabla$.
Substituting in Eq.~(\ref{eq:weak}), we obtain:
\begin{equation}
    \int_{\R^N} \psi({\bf y}) \partial_t dP_t({\bf y})=- \int_{\R^N} \psi({\bf y}) \nabla.(b({\bf y},t) P_t({\bf y})) d{\bf y}.
\end{equation}
Furthermore, for all $t > \epsilon$ for fixed $\epsilon > 0$, we have $\beta(t) > 0$.
Then, Lemma~\ref{lem:fdt_vec} implies that 
$b({\bf y},t)$ is Lipschitz w.r.t ${\bf y}$. Thus, the probability flow $\bf{Y}_t$ for $t \in [\epsilon,1]$ exists and is unique.
Since both the push-forward measure $\bf{Y}_t\# \gamma$ and the law of ${\bf y}(t)$ i.e. $P_t$ satisfy the continuity equation with velocity field $b({\bf y},t)$ (Lemma 4.1.1. in \cite{figalli2021invitation}).
 By the uniqueness of the solution of the continuity equation (see e.g.~\citep{ambrosio2008transport}), $b(t,x)_\#\mu(z)=\mu_t$ . 
\end{proof}

\subsection{Sampling Guarantees}

In this section, we quantify the effect of discretization errors in the velocity field, leading to sampling guarantees for Algorithm~\ref{Algo}. Again, we fix a parameter $\epsilon$ lying in $(0,1)$
We rely on the following assumptions:

\begin{assumption}\label{As:bsup}
    The measure $P_0$ has compact support lying in a sphere with radius bounded as $\mathcal{O}(\sqrt{N})$
    \begin{equation}
        \operatorname{Supp}(P_0) \subseteq \mathbb{S}^{N-1}(\sqrt{N}R),
    \end{equation}
    for some $R$ independent of $N$.
\end{assumption}

\begin{assumption}\label{As:Lipshcy}
The spectral norm of the vector field's Jacobian $\frac{\partial b({\bf Y}_t ,t)}{\partial {\bf Y}_t}$ w.r.t ${\bf Y}$ is uniformly bounded, i.e.
  $\norm{\frac{\partial b({\bf Y} ,t)}{\partial {\bf Y}}}_2 \leq  L, \forall t \in (\epsilon,1), \forall {\bf Y} \in \R^d$ for some $L \geq 0$.
\end{assumption}

\begin{assumption}\label{As:Lipshct}
$\norm{\frac{\partial b({\bf Y} ,t)}{\partial t}}$ is uniformly bounded by $M\sqrt{N}$ in ${\bf Y}, t$ for $t \in (\epsilon,1)$, for some $M \geq 0$. 
\end{assumption}

We have the following error bounds for the discretization error associated to the flow:

\begin{lemma}\label{lem:disc_err}
    Consider the ODE defined by Eq.~(\ref{eq:ODE}) i.e $\frac{d {\bf Y}}{dt}=b({\bf Y} ,t)$, with ${\bf Y} \in \R^N$. Let ${\bf Y}_t({\bf z})$ denote the flow associated to the above ODE at time $t \in (0,1]$ starting from some fixed ${\bf z} \in \R^N$ at time $t=1$.
    Let ${\bf Y}_{\delta,t}({\bf z})$ denote the iterates of the forward Euler method applied to the above ODE (in reverse time) starting from the same initialization ${\bf z}$ with step-size $\delta$ i.e, for $i \in \N$:
    \begin{equation}
        {\bf Y}_{\delta,\delta (i-1)}({\bf z}) =  {\bf Y}_{\delta,\delta i}({\bf z})-
    \delta b({\bf Y}_{\delta ,\delta i}({\bf z}),\delta i),
    \end{equation}
    with ${\bf z}$ fixed.
    Under Assumptions~\ref{As:bsup},\ref{As:Lipshcy},\ref{As:Lipshct}, there exists a constant $A(\epsilon)$ such that,
    for all $k \in \N, {\bf z} \in \R^N$ with $k\delta \geq \epsilon$:
\begin{equation}
\norm{{\bf Y}_{\delta,k\delta}({\bf z}) -{\bf Y}_{k\delta}({\bf z})}_2 \leq (\frac{M+AL}{L})\sqrt{N} e^{L k \delta}\delta,
\end{equation}
\end{lemma}
for small enough $\delta$.

\begin{proof}
We first note that Assumption~\ref{As:bsup} and Lemma~\ref{lem:fdt_vec} imply that $b({\bf Y} ,t)$ is uniformly bounded for any fixed ${\bf z}$.Then, the above bound follows from standard analysis of the forward Euler method. See for example Chapter 7 in \cite{hairerode}.    
\end{proof}

\textbf{Remark}: The Lipschitzness of the posterior mean (optimal denoiser) w.r.t ${\bf Y}_t$ is expected to hold for the setups considered in our work in light of similar results proven in \citep{el2022sampling,el2022information}. Furthermore, as Lemma~\ref{lem:fdt_vec} shows, it is equivalent to the boundedness of the covariance of the tilted measure. Similarly, Lemma~\ref{lem:fdt_ov}
provides control over $\frac{\partial b({\bf Y} ,t)}{\partial t}$ through the variance of the overlap.

We now prove that the proposed algorithm with sufficiently small step for the discretized ODE produces a distribution close to the target distribution in Wasserstein distance.
The obtained bounds on the error between the ODE and the algorithm iterates can be related to the Wasserstein distance between the corresponding pushforward measures.

\begin{theorem}
Let $P_{\rm alg,N_{\rm steps}}$ denote the measure of the output produced by Algorithm~\ref{Algo} after $N_{\rm steps}$. Suppose that the vector field, for any $\eta > 0$, there exist $N_{\rm steps}$, independent of $N$, such that the normalized Wasserstein distance $\frac{1}{\sqrt{N}}W_2(P_{\rm alg,N_{\rm steps}},P_0) < \eta$.
\end{theorem}

\begin{proof}
Let $0<\epsilon<h$ be arbitrary.
By definition, ${\bf Y}_{h,\epsilon}$ has law $P_{\rm alg,N_{\rm steps}}$ while, from Lemma~\ref{lem:cont}, ${\bf Y}_{\eps}$ equals in law ${\bf y}(\epsilon)$. Setting the same initialization ${\bf z}$ induces a coupling between the two measures. Recall that the 2-Wasserstein distance \citep{figalli2021invitation,villani2009optimal} between two measures $\mu,\nu$ on $\R^N$ is defined as $W^2_2(\mu,\nu)=\operatorname{inf}_{\Gamma(\mu,\nu)} \mathbb{E}\norm{{\bf X}-{\bf Y}}^2$, where $\Gamma(\mu,\nu)$ denotes the set of couplings between $\mu,\nu$  with ${\bf X},{\bf Y}$ being distributed as $\mu,\nu$ respectively.
Therefore, we obtain:
\begin{equation}
  W^2_2(P_{\rm alg,N_{\rm steps}},P_{\epsilon}) \leq \mathbb{E}\norm{{\bf Y}_{\delta,\epsilon}({\bf z})-{\bf Y}_{\epsilon}({\bf z})}^2.
\end{equation}

 Lemma~\ref{lem:disc_err} then implies that by choosing a small enough step-size $\delta$, or equivalently, with a large enough $N_{\rm steps}$, $W_2(P_{\rm alg,N_{\rm steps}},P_{\epsilon})$ can be made arbitrarily small for any fixed epsilon.
 
 We further have that $W_2(P_0,P_{\epsilon}) \rightarrow 0$ as $\epsilon \rightarrow 0$. By triangle inequality for $W_2$ \cite{figalli2021invitation,villani2009optimal}, we obtain:
 \begin{equation}
     W_2(P_0,P_{\epsilon}) \leq W_2(P_{\rm alg,N_{\rm steps}},P_{\epsilon}) + W_2(P_0,P_{\epsilon}).
 \end{equation}
Now, first pick $\epsilon > 0$ such that $\frac{1}{\sqrt{N}}W_2(P_0,P_{\epsilon}) \leq \eta/2$. Subsequently, we pick $N_{\rm steps}$ such that $\frac{1}{\sqrt{N}}W_2(P_{\rm alg,N_{\rm steps}},P_{\epsilon}) \leq \eta/2$ to complete the proof.
\end{proof}
The above result establishes that Algorithm~\ref{Algo} samples from a distribution approximating the target distribution up to any desired accuracy, with a finite number of forward Euler steps ($N_{\rm steps}$), independent of the dimension $N$.

\textbf{Remark:} In the presence of a first-order phase transition during the interpolation path $\norm{\frac{\partial b({\bf Y} ,t)}{\partial t}}$ may grow as $\omega(\sqrt{N})$. This is apparent from Lemma~\ref{lem:fdt_ov} which relates $\frac{\partial b({\bf Y} ,t)}{\partial t}$ to the variance of the overlap.
This corresponds to $M$ in Lemma~\ref{lem:disc_err} growing with $N$. Lemma~\ref{lem:disc_err} reveals that as long as this growth is polynomial in $N$, the discretization error can still be controlled by choosing $h$ to be $1/\text{poly}(N)$. This would lead only to an additional polynomial time factor in the time complexity of the algorithm. Therefore, with an oracle access to the denoiser, Algorithm~\ref{Algo} might be able to efficiently sample even in the ``inefficient" phase. We believe this to be an interesting question for future research.

\bibliographystyle{unsrt}
\bibliography{refs}

\end{document}